\documentclass{article}

\usepackage{lmodern} %
\usepackage{microtype} %

\usepackage[utf8]{inputenc}

\usepackage[tbtags]{amsmath}
\usepackage{amsthm,amssymb}
\usepackage{float}
\usepackage{enumitem}

\usepackage[%
	backend=biber,
	style=alphabetic,
	sorting=none,
	url=false,
	isbn=true,
	maxbibnames=10,
]{biblatex}
\AtEveryBibitem{%
 \clearlist{address}
 \clearlist{location}
 
 \ifentrytype{book}{}{%
  \clearlist{publisher}
  \clearname{editor}
 }
}

\usepackage{color}
\usepackage{xcolor}
\definecolor{dark-red}{rgb}{0.4,0.15,0.15}
\definecolor{dark-blue}{rgb}{0.15,0.15,0.4}
\definecolor{medium-blue}{rgb}{0,0,0.5}

\definecolor{mycomment}{rgb}{0.3,0.7,0.8}
\definecolor{mygray}{rgb}{0.5,0.5,0.5}
\definecolor{lightgray}{rgb}{0.95,0.95,0.95}
\definecolor{mymauve}{rgb}{0.58,0,0.82}

\usepackage{hyperref}
\hypersetup{%
	unicode=true,
	breaklinks=true,
	colorlinks,
	linkcolor={dark-blue},
	citecolor={dark-blue},
	urlcolor={medium-blue},
	pdfpagemode=UseOutlines,
	pdfborder={0 0 0}
}

\usepackage{adjustbox} %

\usepackage{pgfplots}
\pgfplotsset{compat=1.18}
\usepackage{tikz}
\usetikzlibrary{quantikz2}

\usepackage{graphicx}
\graphicspath{{img/}}

\usepackage{ytableau}
\usepackage{mathdots}
\usepackage{scalerel}

\usepackage{booktabs}
\usepackage{subcaption}

\usepackage[notion,quotation,electronic]{knowledge}
\knowledgeconfigure{protect quotation={quantikz}}

\usepackage{mathtools}
\mathtoolsset{centercolon}
\DeclareMathOperator{\C}{\mathbb{C}} %
\DeclareMathOperator{\polyOp}{poly}

\DeclareMathOperator{\Ex}{\mathbf{E}}
\DeclareMathOperator{\Prob}{P}
\DeclareMathOperator{\tr}{tr}
\DeclareMathOperator{\fidelityOp}{\kl[fidelity]{F}}
\DeclareMathOperator{\spanOp}{span}
\DeclareMathOperator{\RealOp}{Re}
\DeclareMathOperator{\entropyOp}{\kl[entropy]{S}}

\DeclareMathOperator{\informationOp}{\kl[\informationOp]{I_H}}
\knowledge \informationOp {notion}

\DeclareMathOperator*{\selectOp}{select}
\DeclareMathOperator{\idm}{\mathbf{1}}
\DeclareMathOperator{\swaprot}{\kl[\swaprot]{W}}
\DeclareMathOperator{\Prep}{prep}
\knowledge \swaprot {notion}
\DeclareMathOperator{\sym}{\kl[\sym]{Sym}}
\knowledge \sym {notion}
\knowledgenewrobustcmd{\decoder}[0]{\cmdkl{D}}
\knowledgenewrobustcmd{\apxRefCh}{\cmdkl{\mathcal E}}

\newcommand{\rotChannel}{\ensuremath{\kl[rotation channel]{\mathcal R}}}
\newcommand{\refChannel}{\ensuremath{\kl[reflection channel]{\mathcal R}}}

\DeclarePairedDelimiter{\of}{\lparen}{\rparen}

\DeclarePairedDelimiter\ceil{\lceil}{\rceil}

\newcommand{\ct}{^{\dagger}}
\newcommand{\tp}{^{\mathsf{T}}}
\newcommand{\xp}[1]{^{\otimes #1}}

\newcommand{\defeq}{\coloneqq} %

\DeclarePairedDelimiter\paren{(}{)}

\DeclarePairedDelimiter\abs{\lvert}{\rvert}
\DeclarePairedDelimiter{\norm}{\lVert}{\rVert}
\DeclarePairedDelimiterX\brakett[3]{\langle}{\rangle}%
{#1\,\delimsize\vert\,\mathopen{}#2\,\delimsize\vert\,\mathopen{}#3}
\DeclarePairedDelimiterX\ketbra[2]{\vert}{\vert}{#1 {\delimsize\rangle\langle} #2}
\DeclarePairedDelimiterX{\inp}[2]{\langle}{\rangle}{#1, #2} %
\DeclarePairedDelimiterXPP\bigo[1]{O}{(}{)}{}{#1}
\DeclarePairedDelimiterXPP\littleo[1]{o}{(}{)}{}{#1}
\DeclarePairedDelimiterXPP\bigomega[1]{\Omega}{(}{)}{}{#1}
\DeclarePairedDelimiterXPP\bigtheta[1]{\Theta}{(}{)}{}{#1}
\DeclarePairedDelimiterXPP\poly[1]{\polyOp}{(}{)}{}{#1}
\DeclarePairedDelimiterXPP\probability[1]{\Prob}{[}{]}{}{#1}
\DeclarePairedDelimiterXPP\probabilityq[2]{\underset{#1}{\Prob}}{[}{]}{}{#2}
\DeclarePairedDelimiterXPP\expectation[1]{\Ex}{[}{]}{}{#1}
\DeclarePairedDelimiterXPP\expectationq[2]{\underset{#1}{\Ex}}{[}{]}{}{#2}
\DeclarePairedDelimiterXPP\entropy[1]{\entropyOp}{(}{)}{}{#1}
\DeclarePairedDelimiterXPP\information[1]{\informationOp}{(}{)}{}{#1}
\DeclarePairedDelimiterXPP\trace[1]{\tr}{[}{]}{}{#1}
\DeclarePairedDelimiterXPP\ptrace[2]{\tr_{#1}}{[}{]}{}{#2}
\DeclarePairedDelimiterXPP\spanSet[1]{\spanOp}{\{}{\}}{}{#1}
\DeclarePairedDelimiterXPP\fidelity[2]{\fidelityOp}{(}{)}{}{#1, #2}
\DeclarePairedDelimiterXPP\fidelitywc[2]{\fidelityWCOp}{(}{)}{}{#1, #2}
\DeclarePairedDelimiterXPP\Real[1]{\RealOp}{\{}{\}}{}{#1}
\DeclarePairedDelimiterXPP\select[1]{\selectOp}{(}{)}{}{#1}
\newcommand{\reflection}[1]{\ensuremath{e^{i\pi #1}}}

\makeatletter
\newcommand{\opnorm}{\@ifstar\@opnorms\@opnorm}
\newcommand{\@opnorms}[1]{%
	\left|\mkern-1.5mu\left|\mkern-1.5mu\left|
   #1
	\right|\mkern-1.5mu\right|\mkern-1.5mu\right|
}
\newcommand{\@opnorm}[2][]{%
	\withkl{\kl[diamond norm]}{
		\cmdkl{\mathopen{#1|\mkern-1.5mu#1|\mkern-1.5mu#1|}}
		#2
		\cmdkl{\mathclose{#1|\mkern-1.5mu#1|\mkern-1.5mu#1|}}
	}
}
\makeatother

\providecommand\given{}
\newcommand\SetSymbol[1][]{%
\nonscript\:#1\vert
\allowbreak
\nonscript\:
\mathopen{}}

\DeclarePairedDelimiterX\set[1]\{\}{%
\renewcommand\given{\SetSymbol[\delimsize]}
#1
}

\DeclarePairedDelimiterXPP\gammaTwo[1]{\gamma_2}{(}{)}{}{%
\renewcommand{\given}{\SetSymbol[\delimsize]}
#1
}

\newcommand{\apxRef}{\kl[approximate reflection]{R}}
\newcommand{\cyclicPermutation}{\kl[cyclic permutation]{C}}
\newcommand{\unifs}{\kl[uniform superposition]{s}}

\usepackage[nameinlink,capitalise]{cleveref}
\crefname{figure}{Figure}{Figures}
\crefname{equation}{}{} %
\Crefname{equation}{Eq.}{Eqs.} %

\crefname{conjecture}{Conjecture}{Conjectures}
\Crefname{conjecture}{Conjecture}{Conjectures}

\newtheorem{theorem}{Theorem}[section]
\newtheorem{lemma}[theorem]{Lemma}
\newtheorem{fact}[theorem]{Fact}
\newtheorem{corollary}[theorem]{Corollary}
\newtheorem{conjecture}[theorem]{Conjecture}
\theoremstyle{definition}

\newcommand{\Airreps}[2]{\widehat{\mathcal{A}}^{#1}_{#2}}
\newcommand{\A}[2]{\mathcal{A}^{#1}_{#2}}
\newcommand{\GT}[1]{\mathrm{GT}(#1)}
\newcommand{\CG}{\mathrm{CG}}

\usepackage{todonotes}

\knowledge{notion}
	| fidelity
	| Fidelity

\knowledge{notion}
  | worst-case fidelity

\knowledge{notion}
	| infidelity
	| Infidelity

\knowledge{notion}
	| quantum phase estimation
	| QPE

\knowledge{notion}
	| diamond norm

\knowledge{notion}
	| diamond distance

\knowledge{notion}
	| trace norm

\knowledge{notion}
  | reflection

\knowledge{notion}
  | approximate reflection

\knowledge \rotChannel[rotation channel]{notion}
  
\knowledge{notion}
  | reflection channel

\knowledge{notion}
  | approximate reflection channel
  | approximate reflection channels

\knowledge{notion}
  | program register
  | program
  | program registers

\knowledge{notion}
  | cyclic permutation
  | cyclic permutation operator

\knowledge{notion}
  | uniform superposition

\knowledge{notion}
  | Holevo information
  | information

\knowledge{notion}
  | entropy
  | von Neumann entropy

\knowledge{notion}
  | isotypic unitary reflection channel

\knowledge{notion}
  | maximally entangled state

\knowledge{notion}
  | covariant

\knowledge{notion}
  | unitary reflection
  | unitary reflections

\knowledge{notion}
  | program decoder

\usepackage{authblk}

\title{Quantum Programmable Reflections}
\author[1,3]{Eddie Schoute\thanks{\href{mailto:eddie.schoute@ibm.com}{eddie.schoute@ibm.com}}}
\author[2,4,5,6]{Dmitry Grinko}
\author[1]{Yi\u{g}it Suba\c{s}{\i}\thanks{\href{mailto:ysubasi@lanl.gov}{ysubasi@lanl.gov}}}
\author[2]{Tyler Volkoff}
\affil[1]{Computer, Computational and Statistical Sciences Division, Los Alamos National Laboratory, Los Alamos, NM, USA}
\affil[2]{Theoretical Division, Los Alamos National Laboratory, Los Alamos, NM, USA}
\affil[3]{IBM Research, Cambridge, MA, USA}
\affil[4]{QuSoft, Amsterdam, NL}
\affil[5]{Institute for Logic, Language and Computation, University of Amsterdam, Amsterdam, NL}
\affil[6]{Korteweg-de Vries Institute for Mathematics, University of Amsterdam, Amsterdam, NL}

\begin{document}

\date{}

\maketitle

\begin{abstract}\noindent%
    Similar to a classical processor,
    which is an algorithm for reading a program and executing its instructions on input data,
    a universal programmable quantum processor is a fixed quantum channel that reads a quantum program
    $\ket{\psi_{U}}$
    that causes the processor to approximately apply an arbitrary unitary $U$ to a quantum data register.
    The present work focuses on a class of simple programmable quantum processors for implementing reflection operators,
    i.e.\ $U = e^{i \pi \ketbra{\psi}{\psi}}$
    for an arbitrary pure state $\ket\psi$ of finite dimension $d$.
    Unlike quantum programs that assume query access to $U$, our program takes the form of independent copies of the state to be reflected about $\ket{\psi_U} = \ket{\psi}^{\otimes n}$.
    We then identify the worst-case optimal algorithm among all processors of the form
    $\ptrace{\text{Program}}{V \paren{\ketbra{\phi}{\phi} \otimes \paren{\ketbra{\psi}{\psi}}^{\otimes n}} V^\dagger}$
    where the algorithm $V$ is a unitary linear combination of permutations.
    By generalizing these algorithms to processors for arbitrary-angle rotations, $e^{i \alpha \ketbra{\psi}{\psi}}$
    for $\alpha \in \mathbb R$, we give a construction for a universal programmable processor with better scaling in $d$.
    For programming reflections,
    we obtain a tight analytical lower bound on the program dimension
    by bounding the Holevo information of an ensemble of reflections applied to an entangled probe state.
    The lower bound makes use of a block decomposition of the uniform ensemble of reflected states with respect to irreps of the partially transposed permutation matrix algebra,
    and two representation-theoretic conjectures based on extensive numerical evidence.
    Finally, we observe that programming the submanifold of reflections obeys a conjecture by \textcite{Yang2020}
    with the dependence on the dimension being much weaker compared to universal programming.
\end{abstract}
\clearpage%
\tableofcontents

\section{Introduction}
Classical processors receive their instructions from a classical program
and a common paradigm is to instrument quantum processors by classical programs as well.
A natural question is then whether it is possible to engineer a fixed quantum channel
that reads a \emph{quantum} program in order to apply any unitary evolution,
a \emph{universal (quantum) processor}.
Unfortunately, a universal processor can only be approximately constructed~\cite{Nielsen1997}
with a known minimum program dimension~\cite{Yang2020}.
In the programmable metrology algorithm utilized in \cite{Yang2020}
and in other general algorithmic frameworks for programmable quantum processors~\cite{Nielsen1997,PhysRevLett.88.047905,ishi,Hillery2006,,Haah2023},
the program state is prepared by assuming query access to the unitary $U$ and applying these queries to an optimally chosen probe state.
But, when programming just a submanifold of unitaries,
it may be relevant to consider program states that are not obtained by query access to the unitary.
One alternative approach is to assume an information-theoretical encoding of the programmed operation as in the $\epsilon$-net approach to programming compact submanifolds of the set of unitary operations~\cite{PhysRevLett.122.080505, Gschwendtner2021}, although this approach does not give a tight upper bound for the universal case.

In the present work, we consider programming the submanifold of reflections about a state $\ket\psi$ of a finite dimension $d \ge 2$.
This submanifold is homeomorphic to the set of pure states under the map $\psi \mapsto R_{\psi}:= \idm - 2\psi$,
where we denote $\ketbra{\psi}{\psi}$ as $\psi$.
Our quantum program for approximately implementing a reflection does not utilize queries to the reflection operator,
but is given by copies of the state $\ket\psi$ encoding the reflection axis. In this approach, the program complexity is quantified by the number of copies $n$ or, more generally, by the dimension of the symmetric subspace encoding information about $\ket{\psi}$.
Thus, we wish to find a processor $\mathcal C$ independent of $\ket\psi$
such that $\mathcal E_\psi(\phi) \coloneqq \mathcal C(\phi \otimes \psi^{\otimes n})$
minimizes the diamond distance \cref{eq:diamondDistance}, $\opnorm{\mathcal R_\psi - \mathcal E_\psi}$,
for the unitary channel $\mathcal R_\psi(\phi) \coloneqq R_\psi \phi R^\dagger_\psi$. We will focus on processors $\mathcal{C}$ with a non-trivial circuit structure.
While it is possible to program reflections by measuring the copies $\psi^{\otimes n}$, then reflecting about the estimate of $\psi$,
we show in \cref{sec:mr} that $n$ must then scale with $d$
and this is unnecessary.

Despite the ubiquity of reflection operators in quantum query algorithms \cite{Magniez2011,Reichardt2009,Lee2011}, these algorithms often have the goal of approxmiately preparing the state corresponding the reflection axis. The goal of the programmable reflection protocols analyzed in the present work is the opposite: to approximate the reflection using a quantum program consisting of copies of a pure state. This motivation is the same as for the density matrix exponentiation approach to Hamiltonian and Lindbladian simulation \cite{Kimmel2017,go2025,wml} and the \textit{samplizer} approach to sample complexity of computing quantum information-theoretic quantities \cite{wang1,wang2}. In a different programmability setting in which queries to the dynamics are utilized to produce a program state, approximate programmable processors are used to implement the approximate storage and retrieval task \cite{Yang2020,yoshida2025}, which, when optimized, has been shown to be equivalent to optimal unitary estimation, deterministic port-based teleportation, and approximate parallel unitary inversion \cite{yoshida2024}. %

In contrast to approximate programmable processors, constructions of probabilistic exact programmable processors for qudit unitaries have appeared in the literature, in which the program states are linear combinations of maximally entangled states of two qudit registers and the unitary operation is implemented exactly, but with low probability (i.e., the probabilistic setting)~\cite{Hillery2002,Hillery2004}. However, the general goal in the analysis of approximate  programmable quantum processors is to understand how the complexity of the program state relates to the minimal possible error in approximating the target operation. 
For example, one approach for programming reflections about a pure state $\ket{\psi}$ is to identify the orthogonal and perpendicular components of $\ket\phi$ relative to $\ket{\psi}$
by using permutation gates on $\ket{\phi}\ket{\psi}^{\otimes n}$~\cite{Barenco1997}.
\textcite{Harrow2011} give an algorithm for reflections $R_{\psi_{AB}}$ of a bipartite system, using copies $\ket\psi_{AB}$ 
that are distributed between two systems. They show that
approximately identifying which part of the input state $\ket{\phi}_{AB}$ overlaps with $\ket{\psi}_{AB}$
is possible by implementing a global cyclic permutation using only local cyclic permutations and $\log(n)$ qubits of communication in each direction. A non-distributed version of this algorithm is the main example of the class of algorithms that we consider in the present work. 

Instead of global permutation gates,
programmable reflections can also be implemented on a qubit system by coupling to a many-qubit system using Heisenberg interactions~\cite{Mo2019}.
One first prepares a large spin-$\frac{n}{2}$ system in the initial state $\ket{J_{z}=\frac{n}{2}}$
which is then rotated along the representation of $U^{\otimes n}$,
for $U\ket{0} = \ket\psi$.
Then, one couples the spin-$\frac{n}{2}$ system to an input spin-$\frac{1}{2}$ system (essentially a qubit)
with the Hamiltonian $X \otimes J_x + Z \otimes J_z + Y \otimes J_y$,
for $SU(2)$ generators $J_x,J_z,$ and $J_y$ in spin-$\frac{n}{2}$, and Pauli operators $X,Z$, and $Y$,
that reflects the spin-$\frac{1}{2}$ system about the state  $\ket{\psi}$. 
This algorithm has worst-case fidelity $O(1/n)$ with the reflection channel $\mathcal{R}_{\psi}$,
the same as the algorithms considered in the present work for system dimension $d=2$.
Although this approach allows for programming $SU(2)$ rotations on systems of higher spin,
reflections about pure states in the higher spin $SU(2)$ representations do not correspond to elements of $SU(2)$ group,
so we do not pursue this class of spin-interaction algorithms in the present work.

Copies of a state $\rho$ also constitute a program state for approximately implementing the unitary operation given by the time evolution $e^{i t \rho}$, for time $t$ and generator $\rho$~\cite{Lloyd2014} (also known as the LMR algorithm).
Taking $t=\pi$ and providing copies of $\ket\psi$ as the program state, this algorithm implements an approximate reflection.
In the context of Hamiltonian simulation, the LMR algorithm is asymptotically optimal using $\bigtheta{\frac{1}{\epsilon}}$ copies,
for $\epsilon$ trace distance~\cite{Kimmel2017,go2025}.
The LMR algorithm falls into the the general class of algorithms that we focus on in the present work.

\AP In \cref{sec:apref}, we introduce and analyze a broad class of algorithms that use copies of an unknown state as the program.
We define the ""approximate reflection channels""
that are any algorithms specified by a quantum channel
\begin{equation}\label{eqn:appxreflchan}
	\intro*\apxRefCh_{\psi,V}(\phi) \coloneqq \ptrace*{P}{V \phi_S \otimes \psi^{\otimes n}_P V^\dagger}
\end{equation}
where the unitary $V \in U(d^{n+1})$ acts on the subspace $\mathbb{C}^{d} \otimes \sym^{n}(\mathbb{C}^{d})$ as a unitary linear combination of permutation operators, i.e., a unitary element of $\mathbb C[S_{n+1}]$.
The unitary $V$ must also be independent of the state $\ket{\psi}$ so that it is an approximate quantum processor.
For "approximate reflection channels",
we go beyond asymptotics and find the exact minimum diamond distance,
improving on \cite{Harrow2011}.
We then provide a simple implementation (\cref{fig:algorithm}) within this class that attains the minimum distance.

In \cref{sec:unitaryreflections}, we generalize our results to rotations $e^{i \alpha \psi}$, for $\alpha \in \mathbb R$.
We give optimal algorithms in terms of an angle parameter that we obtain numerically.
Moreover, we show that the LMR algorithm with $\rho = \psi$ is an "approximate reflection channel"
that only accesses copies of the program register sequentially.
We also give improved rotation angles for this algorithm at any finite $n$.

In \cref{sec:lower} we derive a lower bound for program dimension necessary to approximately program reflections
by making some assumptions on the existence of optimal probe states which have the general form of an optimal state for estimation of an element of the unitary group $U(d)$ \cite{PhysRevA.72.042338}.
The problem of existence of these probe states is formulated as a solution to a system of linear equations,
which we numerically confirm up to large $n$ and $d$.
In contrast to our algorithms,
this proof does not assume that the program state takes the form of independent copies or a state in the symmetric subspace, but rather allows information about the reflection to be more generally encoded in entangled states of the ampliated Hilbert space $(\mathbb{C}^{d})^{\otimes n}\otimes (\mathbb{C}^{d})^{\otimes n}$.
Our lower bound is tight, 
and recovers the bound of \cite{Kimmel2017},
which only applies when the program state consists of copies of $\psi$.
Moreover, the minimal program dimension is much smaller than needed for universal programming.

Finally, we show that a sequence of programmable rotations can be used to construct
a universal quantum processor that has improved scaling in $d$ when compared with \cite{Yang2020} in \cref{sec:universal}.
We conclude (\cref{sec:conclusion}) by fitting our results into a conjecture by \textcite{Yang2020} regarding the necessary program dimension for a parametric family of quantum gates.

\section{Approximate reflection channels\label{sec:apref}}
To start, let us prove that exact programming of an arbitrary reflection, $R_{\psi}$,
using the program state $\ket{\psi}^{\otimes n}$ is impossible for finite $n$
and, therefore, we can only program reflections approximately.
Similar statements exist for programming other general classes of operations (for example, all of $U(d)$ \cite{Nielsen1997} or the qubit NOT operation \cite{Buzek1999}).
Intuitively, this result is similar to the no-go theorem for universal programming of unitary gates~\cite{Nielsen1997},
namely that the Hilbert space of program states can never be big enough to encode the information about the reflection operators that they program.

We define a rotation $R_\psi(\alpha) \coloneqq e^{i \alpha \ketbra{\psi}{\psi}}$, for $\alpha \in \mathbb R$,
then we let a reflection be $R_\psi \coloneqq R_\psi(\pi)$.
We define the rotation channel to be the unitary channel constructed by conjugation with $R_\psi(\alpha)$, 
\begin{equation}\label{eq:rotationChannel}
    \intro*\rotChannel_\psi(\alpha)(X) \coloneqq R_\psi(\alpha) X R_\psi^\dagger(\alpha),
\end{equation}
then the reflection channel is $\intro*\refChannel_\psi \coloneqq \mathcal R_\psi(\pi)$.

\begin{theorem}[No-go for perfect reflection with copies of a pure state]\label{thm:imposs}
Let $P \cong (\mathbb{C}^{d})^{\otimes n}$ be a program register, $S \cong (\mathbb{C}^{d})^{\otimes k}$ an input register, $E$ an arbitrary environment register, and $V: SP \rightarrow SPE$ an isometry. For any $n,k\in \mathbb{N}$, there is no quantum channel $C_{\psi}$ on states of $S$ of the form 
\begin{equation}\label{eqn:proc}
    C_{\psi}(\phi):= \ptrace*{PE}{V \phi_S \otimes (\psi^{\otimes n})_{P}  V^{\dagger}}
\end{equation}
such that for all normalized pure states $\ket{\psi}\in \mathbb{C}^{d}$
\begin{equation}
    C_{\psi}(\phi)=\mathcal{R}_{\psi}^{\otimes k}(\phi).
\end{equation}
\end{theorem}

\AP For the proof, we introduce some common notation. 
The ""trace norm"" of a linear map, $\Psi$, is defined as
\begin{equation}
	\norm{\Psi}_1 \coloneqq \max_{X : \norm{X}_1 \le 1} \norm{\Psi(X)}_1,
\end{equation}
for a maximization over linear operators fitting the map $\Psi$.
Then the "diamond norm" (or completely bounded trace norm) can be defined as
\begin{equation}\label{eq:diamondDistance}
	\intro*\opnorm{\Psi} \coloneqq \norm{\Psi \otimes \idm}_1,
\end{equation}
where the identity map is of the same dimensionality as $\Psi$.
The "diamond norm" measures that maximum "trace norm" of a map even when considering some entangled ancilla register of arbitrary dimension.
The difference map
\begin{equation}
	\paren{\Psi - \Phi}\paren{X} \coloneqq \Psi(X) - \Phi(X)
\end{equation}
is used to define the ""diamond distance"" between two arbitrary channels, $\Phi$ and $\Psi$,
as $\opnorm{\Phi - \Psi}$.

\begin{proof}[Proof of \cref{thm:imposs}]
Note that $C_{\psi}$ depends on the $\psi$ only through the program register. Therefore,  (\ref{eqn:proc}) expresses the Stinepring dilation of a programmable quantum processor.
Assume that $C_{\psi_{j}}(\phi)=\mathcal{R}_{\psi_{j}}^{\otimes k}(\phi)$ for $j=1,2$, where $\mathcal{R}_{\psi_{j}}$ is the unitary channel associated with $R_{\psi_{j}}$,
and assume that $\abs*{\braket{\psi_{1}}{\psi_{2}}} =\cos \varphi \neq 0$. One finds that
\begin{align}
    \opnorm*{\mathcal{R}_{\psi_{1}}^{\otimes k}-\mathcal{R}_{\psi_{2}}^{\otimes k}}
    &= \sup_{\ket{\phi}\in (\mathbb{C}^{d})^{\otimes k}} \norm*{(\mathcal{R}_{\psi_{1}}^{\otimes k}-\mathcal{R}_{\psi_{2}}^{\otimes k})(\phi)}_{1} \nonumber \\
    &\le \sup_{\ket{\phi}\in (\mathbb{C}^{d})^{\otimes k}} \norm*{V  \phi \otimes (\psi_{1}^{\otimes n}  - \psi_{2}^{\otimes n})V^{\dagger}}_{1} \nonumber  \\
    &= \norm*{\psi_{1}^{\otimes n} -\psi_{2}^{\otimes n}}_{1} \nonumber \\
    &= 2\sqrt{1-\cos^{2n}\varphi}.
\label{eqn:oqoq}
\end{align}
The first inequality is monotonicity of the trace distance under the quantum channel $\text{tr}_{PE}$. The diamond distance between unitary channels $\mathcal{U}^{\otimes k}$ and $\mathcal{V}^{\otimes k}$ is given by 2 if the smallest arc $\Omega$ on the unit circle containing the spectrum $\sigma((U^{\dagger}V)^{\otimes k})$ satisfies $\vert \Omega \vert \ge \pi$, and is given by $2\sin \frac{\abs{\Omega}}{2}$ otherwise \cite[Proposition~18]{nech}.
For $\mathcal{U}=\mathcal{R}_{\psi_{1}}$, $\mathcal{V}=\mathcal{R}_{\psi_{2}}$, this arc has length $4k\varphi$. Therefore, the diamond distance is equal to 2 when $k> \frac{\pi}{4\varphi}$.
But the right hand side of (\ref{eqn:oqoq}) is less than 2 regardless of $n$, a contradiction.
One can consider $\abs{\braket{\psi_{1}}{\psi_{2}}}$ large enough that the contradiction holds for all nonzero $k$.
\end{proof}

A channel $\Phi$ that maps between spaces of the same dimension
is \intro[covariant]{covariant} with respect to a group $G$ if there exist unitary representations $\set{U_g}_{g\in G}$ such that
\begin{equation}
	\Phi = \mathcal{U}_g^\dagger \circ \Phi \circ \mathcal{U}_g,
\end{equation}
for all $g\in G$~\cite{Khatri2024}\footnote{Note that this is a stronger condition than usually required for covariant channels.}.

In \cref{thm:mustCovariant} below, we show that the channels that best approximate $\mathcal R_\psi$ in diamond distance
are covariant with respect to the group $\langle R_{\psi}, \idm\rangle'$,
which is the commutant of the group $\set{R_\psi, \idm}$.
Subsequently, we restrict our analysis to channels that are covariant with respect to $\langle R_{\psi}, \idm\rangle'$. 
Let us briefly show that $\langle R_\psi, \idm\rangle'$ is equivalent to the $U(d-1)$ subgroup of $U(d)$ that leaves $\ket{\psi}$ invariant up to a phase.
First, if $g \in  \langle R_{\psi},\idm\rangle '$ then $[U_g, R_\psi] = 0$ and $U \ket\psi = c\ket\psi$ for a unit modulus complex number $c$.
Therefore, $\mathcal{U}(\psi)=\psi$.
Conversely, if $U\ket{\psi}=c\ket{\psi}$,
then $U$ commutes with $\ketbra{\psi}{\psi}$. Therefore, $[U,R_{\psi}]=0$.

\begin{theorem}\label{thm:mustCovariant}
    Given any quantum channel $\mathcal E$ that approximates a reflection,
    then let
    \begin{equation}
        \bar{\mathcal{E}}(\phi) \coloneqq \int dU \, U^{\dagger}\mathcal{E}(U \phi U^{\dagger})U = \int dU \, \mathcal{U}^{*}\circ \mathcal{E}\circ \mathcal{U}(\phi)
    \end{equation}
    with $dU$ the Haar measure on $\langle R_{\psi},\idm \rangle'$.
    Then
    \begin{equation}
        \opnorm{\mathcal R_\psi - \bar{\mathcal E}} \le \opnorm{\mathcal R_\psi - \mathcal E}.
    \end{equation}
\end{theorem}
\begin{proof}
    By an alternative definition of the diamond distance~\cite{WatrousBook}
    \begin{align}
        \opnorm{ \mathcal{R}_{\psi}  -\mathcal{E} } &= \sup_{\phi_{RS}} \norm*{(\text{Id}_{R}\otimes \mathcal{R}_{\psi} - \text{Id}_{R}\otimes \mathcal{E})(\phi_{RS})}_{1} \\
        &= \sup_{0\preceq X \preceq \idm_{RS}}\sup_{\phi_{RS}} \abs*{\trace*{X(\text{Id}_{R}\otimes \mathcal{R}_{\psi} - \text{Id}_{R}\otimes \mathcal{E})(\phi_{RS})}}
    \end{align}
    we see that the diamond distance is a supremum of convex functions on the set of channels, so it is convex.
    One finds that
    \begin{align}
        \opnorm{\mathcal R_\psi - \bar{\mathcal E}}
        &= \sup_{\phi_{RS}} \norm*{ \text{Id}_R \otimes \paren*{\mathcal R_\psi - \int dU\, \mathcal U^* \circ \mathcal E \circ \mathcal U} \paren{\phi_{RS}}}_1 \\
        &\le \sup_{\phi_{RS}} \int dU\, \norm{ \text{Id}_R \otimes \paren{\mathcal R_\psi - \mathcal U^* \circ \mathcal E \circ \mathcal U} \paren{\phi_{RS}}}_1, \\
        \intertext{by convexity of the diamond distance,}
        &\le \int dU \opnorm{\mathcal U \circ \mathcal R_\psi - \mathcal E \circ \mathcal U} \\
        &= \opnorm{\mathcal R_\psi - \mathcal E},
    \end{align}
    by the unitary invariance of trace norm
    and the covariance of $\mathcal R_\psi$ with respect to $\langle R_\psi, \idm\rangle'$.
\end{proof}
Then, by definition, $\mathcal{E}$ is covariant with respect to this unitary subgroup if and only if $\bar{\mathcal{E}}=\mathcal{E}$.
Therefore, we restrict to channels such that $\bar{\mathcal{E}}=\mathcal{E}$
and with \cref{thm:imposs}, define "approximate reflection channels" in \cref{eqn:appxreflchan}.

We show that $\apxRefCh_{\psi,V}$ is covariant with respect to $\langle R_\psi, \idm \rangle'$.
\begin{lemma}[Covariance property]\label{lem:covariance}
	Given $n \in \mathbb N$, a state $\ket\psi$,
	and any unitary channel $\mathcal U$ satisfying $\mathcal U(\psi) = \psi$,
	then
	\begin{equation}
		\mathcal U \circ \apxRefCh_{\psi,V}(\Phi) = \apxRefCh_{\psi,V} \circ \mathcal U (\Phi).
  \label{eqn:lklk}
	\end{equation}
\end{lemma}
\begin{proof}
	First, noting that $U\ket{\psi}=c\ket{\psi}$ for a complex number of modulus 1 (this condition is equivalent to $\mathcal{U}(\psi)=\psi$), we find
	\begin{align}
		\apxRefCh_{\psi,V} \circ \mathcal U(\Phi) &= \ptrace*{P}{V \paren*{\mathcal U^{\otimes n+1} \paren*{\Phi \otimes \psi^{\otimes n}}}V^{\dagger}} \\
		&= \ptrace*{P}{\mathcal U^{\otimes n+1}\paren*{V (\Phi \otimes \psi^{\otimes n}) V^{\dagger}}} \\
		&= \mathcal U \circ \apxRefCh_{\psi,V}(\Phi),
	\label{eqn:lklklk}
 \end{align}
	where in the second line we used the fact $U^{\otimes n+1}$ commutes with the action of any permutation of the $n+1$ registers and in the last line we used cyclicity of the trace on the $P$ register when taking the partial trace.
\end{proof}

Clearly, $\refChannel_\psi$ is also covariant under $\langle R_{\psi},\idm\rangle'$
and so is the difference map $\refChannel_{\psi} - \apxRefCh_{\psi,V}$.

\subsection{Worst-case states for the diamond distance}
We now show that the maximization involved in the computation of the diamond distance can be performed over a family of states parametrized by a real number in $[0,1]$ without loss of generality,
significantly reducing the search space.

\begin{lemma}\label{lem:twirledStates}
    The family of one-parameter states, with $p \in [0,1]$,
    \begin{equation}
        \ket{\phi_p}_{RS} \coloneqq \sqrt p \ket{0}_R \ket\psi_S + \sqrt{\frac{1-p}{d-1}} \sum_{i=1}^{d-1}\ket{i}_R \ket{\psi_i}_S,
    \end{equation}
    for $\ket{\psi}, \ket{\psi_{1}}, \ldots,\ket{\psi_{d-1}}$ an orthonormal basis of $S$,
    maximizes the diamond distance so that
    \begin{equation}
        \opnorm{\rotChannel_\psi(\alpha) - \apxRefCh_{\psi, V}} = \max_{p \in [0,1]} \Vert\mathrm{Id}_R \otimes \paren*{\rotChannel_\psi(\alpha) - \apxRefCh_{\psi, V}}(\phi_p)\Vert_{1}.
    \end{equation}
\end{lemma}
\begin{proof}
    Let $G = \langle \idm, R_\psi \rangle'$, then \cref{lem:covariance} implies $\apxRefCh_{\psi,V}$ is covariant with respect to $G$,
    and $\refChannel_\psi(\alpha)$ obviously also is.
    We define the twirling map
    \begin{equation}
        \mathcal T^G(\rho) \coloneqq \int_{g \in G} dg \mathcal U^g(\rho),
    \end{equation}
    where $\mathcal U^g$ are the unitary channel representations of the group $G$.
    We have that $\mathcal T^G \circ \paren{\rotChannel_\psi(\alpha) - \apxRefCh_{\psi, V}} = \paren{\rotChannel_\psi(\alpha) - \apxRefCh_{\psi, V}} \circ \mathcal T^G$.
    Now, since $\rotChannel_\psi(\alpha) - \apxRefCh_{\psi, V}$ is Hermiticity-preserving,
    \cite[Proposition~4.20]{Khatri2024} implies
    \begin{equation}
        \opnorm{\rotChannel_\psi(\alpha) - \apxRefCh_{\psi, V}} = \max_{\ket{\phi} : \ptrace{R}{\phi} = \mathcal T^G(\ptrace{R}{\phi})} \Vert\text{Id}_R \otimes \paren*{\rotChannel_\psi(\alpha) - \apxRefCh_{\psi, V}}(\phi)\Vert_{1}.
    \end{equation}
    The twirling map, $\mathcal T^G$, projects any arbitrary state $\rho$ onto two distinct subspaces as
    \begin{align}
        \mathcal T^G(\rho) &= \int_G dg \mathcal U^g(\rho) \\
                           &= \trace{\Pi_\psi \rho} \Pi_\psi + \int_G dg \mathcal U^g(\Pi^\perp_\psi \rho \Pi^\perp_\psi) \\
                           &= \trace{\Pi_\psi \rho} \Pi_\psi + \trace{\Pi^\perp_\psi \rho} \frac{\Pi^\perp_\psi}{d-1}
    \end{align}
    where $\Pi_\psi = \ketbra{\psi}{\psi}$ and $\Pi^\perp_\psi = \idm - \Pi_\psi$.
    Thus any state $\rho$ that is invariant under the twirling map, $\mathcal T^G(\rho) = \rho$,
    must be a diagonal mixed state of the form $\rho = p \Pi_\psi + \frac{1-p}{d-1} \Pi^\perp_\psi$
    with $p \in [0,1]$.
    Up to an unimportant isometry on the reference system, any purification of $\rho$ equals $\ket{\phi_p}$.
\end{proof}

\subsection{Characterizing unitary elements of symmetric algebra}\label{sec:cyclicPermutation}%
\AP \textcite{Harrow2011} realized that it is sufficient to consider only cyclic permutations as opposed to the full permutation group
to minimize the communication complexity.
Similarly, we show that it is sufficient to consider $V$ to only be a linear combination of cyclic permutations,
$\intro*\cyclicPermutation \coloneqq ( 1\; 2 \; \dots \; n)$.

\begin{lemma}\label{lem:cylicPermutationAlgebra}
	Let $\apxRefCh_{\psi,V}$ be an approximate reflection channel \cref{eqn:appxreflchan} and let $C$ be a cyclic permutation on $(\mathbb{C}^{d})^{\otimes n+1}$. Then there exist unique coefficients $c_{\ell}\in \mathbb{C}$, $\ell=0,\ldots,n$ such that $\sum_{\ell=0}^{n}c_{\ell}C^{\ell}$ is unitary and
    \begin{equation}
	    \apxRefCh_{\psi,V}(\phi) = \ptrace*{P}{\sum_{\ell,\ell'=0}^{n}c_{\ell}\overline{c_{\ell'}} \cyclicPermutation^{\ell}\psi^{\otimes n} \otimes \phi \cyclicPermutation^{-\ell '}}.
	\end{equation}
	In words, the action of the unitary operator $V$ in \cref{eqn:appxreflchan} agrees with that of a unitary element of the group subalgebra of cyclic permutations.
\end{lemma}
\begin{proof}
We decompose the symmetric group on $n+1$ elements, $S_{n+1}$, using the subgroup $H=S_n$ of size $n!$ that permutes the last $n$ factors in $(\mathbb C^d)^{n+1}$ and the subgroup of cyclic permutations,
$\cyclicPermutation^k$ for $k=1,\dots,n$, on $n+1$ elements.
We use the symbol $C$ both for a cyclic permutation in $S_{n+1}$ and its corresponding standard permutation representation on $(\mathbb{C}^{d})^{\otimes n+1}$.
It is a basic fact that we can write distinct left cosets $S_{n+1} = \set{\cyclicPermutation^0H, \cyclicPermutation H, \cyclicPermutation^2H, \dots, \cyclicPermutation^nH}$.
Therefore, as a general element of the group algebra $\mathbb{C}[S_{n+1}]$,
$V$ can be written
\begin{equation}
    V=\sum_{\ell=0}^{n}\sum_{h\in H}c_{\ell, h} \cyclicPermutation^{\ell}h
\end{equation}
because every $g\in S_{n+1}$ can be written as $g=\cyclicPermutation^{j}h$ for some $j$ and some $h\in H$.
Since $\psi^n$ is permutation-invariant, we can now simplify
\begin{align}
	\apxRefCh_{\psi,V}(\phi)&= \ptrace*{P}{\sum_{\ell,\ell'=0}^{n}\sum_{h,h'\in H}c_{\ell, h} \overline{c_{\ell',h'}} \cyclicPermutation^{\ell}h\psi^{\otimes n}h^{'-1} \otimes \phi \cyclicPermutation^{-\ell '}} \\
    &= \ptrace*{P}{\sum_{\ell,\ell'=0}^{n}c_{\ell} \overline{c_{\ell'}} \cyclicPermutation^{\ell}\psi^{\otimes n} \otimes \phi \cyclicPermutation^{-\ell '}}
\end{align}
where $c_{\ell}:=\sum_{h\in H}c_{\ell,h}$.
\end{proof}

Note that the coefficients $c_\ell$ in \cref{lem:cylicPermutationAlgebra} are constrained due to unitarity of $V$.
The following lemma characterizes the unitary linear combinations of cyclic permutations by using the discrete Fourier transform of the coefficients.
\begin{lemma}\label{lem:cyclicCoefficients}
	An element $\sum_{\ell=0}^{n}c_{\ell} \cyclicPermutation^{\ell} \in \mathbb{C}[S_{n+1}]$ is unitary if and only if $\vert \tilde{c}_{k}\vert=1$ for $k=0,\ldots,n$
	where $(\tilde{c}_{0},\ldots,\tilde{c}_{n})$ is related to $(c_{0},\ldots,c_{n})$ by the discrete Fourier transform
	\begin{align}
	    \tilde{c}_{k}&=\sum_{\ell=0}^{n}e^{2\pi i \ell k \over n+1}c_{\ell} \label{eq:fourierCoefficients}\\
        c_{\ell}&= {1\over n+1}\sum_{k=0}^{n}e^{- \frac{2\pi i \ell k}{n+1}}\tilde{c}_{k}.
	\end{align}
    For $n+1$ odd, any unitary linear combination of cyclic permutations can be written as
    \begin{equation}
        e^{i c_{0}\idm+i\sum_{j=1}^{{n\over 2}}\overline{c}_{j} \cyclicPermutation^{j} + c_{j} \cyclicPermutation^{-j}},
        \label{eqn:uhuh}
    \end{equation}
    with $c_{0}\in \mathbb{R}$ and $c_{j} \in \mathbb{C}$ for $j\neq 0$; and for $n+1$ even, 
    any unitary linear combination of cyclic permutations can be written as \begin{equation}
        e^{i c_{0}\idm+ic_{{n+1\over 2}} \cyclicPermutation^{{n+1\over 2}}+i\sum_{j=1}^{{n-1\over 2}}\overline{c}_{j} \cyclicPermutation^{j} + c_{j} \cyclicPermutation^{-j}},
        \label{eqn:uhuh2}
    \end{equation}
    with $c_{0},c_{{n+1\over 2}}\in \mathbb{R}$ and $c_{j} \in \mathbb{C}$ for $j\neq 0,{n+1\over 2}$.
\end{lemma}
\begin{proof}
Let $S_{n+1}$ act in its defining representation, in which the cyclic permutation $C$ is a $(n+1)\times (n+1)$ permutation matrix having simple eigenvalues $\set{e^{\frac{2\pi i k}{n+1}}}_{k=0}^n$ and associated orthonormal eigenvectors $\lbrace \ket{e_{k}}\rbrace_{k=0}^{n}$. Then 
\begin{equation}
    \sum_{\ell=0}^{n}c_{\ell} \cyclicPermutation^{\ell} = \sum_{\ell=0}^{n}c_{\ell}\sum_{k=0}^{n}e^{{2\pi i \ell k \over n+1}}\ket{e_{k}}\bra{e_{k}} = \sum_{k=0}^{n}\tilde{c}_{k}\ket{e_{k}}\bra{e_{k}}
\end{equation}
so the unit modulus condition $\vert \tilde{c}_{k}\vert=1$  for unitarity follows from the spectral theorem for unitary matrices.

To prove \cref{eqn:uhuh,eqn:uhuh2}, note that when $n+1$ is even, the subalgebra of cyclic permutations  has two self-inverse elements, $\idm$ and $C^{{n+1\over 2}}$. The proofs for $n+1$ even or odd are the same, so we simply consider $n+1$ to be odd.
First, note that the operator \cref{eqn:uhuh} is unitary because its logarithm is anti-Hermitian. Mapping $c_{0}\mapsto 2c_{0}$ in \cref{eqn:uhuh},
the operator \cref{eqn:uhuh} becomes 
\begin{equation}
    \sum_{k=0}^{n}e^{i\sum_{j=0}^{{n\over 2}} \left[\overline{c}_{j}e^{2\pi i j k\over n+1} + c.c.\right]}\ket{e_{k}}\bra{e_{k}}.
    \label{eqn:dfdf}
\end{equation}
But any unitary element of  the cyclic permutation subalgebra of $\mathbb{C}[S_{n+1}]$ can be written $\sum_{k=0}^{n}e^{i\theta_{k}}\ket{e_{k}}\bra{e_{k}}$ for $\theta_{k}\in\mathbb{R}$. The real sequence $\lbrace \theta_{k}\rbrace_{k=0}^{n}$ has a discrete Fourier transform specified by $c_{0}\in \mathbb{R}$ and ${n\over 2}$ complex numbers according to the exponent in (\ref{eqn:dfdf}).
\end{proof}

We have described the quantum channel $\apxRefCh_{\psi,V}$ in terms of its Stinespring unitary. The Kraus operators for the channel can be found by first noting that the orthonormal basis \begin{equation}
      {1\over \sqrt{d-1}}\sum_{j=1}^{d-1}\ket{\psi_{j}}\ket{\psi_{j}} \cup \lbrace \ket{\psi}\ket{\psi_{j}}\rbrace_{j=0}^{d-1}
    \label{eqn:bas2}
\end{equation}
with $\ket{\psi_{0}}:=\ket{\psi}$ contains the support of the Choi operator
\begin{equation}
    (\apxRefCh_{\psi,V}\otimes \text{Id}_{d})\text{vec}\idm_{d}(\text{vec}\idm_{d})^{\intercal} \in \text{End}((\mathbb{C}^{d})^{\otimes 2}),
    \label{eqn:choiop}
\end{equation}
where for complex vector spaces $X$, $Y$ with respective bases $\lbrace \ket{x_{j}}\rbrace_{j}$ and $\lbrace \ket{y_{k}}\rbrace_{k}$, we define the linear $\text{vec}$ mapping by
\begin{align}
    &{}\text{vec}: \text{Hom}(X,Y) \rightarrow Y\otimes X \nonumber \\
    &{} \text{vec}\left( \ket{y_{k}}\bra{x_{j}} \right) := \ket{y_{k}}\ket{x_{j}}. 
\end{align}
In the basis (\ref{eqn:bas2}), the Choi operator is given by
\begin{align}
    &{} \begin{pmatrix}
         (d-1)\vert c_{0}\vert^{2} & \overline{c}_{0}\tilde{c}_{0} \sqrt{d-1} & & & \\
        c_{0}\overline{\tilde{c}}_{0} \sqrt{d-1} & \sum_{\ell,\ell'}c_{\ell}\overline{c}_{\ell'} & & & \\
        & & \sum_{\ell>0}\vert c_{\ell}\vert^{2} & & \\
        & & & \ddots & \\
        & &  & & \sum_{\ell>0}\vert c_{\ell}\vert^{2}
    \end{pmatrix}  \nonumber \\
    &{} = \text{vec}(\tilde{c}_{0}\ket{\psi}\bra{\psi} + c_{0}\sum_{j>0}\ket{\psi_{j}}\bra{\psi_{j}})\left(\text{vec}(\tilde{c}_{0}\ket{\psi}\bra{\psi} + c_{0}\sum_{j>0}\ket{\psi_{j}}\bra{\psi_{j}})\right)^{\dagger} \nonumber \\
    &+ \sum_{\ell>0}\vert c_{\ell}\vert^{2}\sum_{j=1}^{d-1}\text{vec}(\ket{\psi}\bra{\psi_{j}})\left(\text{vec}(\ket{\psi}\bra{\psi_{j}})\right)^{\dagger}
\end{align}
where the matrix elements follow from using $\idm_{d} = \ket{\psi}\bra{\psi}+\sum_{j=1}^{d-1}\ket{\psi_{j}}\bra{\psi_{j}}$ in the definition (\ref{eqn:choiop}) and applying $\apxRefCh_{\psi,V}\otimes \text{Id}_{d}$.
This shows that the Choi rank of $\apxRefCh_{\psi,V}$ is $d$. The $d$ Kraus operators can now be read off
\cite[Proposition~2.20]{WatrousBook}\begin{align}
    K_{0}&:= \tilde{c}_{0}\ket{\psi}\bra{\psi} + c_{0}\sum_{j>0}\ket{\psi_{j}}\bra{\psi_{j}} \nonumber \\
    K_{j}&:= \sqrt{\sum_{\ell>0}\vert c_{\ell}\vert^{2}}\ket{\psi}\bra{\psi_{j}} \nonumber \\
    &= \sqrt{1-\vert c_{0}\vert^{2}} \ket{\psi}\bra{\psi_{j}}\; , \; j=1,\ldots, d-1
    \label{eqn:krausdirect}
\end{align}
with $\sum_{j=0}^{d-1}K_{j}^{\dagger}K_{j} = \idm_{d}$ following from the unitarity of $V$.
Specifically, $\vert \tilde{c}_{0}\vert^{2}=1$ from Lemma \ref{lem:cyclicCoefficients}, and the corresponding Parseval relation gives $\sum_{\ell=0}^{n}\vert c_{\ell}\vert^{2}={1\over n+1}\sum_{k=0}^{n}\vert \tilde{c}_{k}\vert^{2}=1$.

The complementary channel $\hat{\mathcal{E}}_{\psi,V}$ is defined by taking the trace over $S$ instead of the trace over $P$ in the definition (\ref{eqn:appxreflchan}) and defines the dynamics of the program register. For the purposes of calculating the corresponding Choi operator in $\text{End}(PS)$, we introduce the following orthonormal set in $P$
\begin{align}
    &{}\ket{\psi}^{\otimes n}_{P} \cup \lbrace \ket{v_{j}}_{P}\rbrace_{j=1}^{d-1}\nonumber \\
    &{}\ket{v_{j}}:={1\over \sqrt{\sum_{\ell=1}^{n}\vert c_{\ell}\vert^{2}}}\sum_{\ell=1}^{n}c_{\ell}C^{\ell}_{P}\ket{\psi}^{\otimes n-1}_{P_{1}\cdots P_{n-1}}\ket{\psi_{j}}_{P_{n}},
    \label{eqn:vvvec}
\end{align}
where $C_{P}$ is the cyclic permutation on the $P$ register only and we note that $C^{n}_{P}=\idm_{P}$. In $PS$, the orthonormal set containing the support of the Choi operator is given by
\begin{align}
    \ket{\psi}^{\otimes n}_{P}\ket{\psi}_{S} \cup {1\over \sqrt{d-1}}\sum_{j>0}\ket{v_{j}}_{P}\ket{\psi_{j}}_{S} \cup \lbrace \ket{\psi}^{\otimes n}_{P}\ket{\psi_{j}}_{S}\rbrace_{j=1}^{d-1},
\end{align}
and the rank-$d$ Choi operator is 
\begin{align}
   &{} \text{vec}\left( K_{0} \right)  \text{vec}\left( K_{0}\right)^{\dagger} +\sum_{j>0}\text{vec}\left( K_{j} \right)  \text{vec}\left( K_{j} \right)^{\dagger},
\end{align}
where the Kraus operators are
\begin{align}
    K_{0}&:= \tilde{c}_{0}\ket{\psi}_{P}^{\otimes n}\bra{\psi}_{S} + \sqrt{1-\vert c_{0}\vert^{2}}\sum_{j>0}\ket{v_{j}}_{P}\bra{\psi_{j}}_{S} \nonumber \\
    K_{j}&:= c_{0}\ket{\psi}^{\otimes n}_{P}\bra{\psi_{j}}_{S}\; , \; j=1,\ldots,d-1.
    \label{eqn:compkraus}
\end{align}

The Kraus operators (\ref{eqn:krausdirect}) will be derived from a different perspective in Appendix \ref{sec:directkraus} and used in Appendix \ref{sec:chanupper} to derive an upper bound on the worst case error the programmable processor for reflections that we consider in Section \ref{sec:unitaryImplementation}. The Kraus operators (\ref{eqn:compkraus}) are again derived in Appendix \ref{sec:cckraus} and then used in Appendix \ref{sec:compchanlower} to derive a lower bound on the worst case error between the complementary channel of a programmable processor for reflections and the program preparation channel.

\subsection{Optimal distance}\label{sec:optimalDistance}
We evaluate, in more generality than needed at the moment (see \cref{sec:unitaryreflections}),
the trace distance of $\apxRefCh_{\psi, V}$ with a rotation of arbitrary angle $\alpha \in \mathbb R$
when given a $\phi_p$ state, defined in \cref{lem:twirledStates}.

\begin{lemma}\label{lem:normPOptimization}
    For any $\alpha \in \mathbb R$, unitary $V = \sum_{\ell=0}^n c_{\ell} \cyclicPermutation^{\ell}$, and $p \in [0,1]$
    it holds that
    \begin{multline}
        \norm*{\idm_R \otimes \paren{\rotChannel_\psi(\alpha) - \apxRefCh_{\psi, V}} (\phi_p)}_1 = 
        (1-p) \paren*{1- {\abs{c_{0}}^{2}}} \\
        + \sqrt{(1-p)^2 \paren*{1 - \abs*{c_0}^2}^2 + 4(1-p)p\abs*{\tilde c_0 \overline{c_0} - e^{i\alpha}}^2},
    \end{multline}
    with $\tilde c_0 = \sum_{\ell = 0}^n c_\ell$ following \cref{eq:fourierCoefficients}.
\end{lemma}

\begin{proof}
    Recall the orthonormal basis from \cref{lem:twirledStates} to write $\phi_{p}$.
    By using the inverse Fourier transform and the fact that $\sum_{\ell=0}^n e^{\frac{2\pi i (k - k')\ell}{n+1}} = (n+1)\delta_{kk'}$, we obtain
    {\allowdisplaybreaks%
    \begin{align}
        \begin{split}
    &\idm_R \otimes \paren*{\rotChannel_\psi(\alpha) - \apxRefCh_{\psi,V}} (\phi_p) = p\paren{1 - \abs{\tilde c_0}^2} \ketbra{0}{0} \otimes \ketbra{\psi}{\psi} \\
    &\quad + \frac{1-p}{d-1} \paren*{1 - \abs*{c_0}^2} \paren*{\sum_{j=1}^{d-1} \ket{j}_R \ket{\psi_j}} \paren*{\sum_{j'=1}^{d-1} \bra{j'}_R \bra{\psi_{j'}}} \\
    &\quad - \sqrt{\frac{p(1-p)}{d-1}} \sum_{j=1}^{d-1} \paren*{\paren*{\tilde{c}_{0}\overline{c_{0}} - e^{i\alpha}}\ketbra{0}{j}_R \otimes \ketbra{\psi}{\psi_j} + \text{h.c.}} \\
    &\quad - \frac{1-p}{(d-1)} \paren*{ \frac{\sum_{k=0}^n \abs{\tilde c_k}^2}{n+1} - \abs*{c_0}^2 } \paren*{\sum_{j=1}^{d-1} \ketbra{j}{j}_R \otimes \ketbra{\psi}{\psi} }
        \end{split}\\
        \begin{split}
    &= \frac{1-p}{d-1} \paren*{1 - \abs{c_0}^2} \paren*{\sum_{j=1}^{d-1} \ket{j}_R \ket{\psi_j}} \paren*{\sum_{j'=1}^{d-1} \bra{j'}_R \bra{\psi_{j'}}} \\
    &\quad - \sqrt{\frac{p(1-p)}{d-1}} \sum_{j=1}^{d-1} \paren*{\paren*{\tilde c_0 \overline{c_0} - e^{i\alpha}} \ketbra{0}{j}_R \otimes \ketbra{\psi}{\psi_j} + \text{h.c.}} \\
    &\quad - \frac{1-p}{d-1} \paren*{1 - \abs{c_0}^2} \sum_{j=1}^{d-1} \ketbra{j}{j}_R \otimes \ketbra{\psi}{\psi}
        \end{split}
    \end{align}%
    }%
    since $\abs{\tilde c_k} = 1$ by \cref{lem:cyclicCoefficients} and where we suppressed the register label $S$.
    Observe that the above matrix contains a diagonal submatrix spanned by $\set{\ket{j}\ket{\psi}}_{j=1}^{d-1}$,
    so the trace norm can be written as
    \begin{multline}
        \norm*{\text{Id}_R \otimes \paren{\rotChannel_\psi(\alpha) - \apxRefCh_{\psi,V}}(\phi_p)}_1  = (1-p) \paren*{1- {\abs{c_{0}}^{2}}} \\
        + \norm*{\begin{pmatrix}
                0&\sqrt{p(1-p)}\left( {\tilde{c}_{0}\overline{c_{0}}} - e^{i \alpha}\right) \\
                \sqrt{p(1-p)}\left( {\overline{\tilde{c}_{0}}c_{0}} - e^{-i\alpha} \right) & -(1-p) \paren{1 - \abs{c_{0}}^{2}}
        \end{pmatrix} }_1,
        \label{eqn:firsttrn}
    \end{multline}
    where the matrix acts in a two-dimensional subspace spanned by the orthonormal vectors
    \begin{align}
        \ket{\phi_0} &\coloneqq \ket{0}\ket{\psi}\\
        \ket{\phi_1} &\coloneqq \sqrt{\frac{1}{d-1}} \sum_{j=1}^{d-1} \ket{j} \ket{\psi_j}.
        \label{eqn:orthonorm2}
    \end{align}
    Evaluating the trace norm shows the result.
\end{proof}

Now we are ready to show the main result of this section,
the analytical computation of the diamond distance.

\begin{theorem}\label{thm:optimalDistance}
    Let $n>1$, then
    \begin{equation}\label{eqn:fsfsfs2}
	    \min_{V\in \mathbb C[S_{n+1}]} \opnorm{\refChannel_{\psi}- \apxRefCh_{\psi,V}} = {8(n+2)\over 8+4n+n^{2} },
    \end{equation}
    with the minimum only obtained by the algorithms
    \begin{equation}\label{eqn:optalg}
    V= \idm+ {e^{\pm i\cos^{-1}(f(n))}-1\over n+1}\sum_{\ell=0}^{n} \cyclicPermutation^{\ell}  \in \mathcal{W}
    \end{equation}
    up to a global phase, where $\mathcal{W}$ is the set of unitary elements in the subalgebra of $\mathbb{C}[S_{n+1}]$ spanned by the cyclic permutations and
    \begin{equation}\label{eq:optAngle}
        f(n) = -{n^{3}+6n^{2}+6n \over (n+2)^{3}}.
    \end{equation}
\end{theorem}

\begin{proof}
    By \cref{lem:twirledStates,lem:normPOptimization},
    \begin{equation}\label{eq:analyticalDistance}
        \begin{split}
	    \min_{V \in \mathcal{W}} \opnorm{\mathcal R_\psi - \apxRefCh_{\psi, V}} &= \min_{V\in \mathcal{W}}\max_{p\in [0,1]} \norm{\idm_R \otimes \paren*{\mathcal R_{\psi} - \apxRefCh_{\psi,V}} (\phi_p)}_{1} \\
            &= \min_{V\in\mathcal{W}}\max_{p\in [0,1]} (1-p) \paren*{1- {\abs{c_{0}}^{2}}} \\
            & \quad + \sqrt{(1-p)^2 \paren*{1 - \abs*{c_0}^2}^2 + 4(1-p)p\abs*{\tilde c_0 \overline{c_0} + 1}^2}.
        \end{split}
    \end{equation}
    This function is differentiable on the $(n+1)$-torus $\vert \tilde{c}_{k}\vert=1$, $k=0,\ldots, n$.
    By \cref{lem:cyclicCoefficients} we have $\tilde c_0 = e^{i\eta}$ for $\eta \in [0,2\pi)$,
    so \cref{eq:analyticalDistance} for fixed $p$ is a function only of $\sum_{k=1}^{n}\tilde{c}_{k}$ and $\eta$.
For $n\ge 2$, it is clear that $\frac{1}{n} \sum_{k=1}^{n}\tilde{c}_{k}$ varies over the closed unit disk in $\mathbb{C}$.
    The $n=1$ case is pathological because $\tilde{c}_{1}$ varies over the unit circle,
    but the optimization analysis below holds even for $n=1$.
    Therefore, one can parametrize ${1\over n}\sum_{k=1}^{n}\tilde{c}_{k} = re^{i\gamma}$ with $0\le r\le 1$ and $\gamma \in [0,2\pi)$. 

    The maximization over $p$ in \cref{eq:analyticalDistance} can be obtained by finding the critical points.
    In the domain of algorithms $V$ that satisfy $1-\vert c_{0}\vert^{2} \ge \abs*{\tilde{c}_{0}\bar{c_{0}} +1}$ (call this Domain A),
    the function \cref{eq:analyticalDistance} has no critical points with respect to $p$ so its maximum over $p$ in that domain is obtained at $p=0$. In Domain A, the diamond distance monotonically decreases with $r\in [{1\over n},1]$ and for any $r$, the diamond distance monotonically decreases as $u:=\eta - \gamma$ deviates from $\pi$.\footnote{$r=1/n$ is the smallest value of $r$ in Domain A.} Therefore, the minimal diamond distance over algorithms in Domain A occurs at the boundary of the domains and $r=1$.
    
    \begin{figure}[tb]
    	\centering
    	\includegraphics[scale=0.5,alt={A 3D plot showing the diamond distance dependent on r and eta minus gamma.}]{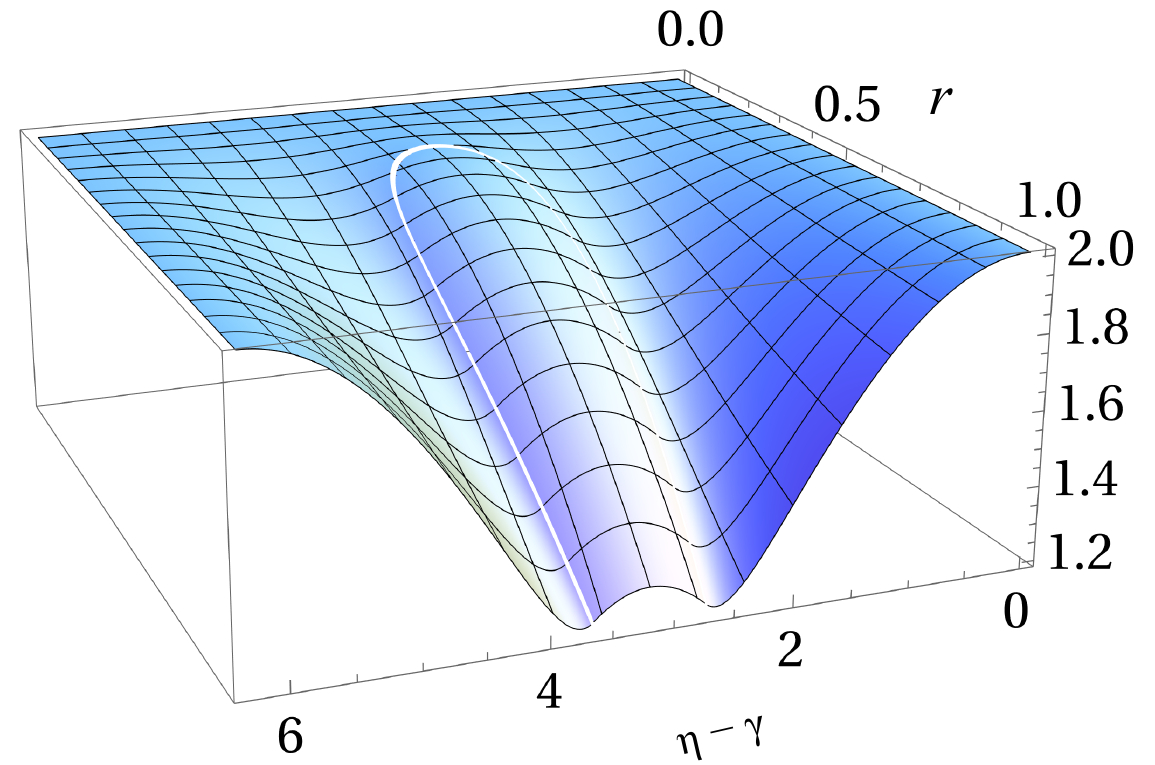}
    	\caption{%
    		The globally minimal diamond distance in \cref{thm:optimalDistance} is given by the minimum of the plotted function in the  domain
    		$1- \abs{c_{0}}^2 < \abs{\tilde{c}_{0} \overline{c_{0}} +1}$ (outside the white curve). This domain (Domain B) is characterized by the existence of a critical point with respect to $p$ and the diamond distance is given by the function in \cref{eqn:plpl}. Inside the white curve (Domain A), there is no critical point with respect to $p$ and the diamond distance $2(1-\vert c_{0}\vert^{2})$ is obtained at $p=0$. The figure shows $n=4$ so that the global minimum in Domain B is clearly visible.
    	}\label{fig:diamond}
    \end{figure}

    However, the global minimum of the diamond distance over algorithms $V$ occurs in the opposite domain
    where $1-\vert c_{0}\vert^{2} < \abs*{\tilde{c}_{0}\bar{c_{0}} +1}$ (Domain B). At the boundary between Domain A and B, the diamond distance is $2(1-\vert c_{0}\vert^{2})$, so the diamond distance is continuous.
    The critical point with respect to $p$ in Domain B is
    \begin{equation}\label{eqn:critp}
        p^* \coloneqq \frac{\abs*{{\tilde{c}_{0}\bar{c_{0}} } +1 } - (1-\vert c_{0}\vert^{2})}{2 \abs*{\tilde{c}_{0}\bar{c_{0}} +1 }- (1- \abs{c_{0}}^{2})}.
    \end{equation}
    Substituting $p^*$ in \cref{eq:analyticalDistance}, we obtain
    \begin{equation}
        \frac{2 \abs*{\tilde{c}_{0}\bar{c_{0}} +1}^{2}}{2 \abs*{\tilde{c}_{0}\bar{c_{0}} +1} + \abs{c_{0}}^{2}-1}.
	    \label{eqn:putput}
    \end{equation}
    Taking $u$ in \cref{eqn:putput} gives the full expression
    \begin{multline}\label{eqn:plpl}
	    2\paren*{(n+2+nr\cos u)^{2} + n^{2}r^{2}\sin^{2}u} / 
        \Bigl( (1+nr\cos u)^{2} + n^{2}r^{2}\sin^{2}u \\
        + 2(n+1)\sqrt{ (n+2+nr\cos u)^{2} + n^{2}r^{2}\sin^{2}u} 
        - (n+1)^{2} \Bigr)
    \end{multline}
    which is shown in \cref{fig:diamond}.
    This function is symmetric under $u\mapsto 2\pi -u$, so it is sufficient to examine the function for $u\in [0,\pi]$.
    The first critical line at $u=0$ gives the value 2, so it is a global maximum.
    Second, on $(u,r)\in [0,\pi]\times [0,1]$, one finds that regardless of $r$, \cref{eqn:plpl} has a critical line
    \begin{equation}\label{eqn:critu}
	    u^* = \cos^{-1}\left( {n^{3}(r^{3}-3r)-12n^{2}r-12nr \over 2(n+2)^{3}} \right)
    \end{equation}
    which is a local minimum of \cref{eqn:plpl} for fixed $r$.
    We substitute $u^*$ into \cref{eqn:plpl},
    resulting in a function of $r$ that is monotonically decreasing on $[0,1]$.
    The value of the resulting function at $r=1$ is the right hand side of \cref{eqn:fsfsfs2}. Because the values of the diamond distance at the other boundary points of Domain B are larger than this,
    we conclude that it is global minimum value. We obtain $f(n)$ in \cref{eq:optAngle} by taking $r=1$ in \cref{eqn:critu}.
    At $r=1$, the operator $V$ has the form
    \begin{align}
	    V&= \idm + \frac{ e^{iu}-1}{n+1}\sum_{\ell=0}^{n}C^{\ell}
    \end{align}
    up to a phase because $r=1$ if and only if $\tilde{c}_{k}=\tilde{c}_{k'}$ for all $k,k'>0$. We obtain (\ref{eqn:optalg}) by substituting $u^{*}$ and noting that the only remaining freedom in $V$ is the $\pm u$ symmetry, so (\ref{eqn:optalg}) describes all optimal algorithms.
\end{proof}

\subsection{Efficient unitary implementation}\label{sec:unitaryImplementation}
\begin{figure}
	\centering
	\begin{quantikz}
												& \gategroup[5,steps=7]{approximate rotation unitary}\wireoverride{n} & \lstick{\bra{\unifs}}\wireoverride{n} & \ctrl[wire style={"i"}]{1}\qwbundle{} & \gate{e^{i\theta \unifs}} & \ctrl[wire style={"j"}]{1} & \rstick{\ket{\unifs}} & \wireoverride{n} & \wireoverride{n} \\
		\lstick{\bra{\phi}} &                                              &                                  & \gate[4]{\cyclicPermutation^i}                         &                       & \gate[4]{\cyclicPermutation^{-j}}           &                  &                  &                  \\
		\lstick{\bra{\psi}} &                                              &                                  &                                       &                       &                            &                  &                  &                  \\
		\lstick{\bra{\psi}} &                                              &                                  &                                       &                       &                            &                  &                  &                  \\
		\lstick{\bra{\psi}} &                                              &                                  &                                       &                       &                            &                  &                  &
	\end{quantikz}
	\caption{%
        For appropriately chosen $\theta \in \mathbb R$,
        this unitary
		approximately implements the rotation $e^{i \alpha \ketbra{\psi}{\psi}}$,
        for $\alpha \in \mathbb R$,
		when given copies of $\ket{\psi}$ in the "program register".
        The ancilla is initialized in the uniform superposition,
		$\ket{\unifs} \coloneqq \frac{1}{\sqrt{n+1}} \sum_{i=0}^n \ket i$,
		and the control selects the multiplicity of the "cyclic permutation", $\cyclicPermutation$,
		implementing the operator $\sum_i \ketbra{i}{i} \otimes \cyclicPermutation^i$.
        In an "approximate reflection channel" we discard the "program registers" at the end.
        A natural, but suboptimal, choice for reflections ($\alpha = \pi$) is $\theta = \pi$~\cite{Harrow2011}.
        We find the optimal $\theta$ for reflections in \cref{sec:unitaryImplementation}.
        We discuss $\alpha$ in generality in \cref{sec:unitaryreflections}.
	}\label{fig:algorithm}
\end{figure}

\AP We now proceed to motivate our algorithm, which couples the program state $\ket{\psi}^{\otimes n}$ to an input state $\ket{\phi}$ in a $\ket{\psi}$-independent way to approximately yield $R_{\psi}\ket{\phi}$.
We use "QPE" to approximately distinguish the parallel and orthogonal components of the input state, $\ket\phi$, to $\ket\psi$~\cite{mosca}.
\textcite{Harrow2011} showed that
"QPE" with a cyclic permutation, $\cyclicPermutation$,
on the input state $\ket\phi \ket{\psi}^{\otimes n}$ roughly separates out the parallel and orthogonal components of the state.
In particular, $\ket{\psi}^{\otimes (n+1)}$ is a 1-eigenstate of $\cyclicPermutation$,
whereas $\ket{\psi^\perp}\ket{\psi}^{\otimes n}$ is contained in the symmetric subspace
with eigenvalues $\exp(2\pi i \frac{k}{n+1})$,
for $k \in [n] \coloneqq \set{0,1,2,\dots,n}$.
So, if we apply a negative phase to the 1-eigenspace of $\cyclicPermutation$,
then the only error we make is the $k=0$ eigenstate,
which is an $\bigo{\frac{1}{n}}$ fraction of $\ket{\psi^\perp}\ket{\psi}^{\otimes n}$.
More generally, we allow this phase to be a rotation defined by a parameter $\theta$.
The eigenvalue-1 subspace is marked simply by the 0 phase outcome of "QPE" with $\cyclicPermutation$,
so we perform the rotation $e^{i \theta \ketbra{0}{0}}$.
All that remains is to uncompute the "QPE".
By simplifying the circuit, we obtain \cref{fig:algorithm}.

Let $\intro*\apxRef_\theta$ be the unitary defined by the approximate rotation unitary in \cref{fig:algorithm}.
We now show an important property of $\apxRef_\theta$,
so that \cref{thm:optimalDistance} implies $\apxRefCh_{\psi, \apxRef_\theta}$ is
an optimal "approximate reflection channel" for some $\theta$.

\begin{theorem}\label{thm:algorithm}
    \Cref{fig:algorithm} defines an approximate rotation unitary, $\apxRef_\theta$,
    such that
    \begin{equation}\label{eqn:genth}
        \apxRef_\theta \ket{\phi}_S\ket{\psi}^{\otimes n}_P \coloneqq \paren*{\idm_{SP} + \frac{e^{i\theta} - 1}{n+1}\sum_{\ell=0}^{n} \cyclicPermutation^{\ell}}\ket{\phi}_S\ket{\psi}^{\otimes n}_P.
    \end{equation}
\end{theorem}
\begin{proof}
    Let us denote with $A$ the register containing the ancilla state,
    initialized in $\ket{\unifs} \coloneqq \frac{1}{\sqrt{n+1}} \sum_{\ell = 0}^n \ket{\ell}$.
    Note that $e^{i \theta \unifs} = \idm - (1-e^{i\theta}) \unifs$.
    Then \cref{fig:algorithm} equals
    \begin{multline}
        \paren*{\sum_{k=0}^n \ketbra{k}{k}_A \otimes \cyclicPermutation_{SP}^{-k}} \paren*{e^{i\theta \unifs}_A \otimes \idm_{SP}} \paren*{\sum_{j=0}^n \ketbra{j}{j}_A \otimes \cyclicPermutation_{SP}^j} \\
        = \frac{1}{\sqrt{n+1}} \sum_{k=0}^n \ket{k}_A \ket{\phi}_S \ket{\psi}^{\otimes n}_P - \frac{1-e^{i \theta}}{n+1} \frac{1}{\sqrt{n+1}} \ket{k}_A \sum_{j=0}^n \cyclicPermutation^{j-k} \ket{\phi}_S \ket{\psi}^{\otimes n}_P.
    \end{multline}
    Note that, for any $k \in \mathbb N$,
    \begin{equation}
        \sum_{j = 0}^n \cyclicPermutation^{j-k} \ket{\phi}_{S}\ket{\psi}^{\otimes n}_{P} = \sum_{j=0}^n \cyclicPermutation^j \ket{\phi}_{S}\ket{\psi}^{\otimes n}_{P}
    \end{equation}
    so that the state at the end of the algorithm equals
    \begin{equation}
        \ket{\unifs}_A \paren*{\idm_{SP} - \frac{1 - e^{i\theta}}{n+1} \sum_{j=0}^n \cyclicPermutation^j_{SP}} \ket{\phi}_S\ket{\psi}^{\otimes n}_P.
    \end{equation}
    After tracing out the ancilla register in the known state $\ket{\unifs}$,
    we get the result.
\end{proof}

By inspecting \cref{thm:optimalDistance,thm:algorithm}, it is clear that
the optimal $V$ can be constructed by setting $\theta = \pm \cos^{-1} f(n)$.

An alternative, natural choice is to maximize the negative phase in the 1-eigenspace of $\cyclicPermutation$
and pick $\theta = \pi$~\cite{Harrow2011}.
We show that this choice of $\theta$ is suboptimal,
but only to an additive distance of $\bigo*{n^{-3}}$. An $O({1\over n})$ upper bound on the diamond distance can be obtained without fully computing the diamond distance in Appendix \ref{sec:chanupper}.

\begin{corollary}\label{cor:piDistance}
    Let $\apxRef \coloneqq \apxRef_\pi$,
    then the family of "approximate reflection channels" $\apxRefCh_{\psi,\apxRef}$~\cite{Harrow2011} has
    \begin{equation}
        \opnorm{\mathcal{R}_{\psi}- \apxRefCh_{\psi,\apxRef}} = \frac{8n}{(n+1)^2}.
    \end{equation}
\end{corollary}
\begin{proof}
    We can write $\apxRef = \sum_{\ell=0}^n c_\ell \cyclicPermutation^\ell$ with
    \begin{align}
        c_{0}&= {n-1\over n+1} \\
        c_{\ell}&= -{2\over n+1} \, , \, \ell=1,\ldots, n.
    \end{align}
    Using these values of $c_\ell$, the Fourier coefficients are $\tilde c_0 = \sum_{\ell = 0}^n c_\ell = -1$ and $\tilde{c}_{k}=1$, $k=1,\ldots, n$.
    \cref{lem:covariance,lem:normPOptimization} imply
    \begin{equation}
        \begin{split}
	    \opnorm{\mathcal R_\psi - \apxRefCh_{\psi, \apxRef}} &= \max_{p\in [0,1]} \norm{\text{Id}_R \otimes \paren*{\mathcal R_{\psi} - \apxRefCh_{\psi,\apxRef}} (\phi_p)}_{1} \\
            &= \max_{p\in [0,1]} \frac{4}{(n+1)^2} \paren*{n \paren{1-p} +\sqrt{\paren{1-p}\paren{n^2 + 2np + p}}}.
        \end{split}
    \end{equation}
    As we saw in the proof of \cref{thm:optimalDistance},
    this algorithm is in Domain A ($\eta=\pi$, $r=1$, and $\gamma =0$) and its diamond distance from $\mathcal{R}_{\psi}$ is obtained when taking $p=0$ in $\phi_{p}$.
    Note that for $n>1$, this algorithm is not the optimal algorithm in Domain A.
    The optimal algorithm in Domain A occurs at the intersection points of $r=1$ and the boundary of Domain A and B.
\end{proof}

We show that $\apxRef_\theta$ can be implemented efficiently when $n$ is a power of two.

\begin{theorem}\label{thm:efficientRotation}
Assuming $L \coloneqq \log_2(n+1)$ is an integer, $\apxRef_\theta$ can be implemented in $2nL$ controlled-SWAP gates and $\bigo{n}$ single-qubit gates.
\end{theorem}
\begin{proof}We can decompose the controlled cyclic permutations into singly-controlled cyclic permutations since
\begin{equation}
    \sum_{b_1,\ldots, b_L \in \set{0,1}} \ketbra{b}{b}_{1\cdots L} \otimes \cyclicPermutation^b_{SP}
    = \prod_{j=1}^L \ketbra{0}{0}_j \otimes \idm_{SP} + \ketbra{1}{1}_j \otimes \cyclicPermutation^{2^j}_{SP}
\end{equation}
where $b$ is the integer specified by the bit string $b_K\ldots b_1$,
and we order the ancilla bits specifying $b_1$ as register $1$ to $b_K$ as register $K$.
Every cyclic permutation $\cyclicPermutation^{2^j}$ can be computed classically and decomposes into cyclic permutations.
It is easy to see that a cyclic permutation acting on $k \in \mathbb N$ qubits can be implemented in $k-1$ SWAP gates.
Every gate $\cyclicPermutation^{2^k}$ acts on $n+1$ qubits,
therefore it can be implemented in $n$ SWAP gates.
The controlled variant of $\cyclicPermutation^{2^k}$ can be implemented by controlling every SWAP gate.
The result follows trivially.\end{proof}

We conclude this section by emphasizing a property of our approximate programmable quantum processors $\mathcal{E}_{\psi,R_{\theta}}$ that is not part of the general definition of approximate programmable quantum processor,
and which may be useful when applying programmable processors in broader applications.
Consider the program preparation channel defined by \begin{align}
    \Prep(\psi)_{S\rightarrow P}(\phi)&:= \trace{\phi} \psi^{\otimes n} \nonumber \\
    &= \text{tr}_{E}W(\psi)_{S\rightarrow PE}\phi W(\psi)^{\dagger}_{S\rightarrow PE} 
    \label{eqn:prepchannelstine}
\end{align} for all states $\phi$ on $S$, where the isometry $W(\psi)_{S\rightarrow PE}\ket{i}:=\ket{\psi}^{\otimes n}_{P}\ket{i}_{E}$.
The following lemma shows that our approximate programmable quantum processors leaves the program state approximately unchanged,
but at a slower rate of convergence than that characterizing implementation of the rotation $\mathcal{R}_{\psi}(\theta)$ (see \cref{cor:rotationDistance}).

\begin{lemma}\label{lem:compchan}
The channel $\mathcal{E}_{\psi,R_{\theta}}$ has the property that
\begin{equation}
\opnorm{\widehat{\mathcal{E}}_{\psi, \apxRef_{\theta}} - \Prep(\psi)_{S\rightarrow P}}\sim  {2\vert \sin\left({\theta\over 2}\right) \vert \over \sqrt{n}}
\label{eqn:compldiam}
\end{equation}
as $n\rightarrow \infty$, where $\widehat{\mathcal{E}}_{\psi, \apxRef_{\theta}}(\phi):= \ptrace{S}{\apxRef_{\theta}\phi\otimes \psi^{\otimes n} \apxRef_{\theta}^{\dagger}}$ is the complementary channel to  $\mathcal{E}_{\psi,R_{\theta}}$. 
\end{lemma}

The proof of \cref{lem:compchan} is a lengthy exercise in identifying the support of $\widehat{\mathcal{E}}_{\psi, \apxRef_{\theta}} - \Prep(\psi)_{S\rightarrow P}$ and is given in \cref{sec:compproof}. Readers interested in an $\Omega({1\over \sqrt{n}})$ lower bound without an exact asymptotic for the diamond distance can refer to the simpler calculation in \cref{sec:bdd}. This property of approximately returning the program state  allows the possibility of reusing the program register when concatenating programmable processors as part of a larger approximately programmable algorithm.

\subsection{Saving space in the program register}\label{sec:savingSpace}
We have assumed a program register of the form $\ket\psi^{\otimes n}$,
but this is not the most efficient use of space.
The symmetric subspace is of lower dimension than the tensor product space of $n$ copies.
We show that there is an efficient state-independent isometry
that maps the larger symmetric state into a more compact representation and vice versa.

\AP To see why the symmetric state $\ket{\psi}^{\otimes n}$ can be more compactly represented,
let us define orthonormal basis vectors of the symmetric space, $\intro*\sym^{n}(\mathbb{C}^{d})$, of $n$ elements with dimension $d$ as
\begin{equation}
    \ket{e_{i_1,\ldots,i_n}} \coloneqq \frac{\sum_j \pi_j \ket{i_1} \otimes \cdots \otimes \ket{i_n}}{\norm{\sum_j \pi_j \ket{i_1} \otimes \cdots \otimes \ket{i_n}}},
\end{equation}
for $0 \le i_1 \le \cdots \le i_n \le d-1$
and where the sum is over all permutations $\pi_j$.
We can calculate the number of basis elements in the symmetric subspace using the ``stars and bars'' method as
\begin{equation}\label{eq:symmetricSpaceDimension}
	\dim\paren*{\sym^{n}(\mathbb{C}^{d})} = \binom{n + d - 1}{d-1} \le \frac{(n+d-1)^{d-1}}{(d-1)!}
\end{equation}
by having $n$ \emph{stars} and $d-1$ \emph{bars},
representing where the index $i_j$ increases.
Taking the logarithm,
\begin{equation}\label{eq:symmetricCost}
    \log_2\paren{\dim\paren*{\sym^{n}(\mathbb{C}^{d})}}
        \le (d-1)\log_2\paren*{\frac{n+d-1}{d-1}} + \bigo{d}
\end{equation}
by Stirling's approximation.
This provides an upper bound for the number of qubits in the program register.

We now define an isometry $\decoder_n : \mathbb{C}^{\dim(\sym^{n}(\mathbb{C}^{d}))} \to (\mathbb C^d)^{\otimes n}$ as
\begin{equation}
	\intro*\decoder_n \coloneqq \sum_{0 \le i_1 \le \dots \le i_n \le d-1} \ket{e_{i_1,\ldots,i_n}} \bra{i_1,\ldots,i_n},
\end{equation}
where $\ket{i_1,\ldots,i_n}$ are orthogonal states that index $\ket{e_{i_1,\ldots,i_n}}$.

The isometry $\decoder_n$ is only dependent on the uniform properties of the problem,
i.e., $d$ and $n$, and not $\ket\psi$.
Therefore, any processor having $\ket{\psi}^{\otimes n}$ as input may be converted
to one that has program dimension $\dim(\sym^{n}(\mathbb{C}^{d}))$,
allowing any pure program state similar to other universal programming results~\cite{Hillery2006,Yang2020}.

The quantum Schur transform~\cite{Bacon2007} is a unitary, $U_{\text{Schur}}$, that can be used to implement $\decoder_n$
if we perform
\begin{equation}\label{eq:decoderConstruction}
    \decoder_n \ket{\psi_P} = (\idm\otimes\bra{0}_a) U_{\text{Schur}}^\dagger \ket{\lambda = (n,0)}\ket{\psi_P} \ket{p_\lambda = 0}
\end{equation}
where we project out the ancilla register $a$ that is in the known state $\ket{0}_a$.
The gate complexity of $U_{\text{Schur}}$ is $\bigo{\poly{n, d, \log \frac{1}{\delta}}}$~\cite{Bacon2007},
for $\delta$ the precision.

We can now bound the necessary dimension to program reflections in a common format~\cite{Yang2020}.

\begin{corollary}\label{cor:reflectionProgramDimension}
    The gate family of programmable reflections has program dimension, $d_P$, bounded by
    \begin{equation}
        \label{eq:progdimforreflections}
        \log(d_p) \le (d-1) \log \paren*{\frac{\bigo*{d^{-1}}}{\epsilon}}.
    \end{equation}
\end{corollary}
\begin{proof}
The upper bound from \cref{thm:optimalDistance} that is attained by our algorithm, $\apxRefCh_{\psi, \apxRef}$,
implies $n = \bigo{\frac{1}{\epsilon}}$.
By applying $\decoder_n$ on the program register before applying $\apxRefCh_{\psi,R}$,
we can construct program states $\ket{\psi}^{\otimes n} \in \sym^n(\mathbb{C}^{d})$
using program states of dimension $\dim(\sym^n(\mathbb{C}^d))$.
Therefore, \cref{eq:symmetricCost} implies
\begin{align}
    \log_2(d_P) &\le (d-1)\log_2\paren*{\frac{n+d-1}{d-1}} + \bigo{d} \\
        &= (d-1) \log_2\paren*{\frac{\bigo{d^{-1}}}{\epsilon}}
\end{align}
as required.
\end{proof}

\section{Adjusting the rotation angle}\label{sec:unitaryreflections}
\AP We now generalize from reflections, $\rotChannel_\psi$, from the previous sections to rotations that have an arbitrary rotation angle, $\rotChannel_\psi(\alpha)$.
We first give an optimal algorithm that assumes access to all copies simultaneously.
We then analyze an algorithm due to \textcite{Lloyd2014}
that only needs copies of $\ket\psi$ provided one at a time.
We improve this algorithm by tweaking its parameters.

We can generalize our approach to finding the optimal approximate reflection channel (\cref{thm:optimalDistance}) to rotations $\mathcal{R}_{\psi}(\alpha)$ to obtain a general formula for the error of $\apxRefCh_{\psi,V}$ (keeping the coefficients $c_j$ abstract).

\begin{theorem}\label{thm:rotationBound}
    For any $\alpha \in [0,\pi]$, it holds that
    \begin{multline}
        \opnorm{\rotChannel_\psi(\alpha) - \apxRefCh_{\psi, V}} =
        \begin{dcases}
            2\paren*{1-\abs{c_0}^2} & \text{if } \frac{\abs*{\tilde c_0 \bar c_0 - e^{i\alpha}}}{1-\abs{c_0}^2} \le 1 \\ %
            \frac{2 \abs*{\tilde c_0 \bar c_{0} - e^{i\alpha}}^2}{2 \abs*{\tilde{c}_{0} \bar{c_{0}} - e^{i\alpha}} + \abs{c_{0}}^{2}-1} & \text{otherwise},
        \end{dcases}
    \end{multline}
    where the unitary $V = \sum_{\ell=0}^n c_{\ell} \cyclicPermutation^\ell$ for $c_\ell \in \mathbb C$.
\end{theorem}
\begin{proof}
    By \cref{lem:twirledStates,lem:normPOptimization},
    \begin{equation}\label{eq:pOptAlpha}
        \begin{split}
	    \opnorm{\rotChannel_\psi(\alpha) - \apxRefCh_{\psi, V}} &= \max_{p\in [0,1]} \norm{\text{Id}_R \otimes \paren*{\mathcal R_{\psi}(\alpha) - \apxRefCh_{\psi, V}} (\phi_p)}_{1} \\
            &= \max_{p\in [0,1]} (1-p) \paren*{1- {\abs{c_{0}}^{2}}} \\
            & \quad + \sqrt{(1-p)^2 \paren*{1 - \abs*{c_0}^2}^2 + 4(1-p)p\abs*{\tilde c_0 \overline{c_0} - e^{i\alpha}}^2}
        \end{split}
    \end{equation}
    Following the proof of \cref{thm:optimalDistance}, we first consider the domain $1 - \abs{c_0}^2 \ge \abs{\tilde c_0 \bar c_0 - e^{i \alpha}}$.
    Here \cref{eq:pOptAlpha} has no critical points with respect to $p$ so its maximum is obtained at $p=0$.
    The result for this case follows easily.

    Now consider the opposite domain where $1 - \abs{c_0}^2 < \abs{\tilde c_0 \bar c_0 - e^{i \alpha}}$.
    We maximize over $p$ by finding the critical point of \cref{eq:pOptAlpha} to obtain
    \begin{equation}
        p^* \coloneqq \frac{\abs*{\tilde{c}_{0}\bar{c_{0}} - e^{i \alpha}} - (1- \abs{c_{0}}^{2})}{2 \abs*{\tilde{c}_{0}\bar{c_{0}} - e^{i \alpha} } - \paren*{1-\abs{c_{0}}^{2}}}.
    \end{equation}
    We substitute $p^*$ into \cref{eq:pOptAlpha} to get
    \begin{multline}
        (1-p^*) \paren*{1- {\abs{c_{0}}^{2}}}
            + \sqrt{(1-p^*)^2 \paren*{1 - \abs*{c_0}^2}^2 + 4(1-p^*)p^* \abs*{\tilde c_0 \overline{c_0} - e^{i\alpha}}^2} \\
            = \frac{2 \abs*{\tilde c_0 \bar c_{0} - e^{i\alpha}}^2}{2 \abs*{\tilde{c}_{0} \bar{c_{0}} - e^{i\alpha}} + \abs{c_{0}}^{2}-1}
    \end{multline}
    as required.
\end{proof}

Although it would be desirable to follow up by minimizing the above diamond distance over all unitary elements of $\mathbb{C}[S_{n+1}]$, 
we were not able to successfully carry out an analysis along the lines of \cref{thm:optimalDistance} for values of $\alpha$ other than $\pi$,
due to the more complicated structure of the function to be optimized. 
To illustrate some of the complexity of optimizing \cref{thm:rotationBound},
consider the fact that the first domain in \cref{thm:rotationBound} does not exist for $\theta=0$ and $\alpha \neq 0$.
More generally, finding the existence and structures of either domain for a given algorithm $V$,
then carrying out the optimization in each domain,
is challenging for general $\alpha$.

\subsection{Processors that access the program register in parallel}
Let us consider the special class of algorithms $\apxRefCh_{\psi, R_\theta}$.
\Cref{thm:algorithm} gives the relation between $\theta$ and the unitary $V$ in the subalgebra of $\mathbb C[S_{n+1}]$
that is implemented.
In particular, when acting on $\ket{\phi}_S\ket{\psi}_P^{\otimes n}$,
$\apxRefCh_{\psi, \apxRef_\theta}$ has
\begin{align}
    c_0 &= \frac{n + e^{i \theta}}{n+1}, \\
    c_\ell &= \frac{e^{i \theta} - 1}{n+1}, & \ell=1,\ldots,n,
\end{align}
by \cref{lem:cyclicCoefficients}, so that $\tilde c_0 = e^{i\theta}$.
By using these values for $c_0$ and $\tilde c_0$,
we can obtain expressions for the diamond distance of $\apxRefCh_{\psi, \apxRef_\theta}$ from $\mathcal{R}_{\psi}(\alpha)$ using \cref{thm:rotationBound}.

\begin{figure}[tbh]
	\centering
	\includegraphics[alt={A graph plotting the dependence of the optimal theta-star on alpha. Theta-star exceeds alpha until pi/3 and then drops below alpha.}]{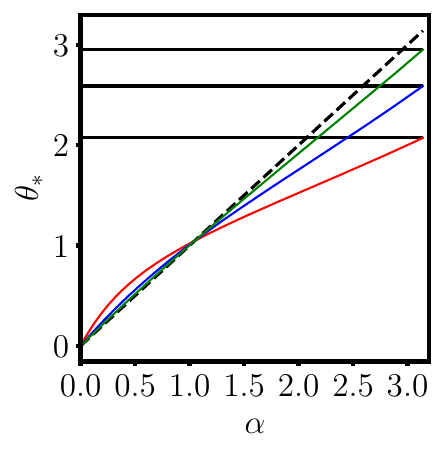}
	\caption{%
		The optimal angle $\theta_{*}$ that minimizes the diamond distance error between the rotation, $e^{i\alpha \ketbra{\psi}{\psi}}$ and our algorithm,
        $\apxRefCh_{\psi, \apxRef_\theta}$ based on \cref{fig:algorithm},
        for $n=1$ (red), $n=4$ (blue), $n=16$ (green).
        Solid black lines are $\cos^{-1}f(n)$ from \cref{thm:optimalDistance}.
        The black dashed line is the line $\theta_* = \alpha$.
	}\label{fig:diamond2}
\end{figure}

However, we still found it necessary to numerically optimize the angle $\theta$.
We plot the numerically optimal angle $\theta_*$ in \cref{fig:diamond2} that attains
\begin{equation}
    \min_{\theta \in [0, \pi]} \opnorm{\rotChannel_\psi(\alpha) - \apxRefCh_{\psi, \apxRef_\theta}}
\end{equation}
for increasing copy number $n$.
Perhaps surprisingly, $\theta_*$ deviates significantly from $\theta_* = \alpha$ for small $n$,
sometimes sitting above $\alpha$ and sometimes below $\alpha$.

The lack of analytical expressions notwithstanding,
it is of interest to consider particular values of $\theta$ and analyze the performance of $\apxRefCh_{\psi,R_{\theta}}$ for approximately implementing $\mathcal{R}_{\psi}(\alpha)$. 
If we make the natural choice $\theta = \alpha$, we can analytically compute a simpler expression for the diamond distance.
\begin{corollary}\label{cor:rotationDistance}
    For any $\alpha \in [0, \pi]$,
    \begin{equation}
        \opnorm{\rotChannel_\psi(\alpha) - \apxRefCh_{\psi, \apxRef_\alpha}} = \begin{dcases}
            \frac{4n\paren*{1 - \cos(\alpha)}}{(n+1)^2} & \text{if } \alpha \ge 2 \sin^{-1}\paren*{\frac{n+1}{2n}} \\
            \frac{2}{(n+1)\csc\paren*{\frac{\alpha}{2}} -n} & \text{otherwise.}
        \end{dcases}
    \end{equation}
\end{corollary}
\begin{proof}
    Consider the expression for $\opnorm{\rotChannel_\psi(\alpha) - \apxRefCh_{\psi, \apxRef_\alpha}}$ from \cref{thm:rotationBound}
    when $\theta = \alpha$.
    We first note that the condition $1-\abs{c_0}^2 \ge \abs{\tilde c_0 \bar c_0 - e^{i\alpha}}$
    can be simplified to
    \begin{equation}
        n\geq (n+1) \sin \left(\frac{\alpha }{2}\right)+n \cos (\alpha ).
    \end{equation}
    Therefore, we are in the first case when
    \begin{equation}
        \alpha \ge 2 \sin^{-1}\paren*{\frac{n+1}{2n}}.
    \end{equation}
    Now, the diamond distance in the first case can be simplified to
    \begin{equation}
        2\paren*{1-\abs{c_0}^2} = \frac{4n \paren*{1 - \cos(\alpha)}}{(n+1)^2}.
    \end{equation}
    In the other domain, we simplify the equation to obtain
    \begin{equation}
        \frac{2 \abs*{\tilde c_0 \bar c_{0} - e^{i\alpha}}^2}{2 \abs*{\tilde{c}_{0} \bar{c_{0}} - e^{i\alpha}} + \abs{c_{0}}^{2}-1} = \frac{2}{(n+1)\csc\paren*{\frac{\alpha}{2}} -n}
    \end{equation}
    as required.
\end{proof}

We show an upper bound of \cref{cor:rotationDistance} that is linear $\alpha$ and will be useful for universal programming (\cref{sec:universal}).
\begin{lemma}\label{lem:linearBound}
    It holds that
    \begin{align}
    \label{eq:diamdistboundalpha}
        \opnorm{\mathcal R_\psi(\alpha) - \apxRefCh_{\psi, \apxRef_\alpha}} \le \frac{3 \alpha}{n}\;.
    \end{align}
\end{lemma}
\begin{proof}
We consider the first case of \cref{cor:rotationDistance} and observe
\begin{equation}
    \frac{4n(1-\cos(\alpha))}{(n+1)^2} \le \frac{4(1-\cos(\alpha))}{n} \le \frac{3 \alpha}{n}
\end{equation}
where both inequalities hold for $\alpha \in [0,\pi]$.

In the other case of \cref{cor:rotationDistance}, we observe
\begin{align} 
    \frac{2}{(n+1)\csc(\alpha/2)-n} &\le \frac{2 \sin \frac{\alpha}{2}}{n\left(1-\sin \frac{\alpha}{2}\right)} \\
    & \leq \frac{\alpha}{n\left(1-\sin \frac{\alpha}{2}\right)} \\
    &\le \frac{\alpha}{n(1-\frac{n+1}{2n})} \\
    &= \frac{2 \alpha}{n}\; ,
\end{align}
where in the third line we inserted the largest value of $\alpha$ in the interval and in the last line assumed $n>1$.
The case of $n=1$ can be treated separately to show the same bound applies. 
\end{proof}

\subsection{Processors that access the program register sequentially}
\AP Although not originally proposed for the purpose of approximate programming of unitary operations,
\textcite{Lloyd2014} proposed an algorithm that can be used to approximately program a rotation~\cite{Kimmel2017}.
Decomposing the program register $P$ into $P_{1}\cdots P_{n}$,
and defining the unitary $\intro*\swaprot_k(\theta) \coloneqq e^{i\theta \text{SWAP}_{SP_k}}$,
this algorithm is given by
\begin{equation}\label{eqn:lmrlmr}
    \mathcal E_{\alpha/n}(\phi) \coloneqq \ptrace*{P}{\swaprot_n\paren*{\frac{\alpha}{n}} \cdots \swaprot_1\paren*{\frac{\alpha}{n}}  \paren*{\vphantom{\frac{\alpha}{n}} \phi\otimes \psi^{\otimes n}} \swaprot_1^\dagger\paren*{\frac{\alpha}{n}} \cdots \swaprot_n^\dagger\paren*{\frac{\alpha}{n}}}
\end{equation}
and has $\opnorm{\mathcal R_{\psi}(\alpha) - \mathcal E_\alpha} = \bigo*{\frac{\alpha^{2}}{n}}$,
according to additivity of the errors incurred at each step~\cite{Kimmel2017}.
We generalize to allow for an adjustable angle $\theta_i$,
for $i \in [n]$,
in the $\swaprot_i$ operators and denote it by $\mathcal E_\theta$.
We can simplify this circuit to only act on two registers at a time,
resulting in
\begin{equation}\label{eq:lmrAlgorithm}
	\begin{quantikz}
		\lstick{$\ket{\phi}$}           & \gate[2]{\swaprot_1(\theta_1)} &&& \gate[2]{\swaprot_2(\theta_2)} & & \ \ldots \ & & \gate[2]{\swaprot_n(\theta_n)} &  \\
		\lstick{\ket{\psi}} && \trash{}         &  \lstick{\ket{\psi}}\wireoverride{n}                      && \trash{} & \ \ldots \ \wireoverride{n}  &  \lstick{\ket{\psi}}\wireoverride{n}                                       && \trash{}
	\end{quantikz}
\end{equation}
Here we write a downward arrow to indicate tracing out the register
and follow it by initializing the state $\ket{\psi}$ in the same register.
Therefore, $\mathcal E_\theta$ may be considered a sequentially programmed approximate rotation
since it has access to only one copy of $\ket\psi$ at a time.

To illustrate the relationship between $\mathcal E_{\alpha/n}$ and the optimal algorithm among unitary elements of $\mathbb{C}[S_{n+1}]$ shown in \cref{fig:diamond2}, we consider the special case $n=1$ when the program is a single copy of the state $\psi$.
In this case, any unitary $V$ in the linear combination of permutations $\mathbb C[S_2]$ has the form
\begin{align}
    V = e^{i (\theta/2)\text{SWAP}},
\end{align}
up to a global phase. The half-angle convention shows that $V$ is implemented by the algorithm $\apxRef_\theta$ in Fig.\ref{fig:algorithm}.
We numerically optimize $\theta$ for each $\alpha$ to obtain $\theta_*(\alpha)$, the red curve in \cref{fig:diamond2}.
When $n=1$ it is impossible to implement a rotation with a large angle $\alpha$ with small error.
Thus, the interesting (i.e., low error) regime for $n=1$ is $\alpha \ll 1$.
In this regime, we have that $\theta_*(\alpha) \approx 2\alpha$ and thus
\begin{align}
      \swaprot_1(\theta_*(\alpha)) = e^{i \alpha \text{SWAP}} \; 
\end{align}
is optimal to leading order in $\alpha$.
Hence, we can interpret $\mathcal E_{\alpha/n}$ as attempting to implement $e^{i \alpha \psi /n}$ with minimal diamond distance using a single copy program and repeating this process $n$ times.

We now show that $\mathcal E_\theta$ is an "approximate reflection channel".

\begin{lemma}\label{lem:lmrApxRef}
    The algorithm $\mathcal E_\theta(\phi)$ has the form $\mathcal E_\theta(\phi) = \mathrm{tr}_{P_{1}\cdots P_{n}}V\phi \otimes \psi^{\otimes n}V^{\dagger}$, where $V = \sum_{\ell =0}^n c_\ell \cyclicPermutation^\ell$ is an "approximate reflection channel" with coefficients
    \begin{align}
        c_0 &= \prod_{j=1}^{n} \cos\theta_j \\
        c_1 &= \left(\prod_{k=2}^n e^{i \theta_k}\right) (i \sin\theta_1 ) \\
        c_l &= \left(\prod_{k=l+1}^n e^{i \theta_k}\right) (i \sin\theta_l ) \left(\prod_{j=1}^{l-1} \cos\theta_j \right), & \text{for } 2\le \ell \le n-1 \\
        c_{n}&= i\sin\theta_{n}\left(\prod_{j=1}^{n-1} \cos\theta_j \right).
    \end{align}
\end{lemma}
\begin{proof}
We start by observing that
\begin{align}
    \prod_{k=1}^n  e^{i \theta_k  \text{SWAP}_{S P_k}} = \prod_{k=1}^n \left( \cos(\theta_k) \mathbf{1} + i \sin(\theta_k) \text{SWAP}_{S P_k} \right)\; .
\end{align}
As can be seen from the circuit, only the first SWAP operator acts nontrivially on the state. All the following SWAP operators leave the state invariant as they SWAP two identical states $\ket{\psi}$.
Thus, we can write the action of this operator on this state as
\begin{align}
    \prod_{k=1}^n  &e^{i \theta_k \text{SWAP}_{S P_k}}\ket{\phi}_{S} \otimes \ket{\psi}^{\otimes n}_{P} \\
    &= \prod_{k=1}^n \left( \cos(\theta_k) \mathbf{1} + i \sin(\theta_k) \text{SWAP}_{S P_k} \right) \ket{\phi}_{S} \otimes \ket{\psi}^{\otimes n}_{P} \\
    &= \left(\prod_{j=1}^{n} \cos(\theta_j) \right) \ket{\phi}_{S} \ket{\psi}^{\otimes n}_{P} +  \left(\prod_{k=2}^n e^{i \theta_k}\right) (i \sin(\theta_1))  \ket{\psi}_{S} \ket{\phi}_{P_{1}} \ket{\psi}^{\otimes n-1}_{P\setminus P_{1}} \nonumber \\
    &\quad +  \sum_{\ell=2}^{n-1} \left(\prod_{k=\ell+1}^n e^{i \theta_k}\right) (i \sin(\theta_\ell)) \left(\prod_{j=1}^{\ell-1} \cos(\theta_j) \right) \ket{\psi}^{\otimes \ell}_{SP_{1}\cdots P_{\ell-1}} \ket{\phi}_{P_{\ell}} \ket{\psi}^{\otimes n-\ell}_{P_{\ell+1}\cdots P_{n}}\nonumber \\
    &\quad + i \sin(\theta_n) \left(\prod_{j=1}^{n-1} \cos(\theta_j) \right) \ket{\psi}^{\otimes n-1}_{SP_{1}\cdots P_{n-1}}\ket{\phi}_{P_{n}}.
\end{align}
The parameters $c_\ell$ can be read from the above equation. The Fourier component $\tilde{c}_{0} = e^{i\sum_{j=1}^{n}\theta_{j}}$ is obtained from the sum of the $c_{\ell}$.
\end{proof}

We show that $\mathcal E_\theta$,
which has sequential access to the program copies,
necessarily has a larger error than the optimal algorithm that has access to all program copies simultaneously.
\begin{theorem}
    It holds that
    \begin{equation}
        \min_{V \in \mathbb C[S_{n+1}]} \opnorm{\mathcal R_\psi(\alpha) - \apxRefCh_{\psi, V}} < \min_{\theta \in [0,2\pi)^n} \opnorm{\mathcal R_\psi(\alpha) - \mathcal E_\theta}.
    \end{equation}
\end{theorem}
\begin{proof}
\cref{thm:optimalDistance} shows that the optimal choice of $c_l$'s is one in which $c_l=c_{l'}$ for $l,l'\neq 0$.
It is easy to see that this cannot be achieved in this setting of sequential program use since
\begin{equation}
    \frac{c_{l+1}}{c_l} = e^{-i \theta_l}\frac{\cos(\theta_l)\sin(\theta_{l+1})}{\sin(\theta_l)} \neq 1,
\end{equation}
proving the result.
\end{proof}

Although the error of $\mathcal E_{\frac{\alpha}{n}}$ exhibits asymptotically optimal scaling, there exists an angle $\theta' \le \frac{\alpha}{n}$
that has strictly better performance for any finite $n > 2$.
\begin{theorem}
    Assume $n > 2$
    and let the angle $\theta' \coloneqq \frac{\alpha}{n+\frac{\alpha\sqrt 3}{2}}$,
    then
    \begin{equation}
        \opnorm{\rotChannel_\psi(\alpha) - \mathcal E_{\frac{\alpha}{n}}} - \opnorm{\rotChannel_\psi(\alpha) - \mathcal E_{\theta '}} = \frac{2\sqrt 3 \alpha^3}{n^2} + \bigo*{\frac{1}{n^3}} 
    \end{equation}
\end{theorem}
\begin{proof}
By \cref{lem:lmrApxRef}, we know the parameters $c_0$ and $\tilde c_0$ of $\mathcal E_{\frac{\alpha}{n}}$
in \cref{lem:normPOptimization}
and evaluate its distance to $\mathcal R_\psi(\alpha)$.
The algorithm $\mathcal E_{\frac{\alpha}{n}}$ satisfies $1-\vert c_{0}\vert^{2} > \vert \tilde{c}_{0}\overline{c_{0}} - e^{i\alpha}\vert$, and as a result
the algorithm is in Domain A as defined in proof of \cref{thm:optimalDistance},
so the maximization appearing in the diamond norm is obtained for $p=0$.
This gives the diamond distance
\begin{equation}\label{eq:equirotationDistance}
    \opnorm{\rotChannel_\psi(\alpha) - \mathcal E_{\frac{\alpha}{n}}} = 2\paren*{1-\cos^{2n} \frac{\alpha}{n}} = \frac{2\alpha^{2}}{n}- \frac{\alpha^4}{n^2}+\bigo*{\frac{1}{n^3}}
\end{equation}
by a series expansion at $n\to \infty$.

The error can be further reduced by choosing a smaller $\theta(n)$ with the aim of saturating the inequality $1- \abs{c_{0}}^{2} > \abs{\tilde{c}_{0}\overline{c_{0}} - e^{i\alpha}}$, noting that the condition of equality defines the boundary between Domain A and Domain B in the proof of \cref{thm:optimalDistance}.
By postulating an inverse power law for $\theta(n)$ (thereby neglecting multiplicative logarithmic corrections) and demanding that $\opnorm{\mathcal R_\psi(\alpha) - \mathcal E_{\theta(n)}}\rightarrow 0$ as $n\rightarrow \infty$, one concludes that $\theta(n)= \alpha/(n+x(\alpha))$ is the correct general form, where $x(\alpha)$ is a constant independent of $n$. However, only for $x(\alpha) = {\alpha\sqrt{3}\over 2}$ is the diamond distance asymptotically lower than  $\opnorm{\mathcal R_\psi(\alpha) - \mathcal E_{{\alpha\over n}}}$.
In particular, taking $\theta':= \frac{\alpha}{n+\frac{\alpha\sqrt 3}{2}}$ , the diamond distance is
\begin{equation}
    \opnorm{\rotChannel_\psi(\alpha) - \mathcal E_{\theta'}} = \frac{2\alpha^{2}}
{n}-\frac{\alpha^4+2 \sqrt{3} \alpha^3}{n^2}+ \bigo*{\frac{1}{n^3}}.
\end{equation}
Subtracting this from \cref{eq:equirotationDistance} gives the result.
\end{proof}

\section{Lower bound \label{sec:lower}}
In this section, we lower bound the dimension of "program register" that is required to obtain a target accuracy. 
We denote a program state by $\psi_{R}$ encoding a reflection $R$ for the quantum processor $\mathcal C$.
Let an $\epsilon$-approximate reflection channel be specified by
\begin{equation}
	\mathcal E_{R}(\cdot) \coloneqq \ptrace{P}{\mathcal C(\cdot \otimes \psi_{R})},
\end{equation}
for any channel $\mathcal C$,
if
\begin{equation}
	\frac{1}{2} \opnorm{\mathcal R - \mathcal E_{R}} \le \epsilon
\end{equation}
for all reflections $R$, where $\mathcal R(\cdot) \coloneqq R (\cdot) R$.

\AP The ""von Neumann entropy"" (or just "entropy") of a state $\rho$ is defined as
\begin{equation}
	\entropy{\rho} \coloneqq - \trace{\rho \log_2 \rho}.
\end{equation}
Let the ""Holevo information"" (or just "information") be defined for an ensemble of states $\set{\rho_x, dx}_{x \in X}$
\begin{equation}
	\information{\set{\rho_x, dx}} \coloneqq \entropy*{\int_{x\in X} dx \rho_x} - \int_{x\in X} dx \entropy{\rho_x},
\end{equation}
which equals the "entropy" for a pure state ensemble.
By the data processing inequality, the "information" may only decrease by a quantum channel,
therefore, the dimension of the "program register" will upper bound the "information" of any implementation of a dynamic reflection. 
This is the basic idea behind the lower bound technique described in \cite{Yang2020}, which we adapt to our setting.

In the following, it is useful to parametrize a reflection channel $\mathcal R$ by a corresponding pure state $\psi$ as
\begin{equation}
	\mathcal R_\psi(\cdot) \coloneqq e^{i \pi \ketbra{\psi}{\psi}} (\cdot) e^{i\pi\ketbra{\psi}{\psi}},
\end{equation}
where we assume that the reflection is given by 
\begin{equation}
	e^{i \pi \psi} = \idm - 2 \psi.
\end{equation}
Equivalently, we can parametrize the pure state as $\ket{\psi} = U \ket{d}$ for some unitary $U \in U(d)$. 
The dimension of the program register $d_P$ is defined as the dimension of the support of a uniform ensemble (over reflections) encoded in any program state according to $R\mapsto \Psi_{R}$ \cite{Yang2020} 
\begin{equation}
    d_P \coloneqq |\mathrm{supp}(\set{\Psi_R}_{R \sim \mathrm{Haar}})|
\end{equation}

Given a positive $n \in \mathbb N$,
Schur-Weyl duality tells us that there exists an isometry such that, for the $d$-dimensional Hilbert space $\mathcal H$,
\begin{equation}
	\mathcal H^{\otimes n} \simeq \bigoplus_{\lambda \vdash_d n} W_\lambda \otimes S_{\lambda},
\end{equation}
where $\lambda \vdash_d n$ denotes the set of partitions of $n$ of length at most $d$,
each $W_\lambda$ is an irrep of $U(d)$ indexed by $\lambda$ with dimension $d_\lambda$,
and $S_{\lambda}$ is the multiplicity subspace.
We define $\mathcal T : =\left(\bigoplus_{\lambda \vdash_{d}n}\text{tr}_{S_{\lambda}} \right) \circ \mathcal{U}_{\text{Sch}}$ as the channel that first incorporates the Schur-Weyl isometry
and then discards the multiplicity subsystems $\set{ S_{\lambda}}_\lambda$~\cite{Yang2020}.
The resulting dimension is
\begin{equation}
	d_n \coloneqq \sum_{\lambda \vdash_d n} d_\lambda^2 = \binom{n+d^2-1}{d^2-1}.
\end{equation}

The lower bound method of \cite{Yang2020} can be summarized in the following lemma.

\begin{lemma}\label{lemma:lb}
	The program dimension of any $\epsilon$-approximate reflection channel is lower bounded by
	\begin{equation}
		\log_2 d_P \ge \information{\set{\mathcal{T} \circ \mathcal{R}_{\psi}^{\otimes n}(\Phi_n) }_{\mathrm{Haar}}} -4n\sqrt{2\epsilon}\log_2 d_{n} -1,
	\end{equation}
    where $\Phi_n$ is an arbitrary probe state.
\end{lemma}
\begin{proof}
	Follows from~\cite[Supp. Matt.~(28)]{Yang2020} by substituting the $\epsilon$-universal processor with $\mathcal E_{R}$
	and the exact implementation with $\mathcal R_\psi$ in
        \begin{align}
            \log d_P  &\ge \information{\set{\psi_{R}}_{R \sim \mathrm{Haar}}} \\
            &= \information{\set{\psi_{R} \otimes \Phi_n}_{R \sim \mathrm{Haar}}} \\
            &\ge \information{\set{\mathcal{M}_{\psi}(\Phi_n)}_{\mathrm{Haar}}} \\
            &\ge \information{\set{\mathcal{R}_{\psi}^{\otimes n}(\Phi_n) }_{\mathrm{Haar}}} - 4n\sqrt{2\epsilon}\log d_{n} - 1 \\
            &= \information{\set{\mathcal{T}\circ \mathcal{R}_{\psi}^{\otimes n}(\Phi_n) }_{\mathrm{Haar}}} -4n\sqrt{2\epsilon}\log d_{n} -1.
            \label{eqn:lbchain}
        \end{align}
    where $\mathcal{M}_{\psi}(\cdot)$ is defined in~\cite[Equation (A16)]{Yang2020} but for a reflection $\reflection{\psi} = \idm - 2 U \proj{0} U^\dagger$ with Haar random $U$.
\end{proof}

Therefore, in order to find a useful lower bound, we need to choose a probe state $\Phi_n$ that maximizes  $\information{\set{\mathcal{T} \circ \mathcal{R}_{\psi}^{\otimes n}(\Phi_n) }_{\mathrm{Haar}}}$.
Equivalently, we would like to understand the spectrum of the mixed state
\begin{align}
    \rho &:= \int_{\mathrm{Haar}} d \psi \; \mathcal{R}_{\psi}^{\otimes n}(\Phi_n) \\
    &= \int_{\mathrm{Haar}} dU \of{ \of{URU^\dagger}^{\otimes n} \otimes I } \Phi_n \of{ \of{URU^\dagger}^{\otimes n} \otimes I }^{\dagger}.
\end{align}
Indeed, it is easy to see that
\begin{equation}
    \entropy{\rho} = \entropy{\mathcal{T}(\rho)} = \information{\set{\mathcal{T} \circ \mathcal{R}_{\psi}^{\otimes n}(\Phi_n) }_{\mathrm{Haar}}}.
\end{equation}
We need to introduce some representation theoretic lemmas first.

\subsection{Preliminaries}

For the subsequent analysis, note that the irreducible representations of $U(d)$ that occur in the physical space $\mathcal{H} := (\C^d)\xp{n} \otimes (\C^d)\xp{n}$ can be analyzed in several equivalent frameworks, which can be chosen according to the application. 
For example, if one is mainly interested in maintaining the bipartition between two systems of $n$ qudits and applying the same unitary query to both systems, the two-sided Schur--Weyl duality (i.e., Schur--Weyl duality in each system) provides a useful decomposition. 
This approach treats each part as tensor product of $n$ defining irreps of $U(d)$.
By contrast, if the application involves applying a unitary to one system and its complex conjugate to the second system, the basis of mixed Schur--Weyl duality is the most relevant decomposition. 
Another choice would be treating a second part of $\mathcal{H}$ as tensor product of $n$ dual defining irreps and using the dual Schur transform.  
Dual of the defining irrep could also be analyzed in an inefficient way by considering each of them as $(1,...,1,0) \otimes \det^{-1}$ and using the Pieri rule to decompose the whole tensor product. 
By contrast, we treat the dual of the defining irrep by its highest weight as $(0,...,0,-1)$ in each register where it occurs.
We show the isomorphism structure of these equivalent viewpoints below indicating the respective basis transformations unitaries:
\begin{align} \label{eq:sw_dualities}
    \mathcal{H} &\cong^{ U_{\text{Sch}}\otimes U_{\text{Sch}} } 
    \of[\bigg]{\bigoplus_{\lambda \in \Airreps{d}{n,0}} W_{\lambda} \otimes S_{\lambda}} \otimes
    \of[\bigg]{\bigoplus_{\lambda^\prime \in \Airreps{d}{n,0}} W_{\lambda^\prime} \otimes S_{\lambda^\prime}} \\
    &\cong^{ U_{\text{Sch}} \otimes U_{\text{dSch}} } 
    \of[\bigg]{\bigoplus_{\lambda \in \Airreps{d}{n,0}} W_{\lambda} \otimes S_{\lambda}} \otimes
    \of[\bigg]{\bigoplus_{\overline{\lambda}^\prime \in \Airreps{d}{0,n}} W_{\overline{\lambda}^\prime} \otimes S_{\lambda^\prime}} \\
    &\cong^{ U_{\text{mSch}} }  \bigoplus_{\nu \in \Airreps{d}{n,n}} W_{\nu} \otimes B_{\nu},
    \label{eqn:swlast}
\end{align}
where $S_{\lambda}$ is an irrep of the symmetric group $\mathrm{S}_n$ and $B_{\nu}$ is an irrep of the partially transposed permutation matrix algebra, which we denote by $\A{d}{n,n}$ \cite{go,gbo}. 
In particular $\A{d}{n,0}$ is a standard tensor representation of the symmetric group algebra $\C[\mathrm{S}_n]$ on $(\C^{d})\xp{n}$. 
$U_{\text{Sch}}$ is the standard Schur transform, $U_{\text{dSch}}$ is the dual Schur transform, and $U_{\text{mSch}}$ is the mixed Schur transform.

The last isomorphism \cref{eqn:swlast} can be understood by taking account of the structure of the multiplicity spaces at each step as follows:
\begin{align}
    \of[\bigg]{\bigoplus_{\lambda \in \Airreps{d}{n,0}} & W_{\lambda} \otimes S_{\lambda}} \otimes
    \of[\bigg]{\bigoplus_{\overline{\lambda}^\prime \in \Airreps{d}{0,n}} W_{\overline{\lambda}^\prime} \otimes S_{\lambda^\prime}} \\
    &\cong \bigoplus_{\lambda, \lambda^\prime \in \Airreps{d}{n,0}} \of*{W_{\lambda} \otimes W_{\overline{\lambda}^\prime}} \otimes \of*{S_{\lambda} \otimes S_{\lambda^\prime}} \\
    &\cong \bigoplus_{\lambda, \lambda^\prime \in \Airreps{d}{n,0}} \of[\bigg]{ \bigoplus_{\nu \in \Airreps{d}{n,n}} W_{\nu} \otimes \C^{r^\nu_{\lambda \overline{\lambda}^\prime}}} \otimes \of*{S_{\lambda} \otimes S_{\lambda^\prime}} \\
    &\cong \bigoplus_{\nu \in \Airreps{d}{n,n}} W_{\nu} \otimes \of[\bigg]{\bigoplus_{\lambda, \lambda^\prime \in \Airreps{d}{n,0}} \C^{r^\nu_{\lambda \overline{\lambda}^\prime}} \otimes S_{\lambda} \otimes S_{\lambda^\prime}} \\
    &\cong \bigoplus_{\nu \in \Airreps{d}{n,n}} W_{\nu} \otimes B_{\nu},
    \label{eqn:detailedtransforms}
\end{align}
where $r^\nu_{\lambda \overline{\lambda}^\prime}$ denote the multiplicity of the irrep $\nu$ inside the tensor product $\lambda \otimes \overline{\lambda}^\prime$.
Now we need to state several useful properties of the Gelfand--Tsetlin basis for the Weyl modules $W_\lambda$.

\begin{lemma}[Equations (13), (14) of \cite{bairdConj1964}] \label{lemma:conj_GT_basis}
    Let $W_\lambda$ be a Weyl module for the unitary group $U(d)$ with $\lambda \in \Airreps{d}{n,0}$ equipped with specified Gelfand-Tsetlin basis $\set{\ket{L} \given L \in \GT{\lambda}}$.
    Consider now the dual Weyl module $W_{\bar{\lambda}}, \, \bar{\lambda} \in \Airreps{d}{0,m}$ with the action on the same vector space $W_\lambda$.
    Then the GT basis $\set{\ket{\bar{L}} \given \bar{L} \in \GT{\bar{\lambda}}}$ of $W_{\bar{\lambda}}$ is related to the GT basis of $W_\lambda$ via
    \begin{equation}
        \ket{L} = (-1)^{\varphi(L)}\ket{\bar{L}},
    \end{equation}
    where the $\varphi(L)$ is defined as follows
    \begin{align}
        \varphi(L) &:= \varphi_{d-1}(L) - \varphi_{d-1}(L_\lambda), \\
        \varphi_k(L) &:= \sum_{j = 1}^{k} \sum_{i = 1}^{j} L_{i,j} \quad \text{for any} \quad k \in [d] \\ 
        (L_\lambda)_{i,j} &:= \lambda_i \quad \text{for every} \quad j \in [d], \, i \in [j].\label{eqn:maxgtpat}
    \end{align}
\end{lemma}

\begin{lemma}[Equation (5), page 363 of \cite{vilenkin}]\label{lemma:reflection_GT_basis}
    Let $R \in U(d)$ be the reflection operator $R = I - 2 \proj{d}$ where $\ket{d}$ is a basis state of the defining representation of $U(d)$ which is annihilated by all operators $E_{ij}:=\ket{i}\bra{j}$, $i>j$. Then for any irrep specified by a highest weight $\lambda$, the representation matrix $R_\lambda$ in the Gelfand--Tsetlin basis looks like
    \begin{equation}
        R_\lambda = \sum_{L \in \GT{\lambda}} (-1)^{w_d(L)} \proj{L},
    \end{equation}
    where $w(L) \defeq (w_1(L),\dotsc,w_d(L))$ is the weight of the Gelfand--Tsetlin pattern $L$, defined for every $j \in [d]$ as
    \begin{equation}
        w_j(L) \defeq \sum_{i=1}^{j} L_{i,j} - \sum_{i=1}^{j-1} L_{i,j-1},
    \end{equation}
    with the convention $w_1(L) \defeq L_{1,1}$.
\end{lemma}

We also need to recall several properties of general Clebsch--Gordan (CG) coefficients.

\begin{lemma}\label{lemma:magic}
    Let $d \geq 2$. Then for any irrep $\lambda \in \Airreps{d}{n,0}$ and any $\nu \in \Airreps{d}{n,n}$, $N \in \GT{\nu}$, $t \in [r_{\lambda,\nu}]$  where $r_{\lambda,\nu} \defeq r^{\nu}_{\lambda,\bar{\lambda}}$, we have
    \begin{equation}
        \sum_{L \in \GT{\lambda}} (-1)^{\varphi(L)} C^{N,t}_{L \bar{L}} = \delta_{\nu,\varnothing} \delta_{N,0} \sqrt{d_\lambda}.
    \end{equation}
    Note that $r_{\lambda,\varnothing} = 1$ always, and there is only one trivial GT pattern $N = 0$ for $\nu = \varnothing$.
\end{lemma}
\begin{proof}
    The state ${1 \over \sqrt{d_{\lambda}}} \sum_{L \in \GT{\lambda}} (-1)^{\varphi(L)}\ket{L}\ket{\overline{L}}$ is the only state invariant under $U_{\lambda \otimes \bar{\lambda}}$ for all $U$, where $\bar{\lambda} \in \Airreps{d}{0,n}$. 
    Over all $\nu \in \Airreps{d}{n,n}$ such that $\nu$ is a subrepresentation of $\lambda \otimes \bar{\lambda}$, the only state that is invariant under $U_{\nu}$ is $\ket{N=0}$ which occurs specifically for $\nu = \varnothing$ (multiplicity $t=1$). 
    These states are related by a Clebsch--Gordan transform $\CG_{\lambda,\bar{\lambda}}$ carrying out the transformation from the second to the third line in (\ref{eqn:detailedtransforms}) in the $\lambda\otimes \bar{\lambda}$ sector, and the relevant matrix elements of this unitary are $C^{0}_{L \bar{L}}$.
\end{proof}

\begin{fact}[Equation (1), page 368 of \cite{vilenkin}]\label{lemma:reduced_cg}\label{fact:RCG}
    Any Clebsch--Gordan coefficient $C^{N,t}_{L M}$ where of $U(d)$ can be reduced to Clebsch--Gordan coefficients $C^{N^{(d-1)},r}_{L^{(d-1)} M^{(d-1)}}$ of $U(d-1)$ via \emph{reduced Clebsch--Gordan} coefficients\footnote{\emph{Reduced Clebsch--Gordan coefficients} are sometimes also called \emph{reduced Wigner coefficients} or \emph{$G_1$-scalar factors}} $RC^{L_d,M_d,N_d,t}_{L_{d-1},M_{d-1},N_{d-1},r}$ as 
    \begin{equation}
        C^{N,t}_{L M} = \sum_{r} RC^{L_d,M_d,N_d,t}_{L_{d-1},M_{d-1},N_{d-1},r} \cdot C^{N^{(d-1)},r}_{L^{(d-1)} M^{(d-1)}},
    \end{equation}
    where $N_d$ is the $d$-th row of the Gelfand--Tsetlin pattern $N$, i.e. it is a highest weight for unitary group $U(d)$, and $N^{(d-1)}$ is the Gelfand--Tsetlin pattern obtained from $N$ by keeping bottom $d-1$ rows, so that $N^{(d-1)}$ corresponds to some basis vector of the irrep $N_{d-1}$ of the group $U(d-1)$.
\end{fact}

\subsection{Optimization of the Holevo information}

As a probe state for maximizing the Holevo information of the ensemble of reflections, we take the following state $\ket{\Phi_{n}} \in \mathcal{H}$ written in the double-sided Schur basis:
\begin{equation}
    U_{\text{Sch}}\otimes U_{\text{Sch}} \ket{\Phi_{n}} := \sum_{\lambda \in \Airreps{d}{n,0}} \sqrt{q_\lambda} \ket{\Phi^{+}_{\lambda}} \otimes \ket{\psi_\lambda}
    \label{eqn:rys}
\end{equation}
where $\ket{\psi_\lambda} \in S_{\lambda} \otimes S_{\lambda}$ are some pure states in the multiplicity registers corresponding to the symmetric group irreps, $\set{q_\lambda}_\lambda$ is some probability distribution to be optimized later. 
For example, the specific choice $q_\lambda := d^2_{\lambda} / d_{m}$ with $d_{m}:=\sum_{\lambda \in \hat{\mathcal{A}}^{d}_{n,0}}d_{\lambda}^{2}$ corresponds to YRC state \cite[Equation (A24)]{Yang2020}. A maximally entangled state $\ket{\Phi^{+}_{\lambda}} \in W_{\lambda} \otimes W_{\lambda}$ is defined as
\begin{equation}
    \ket{\Phi^{+}_{\lambda}} := {1\over \sqrt{d_{\lambda}}}\sum_{L \in \GT{\lambda}}\ket{L} \ket{L},
    \label{eqn:eprrep}
\end{equation}
where $\GT{\lambda}$ is the set of Gelfand-Tsetlin basis labels for $\lambda$, i.e. $L$ is a Gelfand-Tsetlin pattern of shape $\lambda$. 

Note that the ricochet trick works for any $\lambda$ and any matrix $X_\lambda$ acting on Weyl module $W_\lambda$:
\begin{equation}
    X_\lambda \otimes \idm_\lambda \ket{\Phi^{+}_{\lambda}} = \idm_\lambda \otimes X_\lambda\tp \ket{\Phi^{+}_{\lambda}}.
\end{equation} 
Therefore, in the double Schur basis, we can write the state reflected by $URU^{\dagger}$ in the first $(\C^d)\xp{n}$ register as
\begin{multline}
    \of*{U_{\text{Sch}} \otimes U_{\text{Sch}}} \of*{\of{URU\ct}^{\otimes n} \otimes \idm }\ket{\Phi_n} = \sum_{\lambda \in \Airreps{d}{n,0}} \sqrt{q_\lambda} \of*{U_{\lambda} R_{\lambda} U_\lambda\ct  \otimes \idm_\lambda \ket{\Phi^{+}_{\lambda}}} \otimes \ket{\psi_\lambda} \\
   = \sum_{\lambda \in \Airreps{d}{n,0}} \sqrt{q_\lambda} \of*{\of*{U_{\lambda} \otimes \bar{U}_{\lambda}} \of*{R_{\lambda} \otimes \idm_\lambda} \ket{\Phi^{+}_{\lambda}}} \otimes \ket{\psi_\lambda},
\end{multline}
which implies that the ensemble state is
\begin{align}
    \rho = \int dU \of{ U^{\otimes n} \otimes \bar{U}^{\otimes n} } \of{ R^{\otimes n} \otimes \idm } \Phi_n \of{ R^{\otimes n} \otimes \idm } \of{ U^{\otimes n} \otimes \bar{U}^{\otimes n} }\ct ,
\end{align}
so, in particular, $\rho$ in an element of algebra of partially transposed permutations.
Now using \cref{lemma:reflection_GT_basis} we can write
\begin{equation}
    R_\lambda  \otimes I_\lambda \ket{\Phi_{\lambda}} = \frac{1}{\sqrt{d_\lambda}}\sum_{L \in \GT{\lambda}} (-1)^{w_d(L)}  \ket{L}\ket{L}.
\end{equation}
We can ignore the $S_\lambda$ registers by applying $\mathcal{T}$ to $\rho$. Since $\mathcal{T}(\rho)$ is related by an isometry to $\rho$, $\mathcal{T}(\rho)$ will have the same spectrum as $\rho$. Therefore, we get in the double-sided Schur basis,
\begin{multline}
    \mathcal{T}(\rho) = \sum_{\lambda,\lambda'} \sqrt{\frac{q_\lambda q_{\lambda'}}{d_\lambda d_{\lambda'}}} \sum_{\substack{L \in \GT{\lambda} \\ L' \in \GT{\lambda'} }} (-1)^{w_d(L) + w_d(L')} \\
    \times \int dU \of*{U_\lambda \otimes \bar{U}_{\lambda}} \ket{L,L}\bra{L',L'} \of*{U_{\lambda'} \otimes \bar{U}_{\lambda'}}^{\dagger}.
\end{multline}
By using \cref{lemma:conj_GT_basis}, we get the following expression in the half-dual Schur basis
\begin{multline}
    \mathcal{T}(\rho) = \sum_{\lambda,\lambda'} \sqrt{\frac{q_\lambda q_{\lambda'}}{d_\lambda d_{\lambda'}}} \sum_{\substack{L \in \GT{\lambda} \\ L' \in \GT{\lambda'} }} (-1)^{w_d(L) + w_d(L') + \varphi(L')+\varphi(L)} \\
    \times \int dU \of*{U_\lambda \otimes U_{\bar{\lambda}}} \ket{L,\bar{L}}\bra{L',\bar{L'}} \of*{U_{\lambda'} \otimes U_{\bar{\lambda'}}}^{\dagger} .
\end{multline}
Next, we apply the Clebsch--Gordan transformation $\CG_{\lambda,\bar{\lambda}}$ to the half-dual Schur basis state $\ket{L,\bar{L}}$ with $L \in \GT{\lambda}$ to obtain:
\begin{equation}
    \CG_{\lambda,\bar{\lambda}} \ket{L,\bar{L}} = \sum_{\nu \in \Airreps{d}{n,n}} \sum_{t = 1}^{r_{\lambda,\nu}} \sum_{N \in \GT{\nu}} C_{L\bar{L}}^{N,t}\ket{N,t},
\end{equation}
which is in the mixed Schur basis. We have used a shorthand $r_{\lambda,\nu} \defeq r^{\nu}_{\lambda,\bar{\lambda}}$ and the following convention: $r_{\lambda,\nu} = 0$ implies that $\nu$ do not appear as irrep in the decomposition of $\lambda \otimes \bar{\lambda}$.
This allows us to rewrite
\begin{equation}
    \CG_{\lambda,\bar{\lambda}}  \of*{ U_\lambda \otimes U_{\bar{\lambda}}} \CG_{\lambda,\bar{\lambda}}\ct = \bigoplus_{\nu \in \Airreps{d}{n,n}} U_\nu \otimes I_{r_{\lambda, \nu}}.
\end{equation}
From the half-dual Schur basis for $\bigoplus_{\lambda \in \hat{A}^{d}_{n,0}}W_{\lambda}\otimes W_{\bar{\lambda}}$, we can go to the mixed Schur basis by applying a CG transform in each block as 
\begin{equation}
    \CG\, \defeq \bigoplus_{\lambda  \in \Airreps{d}{n,0}} \,\CG_{\lambda,\bar{\lambda}},
\end{equation}
and in this basis, we get the expression
\begin{multline}
    \widetilde{\mathcal{T}(\rho)} \defeq \CG \mathcal{T}(\rho) \CG\ct \\
    = \sum_{\lambda,\lambda'} \sqrt{\frac{q_\lambda q_{\lambda'}}{d_\lambda d_{\lambda'}}} \sum_{\substack{L \in \GT{\lambda} \\ L' \in \GT{\lambda'} }} (-1)^{w_d(L) + w_d(L') + \varphi(L')+\varphi(L)}  \\
    \sum_{\nu, \nu^\prime} \sum_{t = 1}^{r_{\lambda,\nu}} \sum_{t^\prime = 1}^{r_{\lambda', \nu^\prime}} \sum_{\substack{N \in \GT{\nu} \\ N^\prime \in \GT{\nu^\prime}}} 
    C^{N,t}_{L \bar{L}} C^{N^\prime,t^\prime}_{L' \bar{L'}} \int dU U_\nu \ket{N,t}\bra{N^\prime,t^\prime} U_{\nu^\prime}^{\dagger}.
    \label{eqn:refens1}
\end{multline}
In fact, $\widetilde{\mathcal{T}(\rho)}$ is a nice block-diagonal matrix due to the \emph{Grand Schur Orthogonality} relations
\begin{equation}
    \int dU \bra{X,t} U_\nu \ket{Y,t} \bra{Y^\prime,t^\prime} U_{\nu^\prime}\ct \ket{X^\prime,t^\prime} = \delta_{\nu,\nu^\prime} \delta_{X,X^\prime} \delta_{Y,Y^\prime} \frac{1}{d_\nu},
\end{equation}
so it can be expressed as
\begin{equation}
    \widetilde{\mathcal{T}(\rho)} = \sum_{\nu \in \Airreps{d}{n,n}} \idm_{d_\nu} \otimes M_\nu,
    \label{eqn:sumstate}
\end{equation}
where $M_\nu$ is a matrix of size $r_\nu \times r_\nu$, where $r_\nu := \sum_{\lambda \in \Airreps{d}{n,0}} r_{\lambda,\nu}$.
We will now label the rows and columns of $M_\nu$ by $(\lambda,t),(\lambda^\prime,t^\prime)$ where $\lambda, \lambda^\prime \in \Airreps{d}{n,0}$ and $t \in [r_{\lambda,\nu}], t^\prime \in [r_{\lambda^\prime, \nu}]$, and its matrix entries are
\begin{multline}
   \bra{(\lambda,t)} M_\nu \ket{(\lambda^\prime,t^\prime)} = \sqrt{\frac{q_\lambda q_{\lambda^\prime}}{d_\lambda d_{\lambda^\prime}}} \sum_{\substack{L \in \GT{\lambda} \\ L' \in \GT{\lambda^\prime} }} (-1)^{w_d(L^\prime) + w_d(L) + \varphi(L^\prime)+\varphi(L)} \\
   \times \sum_{N \in \GT{\nu}} C^{N,t}_{L \bar{L}} C^{N,t^\prime}_{L^\prime \bar{L}^\prime} \frac{1}{d_\nu}.
\end{multline}
Note that $M_\nu = K_\nu\ct K_\nu$, where $K_\nu$ is a $d_\nu \times r_\nu$ matrix which has entries $[K_\nu]_{N,(\lambda,t)}$ defined by
\begin{equation}\label{def:Knu}
    [K_\nu]_{N,(\lambda,t)} := \sqrt{\frac{q_\lambda}{d_\lambda d_\nu}} \sum_{L \in \GT{\lambda}} (-1)^{w_d(L)+\varphi(L)} C^{N,t}_{L \bar{L}}.
\end{equation}
At this point, we rewrite the function $w_d(L)+\varphi(L)$ for $L \in \GT{\lambda}$, $\lambda \in \Airreps{d}{n,0}$:
\begin{align}\label{eq:phases}
    w_d(L) + \varphi(L) &=  w_d(L) + \varphi_{d-1}(L) - \varphi_{d-1}(L_\lambda) \\
    &= \sum_{i = 1}^{d} L_{i,d} - \sum_{i = 1}^{d-1} L_{i,d-1} + \sum_{j=1}^{d-1} \sum_{i = 1}^{j} L_{i,j} - \varphi_{d-1}(L_\lambda) \\
    &= n + \varphi_{d-2}(L) - \varphi_{d-1}(L_\lambda) \\
    &= n - \varphi_{d-1}(L_\lambda) + \varphi_{d-2}(L^{(d-1)}),
    \label{eqn:phasestophi}
\end{align}
where $L^{(d-1)}$ is the Gelfand--Tsetlin pattern obtained from $L$ by dropping the top row, i.e. it is a GT basis label for the group $U(d-1)$. 
The analysis now will differ for $d=2$ and $d \ge 3$.

\subsubsection{Qubit case}
In the case of qubits ($d=2$), a more convenient notation for Gelfand--Tsetlin basis would be spin labelling $(j,m)$, where $j$ is total spin, replacing $\lambda$, and a spin $z$-projection $m$, replacing Gelfand--Tsetlin patterns $L$ (without the top row). More formally, correspondence between Gelfand--Tsetlin pattern $L$ and $(j,m)$ labelling for $SU(2)$ is as follows:
\begin{align}
	L &= \begin{pmatrix}
		2j && 0 \\
		&j+m
	\end{pmatrix}
        \equiv (j,m)
	\\
	\bar{L} &=
        \begin{pmatrix}
		2j && 0 \\
		&j-m
	\end{pmatrix} 
        \equiv (j,-m) \, .
\end{align}
For $SU(2)$, the spin $j$ also subsumes the number of qubits $n$ according to $2j=n$. When we consider $U(2)$ irreps we take into account the total number of boxes $n$ of $\lambda$ by using the following Gelfand--Tsetlin patterns:
\begin{align}
	L &= \begin{pmatrix}
		\frac{n}{2} + j && \frac{n}{2} - j \\
		&\frac{n}{2} + m
	\end{pmatrix}
        \equiv (j,m)
	\\
	\bar{L} &=
        \begin{pmatrix}
		-\frac{n}{2} + j && -\frac{n}{2} - j \\
		&-\frac{n}{2} - m
	\end{pmatrix}
        \equiv (j,-m)\, ,
\end{align}
with $0\le j\le n/2$.

When $d=2$, there are no multiplicities, so the label $t$ can be dropped completely in \Cref{def:Knu}. 
Moreover, the dual representation $\bar{\lambda}$ is isomorphic to $\lambda$ and corresponds to the same total spin $j$. 

In the context of \Cref{def:Knu} we will denote $\nu \equiv J$, $N \equiv N \in \set{-J,\dotsc,J}$, $\lambda \equiv j$ and drop $t$ index:
\begin{equation}\label{def:Knu_d=2}
    [K_J]_{N,j} := \sqrt{\frac{q_j}{(2j+1) (2J+1)}} (-1)^{n/2 - j} \sum_{m = -j}^{j} C^{J \, N}_{j\,m, j\,-m},
\end{equation}
where we used
\begin{align}
    w_d(L) + \varphi(L) &=  w_d(j,m) + \varphi_1(j,m) - \varphi_1(j,j) \\
    &= (n/2 - m) + (n/2 + m) - (n/2 + j).
\end{align}
Because of the angular momentum conservation, we immediately get $[K_J]_{N,j} = 0$ whenever $N \neq 0$, so
\begin{equation}\label{def:Knu_d=2_N=0}
    [K_J]_{0,j} = \sqrt{\frac{q_j}{(2j+1) (2J+1)}} (-1)^{n/2 - j} \sum_{m = -j}^{j} C^{J \, 0}_{j\,m, j\,-m}
\end{equation}
for $J/2\le j\le n/2$.
Therefore, we can write
\begin{equation}
    K_J = 
    \begin{cases}
        0          &\text{if } J \in \set{1 - (n \mod 2), \dotsc,n - 3,n-1}, \\
        \bra{v_J} &\text{if } J \in \set{ (n \mod 2), \dotsc,n - 2,n},
    \end{cases}
\end{equation}
where
\begin{align}
    \bra{v_J} &\defeq \sum_{ j \geq J/2 }^{n/2} \of[\bigg]{\sqrt{\frac{q_j}{(2j+1) (2J+1)}} (-1)^{n/2 - j} \sum_{m = -j}^{j} C^{J \, 0}_{j\,m, j\,-m}} \bra{j} \\
    \label{eq:norm_vJ_d=2}
    \norm{\ket{v_J}}^2 &= \sum_{ j \geq J/2 }^{n/2}  \frac{q_j}{(2j+1) (2J+1)} \abs[\bigg]{\sum_{m = -j}^{j} C^{J \, 0}_{j\,m, j\,-m}}^2 
\end{align}

Now we are ready to state our key numerical observation as a conjecture:

\begin{conjecture}\label{conj:d=2}
    The probabilities $q_j$ for $j \in \set{(n\mod 2)/2,\dotsc,n/2}$ in \Cref{eq:norm_vJ_d=2}, which define the probe state $\Phi_n$, could be chosen such that 
    \begin{equation}
        \norm{\ket{v_J}}^2 = \frac{1}{\binom{n+2}{2}}
    \end{equation}
    for every $J \in \set{(n \mod 2), \dotsc,n - 2,n}$. 
\end{conjecture}
The reason for this conjecture is the following: we observe numerical solutions $\{q_j\}_j$ for different $n \le 60$ of the following linear system 
\begin{equation}
    \frac{2J+1}{\binom{n+2}{2}} = \sum_{ j \geq J/2 }^{n/2}  \frac{q_j}{2j+1} \abs[\bigg]{\sum_{m = -j}^{j} C^{J \, 0}_{j\,m, j\,-m}}^2 
\end{equation}
for all $J \in \set{ (n \mod 2), \dotsc,n - 2,n}$. This conjecture implies a simple expression for the entropy of state $\widetilde{\mathcal{T}(\rho)}$:

\begin{corollary}\label{cor:entropy_d=2}
    For $d = 2$, assuming \cref{conj:d=2}, the entropy of $\widetilde{\mathcal{T}(\rho)}$ is given by
    \begin{equation}
        \entropy{\widetilde{\mathcal{T}(\rho)}} = \log_2 \binom{n+2}{2}.
    \end{equation}
\end{corollary}
    
\begin{proof}
    Assuming the conjecture,  the ensemble state $\widetilde{\mathcal{T}(\rho)}$ in (\ref{eqn:sumstate}) simplifies to
\begin{equation}
   \widetilde{\mathcal{T}(\rho)} = \frac{1}{{n+2\choose 2}}\bigoplus_{J=n\text{ mod } 2}^{n}\idm_{J} \otimes {\ket{v_{J}}\bra{v_{J}}\over \Vert \ket{v_{J}}\Vert^{2}}, 
\end{equation} where the direct sum is over $J=\lbrace n\text{ mod 2},\ldots,n-2,n\rbrace$. It  is clear that each rank-one projection ${\ket{v_{J}}\bra{v_{J}}\over \Vert \ket{v_{J}}\Vert^{2}}$ occurs exactly $2J+1$ times in $\widetilde{\mathcal{T}(\rho)}$. The dimension of the support of $\widetilde{\mathcal{T}(\rho)}$ is found to be
    \begin{align}
        \abs{\mathrm{supp}(\widetilde{\mathcal{T}(\rho)})} &= (2n+1) + (2(n-2)+1) + \dotsc + (2(n \mod 2)+1) \nonumber \\
        &= \begin{cases}
        \sum_{k=0}^{n/2} 4k + 1 &\text{if $n$ even}, \\
        \sum_{k=0}^{(n-1)/2} 4k + 3 &\text{if $n$ odd}, 
        \end{cases}
        = \binom{n+2}{2},
    \end{align}
    from which we conclude that  \Cref{conj:d=2} implies that $\widetilde{\mathcal{T}(\rho)}$ has flat spectrum. Its entropy is $\log_2 \binom{n+2}{2}$.
\end{proof}

\subsubsection{Qudit case}
We now treat the qudit case with $d\ge 3$.
\Cref{eq:phases}  together with \Cref{fact:RCG} and \Cref{lemma:magic}, assuming $d \ge 3$, gives us the following formula
\begin{align}
    &[K_\nu]_{N,(\lambda,t)} (-1)^{n - \varphi_{d-1}(L_\lambda)} = \sqrt{\frac{q_\lambda}{d_\lambda d_\nu}}  \sum_{L \in \GT{\lambda}}  (-1)^{\varphi_{d-2}(L^{(d-1)})} C^{N,t}_{L \bar{L}} \\
    &= \sqrt{\frac{q_\lambda}{d_\lambda d_\nu}}  \sum_{\mu \sqsubseteq \lambda} \sum_{M \in \GT{\mu}}  (-1)^{\varphi_{d-2}(M)} \sum_{r} RC^{\lambda,\bar{\lambda},\nu,t}_{\mu,\bar{\mu},\tilde{\nu},r} C^{N^{(d-1)},r}_{M \bar{M}} \\
    &= \sqrt{\frac{q_\lambda}{d_\lambda d_\nu}}  \sum_{\mu \sqsubseteq \lambda} (-1)^{\varphi_{d-2}(M_\mu)} \sum_{M \in \GT{\mu}}  (-1)^{\varphi(M)} \sum_{r} RC^{\lambda,\bar{\lambda},\nu,t}_{\mu,\bar{\mu},\tilde{\nu},r} C^{N^{(d-1)},r}_{M \bar{M}} \\
    &= \sqrt{\frac{q_\lambda}{d_\lambda d_\nu}}  \sum_{\mu \sqsubseteq \lambda} (-1)^{\varphi_{d-2}(M_\mu)} \sum_{r} RC^{\lambda,\bar{\lambda},\nu,t}_{\mu,\bar{\mu},\tilde{\nu},r} \sum_{M \in \GT{\mu}}  (-1)^{\varphi(M)} C^{N^{(d-1)},r}_{M \bar{M}} \\
    &= \sqrt{\frac{q_\lambda}{d_\lambda d_\nu}}  \sum_{\mu \sqsubseteq \lambda} (-1)^{\varphi_{d-2}(M_\mu)} \sum_{r} RC^{\lambda,\bar{\lambda},\nu,t}_{\mu,\bar{\mu},\tilde{\nu},r} \; \delta_{\tilde{\nu},\varnothing} \delta_{N^{(d-1)},\mathbf{0}} \sqrt{d_\mu} \label{eqn:jljl}\\
    &= \delta_{N^{(d-1)},\mathbf{0}} \sqrt{\frac{q_\lambda}{d_\lambda d_\nu}}  \sum_{\mu \sqsubseteq \lambda} (-1)^{\varphi_{d-2}(M_\mu)} RC^{\lambda,\bar{\lambda},\nu,t}_{\mu,\bar{\mu},\varnothing} \sqrt{d_\mu},
    \label{eqn:kfin}
\end{align}
where $\tilde{\nu} = N_{d-1}=N^{(d-1)}_{d-1}$, $\nu = N_d$ and $\mu \sqsubseteq \lambda$ means that $\mu$ interlaces $\lambda$, and (\ref{eqn:jljl}) follows from Lemma \ref{lemma:magic}. In the third line, $M_{\mu}$ uses the notation for the maximal GT pattern defined in (\ref{eqn:maxgtpat}). In the last line, we have used the fact that $\varnothing$ appears in the decomposition of $\mu\otimes \bar{\mu}$ with multiplicity 1, so $r=1$.
This calculation implies that actually the only non-zero matrix element (if it exists) in a given row of the matrix $K_\nu$ corresponds to $N$ such that $N^{(d-1)}$ is trivial, 
i.e. all entries of this GT pattern are zeros and, in particular, $\tilde{\nu} = (0,\dotsc,0) \equiv \varnothing$. 
However, for $RC^{\lambda,\bar{\lambda},\nu,t}_{\mu,\bar{\mu},\varnothing}$ to be non-zero it must be $\varnothing \sqsubseteq \nu$, meaning that $\nu = (k,0,0,\dotsc,0,-l)$ for some non-negative $k$ and $l$, where we use the staircase convention for labeling $\nu \in \hat{\mathcal{A}}^{d}_{n,n}$. 
On the other hand, since $\nu \in \lambda \otimes \bar{\lambda}$ then it must have the form $\nu = (k,0,0,\dotsc,0,-k)$ since the weight of $N$ must be zero. 
Therefore, we can conclude that the dimension of the support of $\widetilde{\mathcal{T}(\rho)}$ is at most 
\begin{equation}
    D(n,d) = \sum_{k = 0}^{n} d_{\nu_k},
\end{equation}
where $\nu_k \defeq (k,0,\dotsc,0,-k)$. Due to \emph{Pieri rule} for $\theta_n \defeq (n,0,\dotsc,0) \in \Airreps{d}{n,0}$ the following identity holds
\begin{equation}
    \theta \defeq \theta_n \otimes \bar{\theta}_n \cong \bigoplus_{k = 0}^{n} \nu_k,
\end{equation}
so it implies 
\begin{equation}
    D(n,d) = \sum_{k = 0}^{n} d_{\nu_k} = d_\theta^2 = \binom{n + d - 1}{d-1}^2.
    \label{eqn:dnd}
\end{equation}
To sum up, the matrix $K_\nu$ only has one row: $ K_\nu = \bra{v_{\nu}}$, where we define
\begin{align}
    \bra{v_{\nu}} &\defeq \sum_{\lambda,t} \of[\bigg]{\sqrt{\frac{q_\lambda}{d_\lambda d_\nu}} (-1)^{n - \varphi_{d-1}(L_\lambda)} \sum_{\mu \sqsubseteq \lambda} (-1)^{\varphi_{d-2}(M_\mu)} RC^{\lambda,\bar{\lambda},\nu,t}_{\mu,\bar{\mu},\varnothing} \sqrt{d_\mu}} \bra{(\lambda,t)}.
    \label{eqn:vvec}
\end{align}
The norm of this vector is given by
\begin{align}\label{eq:norm_v_nu_d>2}
     \norm{\ket{v_{\nu}}}^2 &= \sum_{\lambda} \frac{q_\lambda}{d_\lambda d_\nu} \sum_{t} \abs[\bigg]{\sum_{\mu \sqsubseteq \lambda} (-1)^{\varphi_{d-2}(M_\mu)} RC^{\lambda,\bar{\lambda},\nu,t}_{\mu,\bar{\mu},\varnothing} \sqrt{d_\mu}}^2.
\end{align}
Using this expression and our simplified expression for the reflected ensemble in the mixed Schur basis (\ref{eqn:sumstate}), we can maximize the entropy
\begin{equation}
    \entropy{\widetilde{\mathcal{T}(\rho)}} = - \sum_\nu d_\nu \norm{\ket{v_{\nu}}}^2 \log_2 \norm{\ket{v_{\nu}}}^2
\end{equation}
by adjusting the probability distribution $\{q_\lambda\}_\lambda$. 

\begin{figure}[tb]
    	\centering
    	\includegraphics[scale=0.5]{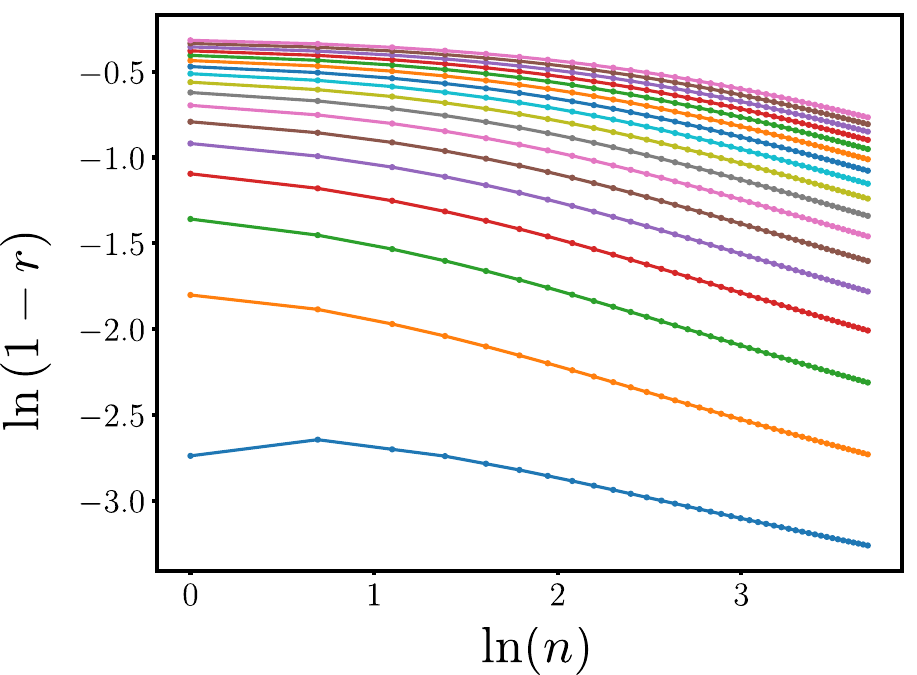}
    	\caption{Ratio $r$ (log scale) of the Holevo information $S\left(\widetilde{\mathcal{T}(\rho)} \right)$ of the state (\ref{eqn:specialrho}) to the upper bound $2\log_{2}D(n,d)$ for $d=4,\ldots, 20$ (increasing from bottom). $d=3$ is omitted due to saturation of the bound at $n=1$. Largest program size is $n=40$.
    	}\label{fig:entfig}
    \end{figure}

From the expressions (\ref{eqn:refens1}), (\ref{eqn:sumstate}) for the reflected ensemble, one can see that the maximal information that a program state can carry about a reflection is given by $2\log_{2}D(n,d)$ with $D(n,d)$ in (\ref{eqn:dnd}). However, for dimension $d$ and program size $n$ amenable to numerical calculation of the $U(d)$ CG coefficients appearing in $K_{\nu}$, optimization of the probe state symmetry weights $q_{\lambda}$ does not result in a maximal Holevo information $$\displaystyle \max_{\substack{\lbrace q_{\lambda} \rbrace_{\lambda \vdash_{d}n}, \: \sum_{\lambda} q_{\lambda}=1}} \entropy{\widetilde{\mathcal{T}(\rho)}}$$ that is close to this value. However, if one considers just tensor product of a $U(d)$ irrep with a symmetric irrep or its conjugate, instead of a with a general $U(d)$ irrep, then the reduced CG coefficients from Fact \ref{fact:RCG} can be computed for much larger $n$ and $d$ from an explicit combinatorial expression Eq. (6) of Section 18.2.6 of \cite{vilenkin}. With this in mind, note that taking a suboptimal program state defined by weights $\lbrace q_{\lambda}\rbrace_{\lambda \vdash_{d}n}$ having support only on the symmetric irrep $\theta_{n}$, i.e., $q_{\lambda}=\delta_{\lambda,\theta_{n}}$, results in a $K_{\nu}$ that can be written as a scalar $v_{\nu}$ (i.e., $\bra{v_{\nu}}$ has only one entry) and $v_{\nu}$ is a function of these easy-to-compute reduced CG coefficients (see (\ref{eqn:vvec})). Specifically, with this assumption, the $\bra{v_{\nu}}$ in (\ref{eqn:vvec}) reduces to a scalar $v_{\nu}$ that satisfies the following special case of (\ref{eq:norm_v_nu_d>2})
\begin{align}
    v_{\nu}^2 &= \frac{1}{d_{\theta_{n}} d_\nu}  \abs[\bigg]{\sum_{\mu \sqsubseteq \theta_{n}} (-1)^{\varphi_{d-2}(M_\mu)} RC^{\theta_{n},\overline{\theta_{n}},\nu}_{\mu,\bar{\mu},\varnothing} \sqrt{d_\mu}}^2 \nonumber \\
    &= \frac{1}{d_{\theta_{n}} d_\nu}  \left( \sum_{q=0}^{n}RC^{(\vec{0}_{d-1},-n),(n,\vec{0}_{d-1}),\nu_{k}}_{(\vec{0}_{d-2},-q),(q,\vec{0}_{d-2}),\varnothing}\sqrt{{q+d-2\choose d-2}}\right)^2
\end{align}
where in the first line the Pieri rule allows us to remove the sum over multiplicities that appears in (\ref{eq:norm_v_nu_d>2}), and in the second line we used the fact that in the present case $\mu$ has the form $\theta_{q}=(q,\vec{0}_{d-2})$ for $q=0,\ldots,n$ and the sum of the lowest $d-2$ rows of $M_{\mu}$ is zero. It is worth mentioning that in deriving the second line, we reflected the probe state ${1/ \sqrt{d_{\theta_{n}}}}\sum_{L\in \text{GT}(\theta_{n})}\ket{L}\ket{L}$ about a random state using $\overline{U_{\lambda}}R_{\lambda}U_{\lambda}^{T}$ instead of $U_{\lambda}R_{\lambda}U_{\lambda}^{\dagger}$ so that the conjugate representation appears in the first register. From (\ref{eqn:sumstate}) and (\ref{eqn:kfin}), the ensemble state is
\begin{small}
\begin{align}
    \widetilde{\mathcal{T}(\rho)}\big\vert_{q_{\lambda}=\delta_{\lambda,\theta_{n}}} &= \bigoplus_{\substack{\nu_{k} =(k,\vec{0}_{d-2},-k)  \\ k=0,\ldots,n}} \left( 
\sum_{q=0}^{n}RC^{(\vec{0}_{d-1},-n),(n,\vec{0}_{d-1}),\nu_{k}}_{(\vec{0}_{d-2},-q),(q,\vec{0}_{d-2}),\varnothing}\sqrt{{q+d-2\choose d-2}}  \right)^{2}{ \idm_{\nu_{k}} \over d_{\nu_{k}}d_{\theta_{n}}}
\label{eqn:specialrho}
\end{align}
\end{small}
where the $U(d)$ reduced CG coefficient $RC^{(\vec{n}_{d-1},0),(n,\vec{0}_{d-1}),(2n-k,\vec{n}_{d-2},k)}_{(\vec{n}_{d-2},n-q),(q,\vec{0}_{d-2}),\vec{n}_{d-1}}$ can also be used. One can see in \cref{fig:entfig} that, for each $d$, increasing the program size $n$ increases the ratio of the entropy of (\ref{eqn:specialrho}) to the upper bound $2\log_{2}D(n,d)$. The scaling of the ratio has the form $1-{c(d)\over n^{\alpha}}$ with $\alpha \approx 0.3$ independent of $d$\footnote{The code used to produce this figure and the related numerics is available at \cite{ReflectionsCode}.}. We therefore posit the following conjecture.

\begin{conjecture}\label{conj:entropy_d>2}
    For any $d \ge 3$ and for any $\delta > 0$ there exists $n_0$ such that for every $n \geq n_0$ the entropy of the state $\widetilde{\mathcal{T}(\rho)}$ is lower bounded as
    \begin{equation}
        \entropy{\widetilde{\mathcal{T}(\rho)}} \geq 2(1-\delta)\log_2 \binom{n+d-1}{d-1}.
    \end{equation}
\end{conjecture}
We are now ready to calculate the lower bound for $d_P$.

\subsection{Calculating the lower bound}

Assuming \cref{cor:entropy_d=2,conj:entropy_d>2}, our lower bound $\ln d_P \ge f_d(\epsilon,n)$ from \cref{lemma:lb}, where
\begin{equation}
f_d(\epsilon,n) \defeq \information*{\set*{\mathcal T \circ \mathcal R^{\otimes n}_\psi(\Phi_n), d\psi}} \ln2 - (4 n \sqrt{2\epsilon}) \ln d_n - \ln2,
\end{equation}
has the following lower bound for every $d$, $\delta > 0$ and sufficiently large $n$:
\begin{equation}
    f_d(\epsilon,n) \geq
     2 (1 - \delta) \ln \binom{n+d-1}{d-1} - 4n \sqrt{2\epsilon} \ln \binom{n+d^2-1}{d^2-1} - 2\ln 2.
\end{equation}
In the above expression, we weakened the lower bound by an additive contribution $\ln 2$ in order to incorporate the  the case $d=2$ from \cref{cor:entropy_d=2} into general $d$ case, due to the following inequality:
\begin{equation}
    \ln \binom{n+2}{2} \geq \ln \frac{(n+1)^2}{2} = 2\ln \binom{n+1}{1} -  \ln 2
\end{equation}
The task now is to choose the correct relationship between $\epsilon$ and $n$, with the aim to maximize the lower bound as much as possible. 
We summarize the result of this section in the following Theorem.
\begin{theorem}\label{thm:lbReflection}
    Assuming \cref{conj:d=2,conj:entropy_d>2},
    the program dimension of programming reflection operators on $\mathbb{C}^{d}$ for $d\ge 2$ is lower bounded by
    \begin{equation}
        \ln d_P \geq (1-\delta) (d-1) \ln\of*{\frac{d^{-4}}{\epsilon}} 
    \end{equation}
    for all $\delta > 0$, and sufficiently small $\epsilon$.
\end{theorem}
\begin{proof}
To simplify the lower bound, we would need to use the following easy inequalities:
\begin{align}
    \frac{n^k}{k!} &\leq \binom{n+k}{k} \leq \frac{(n+k)^k}{k!}, \label{binom_ineq}\\
    k \ln k - k &< \ln k! < (k+1) \ln k - k,
    \label{factorial_ineq}
\end{align}
the latter inequality holding for $k>6$.
By using \cref{binom_ineq} we get
\begin{align}
    f_d(\epsilon,n) &\geq 2(1 - \delta)\ln \binom{n+d-1}{d-1} - 4n \sqrt{2\epsilon} \ln \binom{n+d^2-1}{d^2-1} - 2\ln 2 \\
    &\geq 2(1 - \delta) \ln \frac{n^{d-1}}{(d-1)!} -  4n \sqrt{2\epsilon} \ln \frac{(n+d^2-1)^{d^2-1}}{(d^2-1)!} - 2\ln 2 \\
    &\geq 2(1 - \delta)(d-1) \ln n - 2(1 - \delta)\ln (d-1)! + 4n \sqrt{2\epsilon} \ln (d^2-1)!  \nonumber \\
    & \qquad \qquad \qquad - 4n \sqrt{2\epsilon} (d^2-1) \ln (n+d^2-1) - 2\ln 2. \label{eq:f:3} \\
    \nonumber
    & = 2(1 - \delta)(d-1) \ln n  - 4n \sqrt{2\epsilon}(d^2-1) \left( \ln (n+d^2-1) - \frac{\ln(d^2-1)!}{(d^2-1)}\right)  \\
    & \qquad \qquad \qquad -( 2(1 - \delta)\ln (d-1)!+ 2\ln 2) \;.
    \label{eqn:lowerboundsimplify}
\end{align}
As in \cite{Yang2020}, the parameter $\epsilon$ is fixed to denote the error of an approximate programmable processor for a single reflection channel $\mathcal{R}_{\psi}$. 
We would like to choose the number of copies $n$ as a function of $\epsilon$ in such a way that the lower bound is as large as possible. 
One could do this by taking the derivative of the right hand side with respect to $n$ and set it equal to zero. 
However, the resulting equations cannot be solved analytically. 
Instead, we make some approximation to the expression first and then set the derivative to zero. 
The last term is independent of $n$ so it can be ignored. 
For the second term, notice that as $n\gg d^2$ the expression inside the parenthesis goes as $\ln(n)$. 
Thus, we want to maximize the function $g(n) = 2(1 - \delta)(d-1) \ln n  - 4n \sqrt{2\epsilon}(d^2-1) \ln (n)$. 
Setting the derivative $g'(n)$ to zero we obtain the following relation between $n$ and $\epsilon$ that approximately maximizes the lower bound:
\begin{equation}
    \label{eps_n_relation}
    n \ln n = \frac{1 - \delta}{2(d+1)\sqrt{2\epsilon}}.
\end{equation}
Therefore by using \cref{eps_n_relation} we simplify the lower bound as follows
\begin{multline}
    f_d(\epsilon(n),n) \geq 2(1 - \delta) (d-1) \ln n - 2(1 - \delta) \ln (d-1)!  - 2\ln 2  \\ + 2(1 - \delta) \of*{\frac{\ln (d^2-1)!}{(d+1) \ln n} - (d-1) \frac{\ln n+ \ln(1+\tfrac{d^2-1}{n})}{\ln n}} .
    \label{eqn:substitutedlowerbound}
\end{multline}
Comparing this expression to \cref{eqn:lowerboundsimplify}, one sees that the near-optimal choice of $n$ causes the lower bound to be dominated by the Holevo information  $\information{ \set{\mathcal{T}\circ \mathcal{R}_{\psi}^{\otimes n}(\Psi_{n}), d\psi }}$, while the $O(n\sqrt{\epsilon(n)}\ln n)$ contribution arising from the Alicki-Fannes-Winter inequality, which arises due to imperfect programming, is suppressed. Physically, the near-optimal $n$ in \cref{eps_n_relation} guarantees that the lower bound on $d_{P}$ asymptotically depends only on the amount of information about $\mathcal{R}_{\psi}$ that can be imprinted on $\Psi_{n}$.

In \cref{eqn:substitutedlowerbound} we note that for $n > 1$, the inequality
\begin{equation}
    \frac{ \ln (d^2-1)!}{(d+1) \ln n} -  (d-1) \frac{\ln(1+\tfrac{d^2-1}{n})}{\ln n} \geq 0
\end{equation}
is implied by 
\begin{equation}
    n \geq \ceil*{ \frac{d^2-1}{\exp{\tfrac{\ln (d^2-1)!}{d^2-1}} - 1}},
\end{equation}
if it holds for every $d \geq 3$. But using the lower bound from \cref{factorial_ineq} we see that requiring $n \geq 6$ is sufficient so show
\begin{multline}
    n \geq 6 > \frac{1}{\tfrac{1}{e}-\tfrac{1}{8}} + 1 \geq \frac{1}{\tfrac{1}{e}-\tfrac{1}{d^2-1}} + 1 = \frac{d^2-1}{\exp \left[\ln (d^2-1) - 1\right] -1} + 1 \\
    \geq \frac{d^2-1}{\exp{\tfrac{\ln (d^2-1)!}{d^2-1}} - 1} + 1  \geq  \ceil*{ \frac{d^2-1}{\exp{\tfrac{\ln (d^2-1)!}{d^2-1}} - 1}}
\end{multline} 
for every $d \geq 3$. 
This argument leads to a simplification of the lower bound. 
Namely for every  $n \geq 6$, and with $\epsilon$ taken such that \cref{eps_n_relation} holds, we get
\begin{align}
    f_d(\epsilon(n),n) &\geq 2(1 - \delta)  ((d-1) \ln n - \ln (d-1)!  - (d-1)) - 2\ln 2  \\
    &\geq 2(1 - \delta) ((d-1) \ln n - d \ln (d-1) + (d-1)  - (d-1)) - 2\ln 2 \nonumber \\
    &\geq 2(1 - \delta) ((d-1) \ln n - (d-1) \ln (d-1) - \ln(d-1)) - 2\ln 2  \nonumber \\
    &\geq 2(1 - \delta) (d-1) \ln \frac{n}{d-1} - 2(1 - \delta) \ln(d-1)  - 2\ln 2 \label{f_d:lb:ln n}
\end{align}
where we also used \cref{factorial_ineq}. However, we are really interested in $\epsilon$ dependence of the lower bound. 
To find this, we need to solve \cref{eps_n_relation} for $n$ in terms of $\epsilon$. 
First of all, the condition $n \geq 6$ translates into 
\begin{equation}
    \epsilon \leq \frac{(1 - \delta)^2}{2 (12(d+1) \ln 6) ^2},
\end{equation}
so it is enough to require
\begin{equation}
    \epsilon \leq \frac{10^{-3}(1 - \delta)^2}{(d+1)^2}.
\end{equation}
The solution of \cref{eps_n_relation} is given by so-called \emph{Lambert function} $W_0$:
\begin{equation}\label{eps_n_relation_solution}
    \ln n = W_0\of*{\sqrt{\tfrac{(1 - \delta)^2}{8(d+1)^2\epsilon}}},
\end{equation}
which has the following bounds for every $x \geq e$:
\begin{equation}
    \ln x - \ln \ln x + \frac{\ln \ln x}{2 \ln x} \leq W_0(x) \leq \ln x - \ln \ln x + \frac{e}{e - 1} \frac{\ln \ln x}{\ln x}.
\end{equation}
Combining the lower bound for $W_0$ with \cref{f_d:lb:ln n} leads to
\begin{align}
    &f_d(\epsilon, n(\epsilon)) \nonumber\\
    &\geq 2 (d-1)(1 - \delta) \of*{W_0\of*{\sqrt{\tfrac{(1 - \delta)^2}{8(d+1)^2\epsilon}}} - \ln(d-1)} - K \\
    &\geq 2 (d-1)(1 - \delta) \of*{\ln\of*{\sqrt{\tfrac{(1 - \delta)^2}{8(d+1)^2\epsilon}}} - \ln(d-1) - \ln \ln \of*{\sqrt{\tfrac{(1 - \delta)^2}{8(d+1)^2\epsilon}}}} - K \nonumber\\
    &\geq (1 - \delta)(d-1) \ln\of*{\tfrac{(1 - \delta)^2}{8(d^2-1)^2\epsilon}} \times \nonumber\\
    &\qquad \times \of*{1 - \frac{2\ln \ln \of*{\sqrt{\tfrac{(1 - \delta)^2}{8(d+1)^2\epsilon}}}}{\ln\of*{\tfrac{(1 - \delta)^2}{8(d^2-1)^2\epsilon}}} - \frac{K}{(d-1) \ln\of*{\tfrac{(1 - \delta)^2}{8(d^2-1)^2\epsilon}}}}
    \label{f_d:lb}
\end{align}
where $K \defeq 2(1 - \delta) \ln(d-1) + 2\ln 2$ does not depend on $\epsilon$. 
\Cref{f_d:lb} implies that for every $\delta > 0$ there exists sufficiently small $\epsilon_0(\delta,d)$ such that for every $0 < \epsilon < \epsilon_0(\delta,d)$ we get 
\begin{equation}\label{eq:lbQudit}
    \ln d_P \geq (1-\delta) (d-1) \ln\of*{\frac{d^{-4}}{\epsilon}}.
\end{equation}
\end{proof}

\section{Universal programming}\label{sec:universal}
\begin{figure}
	\centering
	\begin{quantikz}
	    & & & \gate[2]{\apxRef_{\pi a_1}} & \ \ldots \ & \gate[4]{\apxRef_{\pi a_{d-1}}} & \\
		\lstick{$\ket{\psi_{P,1}}$} & \gate{\decoder_{n_1}} & \qwbundle{n_1} &  & \ground{} \\
		\lstick{$\vdots$}\setwiretype{n}& &  & & \ddots & & && & \\
		\lstick{$\ket{\psi_{P,{d-1}}}$} & \gate{\decoder_{n_{d-1}}} & \qwbundle{n_{d-1}} & & & & \ground{} \\
	\end{quantikz}
	\caption{%
		We can implement universal programming by a sequence of $d-1$ "approximate reflection channels"
        $\apxRefCh_{\psi_j, R_{\theta_j}}$,
        where $\ket{\psi_j}$ are the (orthogonal) eigenstates of the target unitary.
        For all $j$, we approximate $\theta_j$ with a binary string of finite length
        representing the number $a_j \in (-1,1)$.
        Universal programming requires us to provide the number of copies $n_1, \dots, n_{d-1} \in \mathbb N$ and $a_1, \dots, a_{d-1}$ as part of the program state.
        We immediately measure these registers (here implicitly) so that these values are available classically.
        The symmetric states $\ket{\psi_j}^{\otimes n_j}$ are given to the processor
        as encoded programs $\ket{\psi_{P,j}} \coloneqq \decoder^\dagger_j \ket{\psi_j}^{\otimes n}$,
        defined in \cref{sec:savingSpace}, to save space.
	}%
	\label{fig:universalProgramming}
\end{figure}

In this section we describe an implementation of a universal quantum processor based on the rotations described in \cref{sec:unitaryreflections}.
We begin by noting that an arbitrary unitary up to a global phase can be written as
\begin{equation}
    \label{eq:genericunitary}
    U = \sum_{j=0}^{d-1} e^{i \alpha_j} \ketbra{\psi_j}{\psi_j} = \prod_{j=0}^{d-1} e^{i \alpha_j \ketbra{\psi_j}{\psi_j}}
\end{equation}
where $\alpha_j\in(-\pi,\pi]$ and $\alpha_0 = 0$.
We can construct an approximate quantum processor for $U$ with a program state consisting of three registers.
The first program register contains a binary encoding in the computational basis of certain optimally chosen phases $\theta_j$.
The second program register encodes in binary the number of copies of each eigenstate $n_j\in \mathbb N$.
The third program register encodes the copies  $\ket{\psi_j}^{\otimes n_j} \in \sym^n(\mathbb{C}^{d})$
using program states of dimension $\dim(\sym^{n_j}(\mathbb{C}^d))$, each.
Correspondingly, there are three sources of error in the algorithm:
one due to the binary representation of the phases using a finite number of qubits, one due to error in the reflection channel implemented using our algorithm with finite number of copies and finally one due to the imperfect implementation of the isometry that recovers $\ket{\psi_j}^{\otimes n_j}$ from the program states.

Specifically, the structure of the phase register is specified as follows:
let $\theta_j$ be the rotation angle of $\apxRefCh_{V, \apxRef_{\theta_j}}$ that approximately programs a rotation by $\alpha_j$ (found numerically in \cref{fig:diamond2} or simply use $\theta_j=\alpha_j$).
Then we can approximate $\theta_j$ by a binary string of length logarithmic in the error.

\begin{lemma}\label{lem:binaryApproximation}
    Let $\theta_j / \pi$ be approached by a series $(-1)^{b_{j,0}} 0.b_{j,1}b_{j,2}\dots$ where we used a binary fraction.
    The binary approximation $a_{j} \coloneqq (-1)^{b_{j,0}} 0.b_{j,1}\dots b_{j,K}$,
    for $K \in \mathbb N_+$,
    of $\theta_j/\pi$ with error
    \begin{equation}
        0 < \abs{\theta_j - \pi a_j} \le \epsilon 
    \end{equation}
    can be obtained using $K \ge \log_2\paren*{\frac{\pi}{\epsilon}}$ bits.
\end{lemma}
\begin{proof}
    The difference
    \begin{equation}
        \frac{\theta_j}{\pi} - a_j = (-1)^{b_{j,0}} 0.0\cdots0b_{j,K+1}\cdots
    \end{equation}
    such that $\abs{\frac{\theta_j}{\pi} - a_j} \le 2^{-K}$.
    The result follows easily.
\end{proof}

The algorithm first measures the registers containing $n_1$, $\ldots,$ $n_{d-1}$ and $a_1$, $\ldots,$ $a_{d-1}$
so that these values are available classically.
The structure of the rest of the algorithm is given in \cref{fig:universalProgramming}.
We now bound the size of the program dimension of this algorithm based on the desired error.

\begin{theorem}\label{thm:universalProgram}
    With additional registers classically encoding $n_1,\dots, n_{d-1} \in \mathbb N$ and $a_1,\dots, a_{d-1} \in (-1,1)$,
    \cref{fig:universalProgramming} implements a universal quantum processor with program dimension
    \begin{equation}\label{eq:progDimUniversal}
        \log_2(d_P) = (d-1)^2 \log_2 \paren*{\frac{\bigo{1}}{\epsilon}}.
    \end{equation}
\end{theorem}
\begin{proof}
Let us then define the unitary
\begin{equation}
    \tilde{U} \coloneqq \sum_{j=1}^{d-1} e^{i \pi a_j} \ketbra{\psi_j}{\psi_j} = \prod_{j=1}^{d-1} e^{i \pi a_j \ketbra{\psi_j}{\psi_j}}
\end{equation}
for some $a_1,\dots,a_{d-1} \in \mathbb (-1,1)$ specified later.
The diamond distance between the target unitary channel $\mathcal U$ and the channel $\mathcal E$ implemented by the universal processor can be upper bounded by the triangle inequality as 
\begin{align}
    \label{eq:diamdistbound1}
    \opnorm{\mathcal U - \mathcal{E}} \le  \opnorm{\mathcal U - \tilde{\mathcal{U}}} + \opnorm{\tilde{\mathcal{U}} - \tilde{\mathcal E}} + \opnorm{\tilde{\mathcal{E}} - \mathcal E},
\end{align}
where we have indicated by $\tilde{\mathcal{E}}$ the channel implemented by the universal processor if the isometries $\decoder_n$ were implemented without error. We will bound each term on the right-hand side by $\epsilon/3$. This can be done by approximating each term in \cref{eq:genericunitary} with error $\epsilon/(3(d-1))$. 
Now, \cref{lem:binaryApproximation} implies $K = \ceil{\log_2 \paren{\frac{6\pi(d-1)}{\epsilon}}}$ bits is sufficient for each $a_j$ to approximate $\theta_j$ up to error $\epsilon/(6(d-1))$.
Then \cite[Theorem 12]{Johnston2009}\footnote{Theorem 26 in the ArXiv version.} implies $\opnorm{\mathcal U - \tilde{\mathcal{U}}} \le \epsilon/3$.
In order to ensure that also the second term on the right-hand side of \cref{eq:diamdistbound1} is bounded by $\epsilon/3$,
it suffices \cref{eq:diamdistboundalpha} to choose $n_j= \ceil{9(d-1) \abs{\alpha_j}/\epsilon}$ since $\abs{\pi a_j} \le \abs{\alpha_j}$.
Therefore, the second register encodes the size of $n_j$ in $\ceil{\log_2 \ceil{9\pi(d-1)/\epsilon}}$ qubits each.
Finally, each isometry $\decoder_n$ can be implemented to $\epsilon/(3(d-1))$ error by calling the quantum Schur transform algorithm of \cite{Bacon2007} with error parameter $\epsilon/(6(d-1))$ \cite[Lemma 12.6]{Aharonov1998}.

Each rotation is programmed using symmetric states in $\sym^{n_j}(\mathbb C^d)$.
Since we need to encode a total of $d-1$ rotations, the number of qubits needed for the total program register is
\begin{multline}
    \log_2(d_P) \le (d-1) \paren*{ \ceil*{\log_2 \ceil*{\frac{6\pi(d-1)}{\epsilon}}}+ \ceil*{\log_2 \ceil*{\frac{ 9\pi (d-1)}{\epsilon}}} } \\
    + \sum_{j=1}^{d-1} \ceil*{\log_2 \paren{\dim\paren{\sym^{n_j}(\mathbb C^d)}}}.
\end{multline}
The program cost is clearly dominated by the last term that, upon substituting $n_j= \ceil{9(d-1) \abs{\alpha_j}/\epsilon}$ in \cref{eq:symmetricCost}, can be upper bounded as
\begin{align}
    \sum_{j=1}^{d-1} \ceil*{\log_2 \paren{\dim\paren{\sym^{n_j}(\mathbb C^d)}}} \le (d-1) \sum_{j=1}^{d-1} \ceil*{\log_2 \paren*{\frac{9|\alpha_j|}{\epsilon} +  1}}  + \bigo*{d^2}
\end{align}
Since we can bound $|\alpha_j|= \bigo{1}$ for all $j$, \cref{eq:progDimUniversal} follows as required.
\end{proof}

Now we derive the gate complexity of $\mathcal E$
assuming all $n_j$ are powers of two.
(Note that we may always at most double all $n_j$ values to satisfy this condition, leading to a constant overhead.)
The construction of $\decoder_{n_j}$ \cref{eq:decoderConstruction},
together with our construction of $\apxRef_\theta$ (\cref{thm:efficientRotation}),
gives a gate complexity of
\begin{equation}
    \bigo*{\sum_{j=1}^{d-1} \poly*{n_j, d, \log \frac{1}{\delta_j}} + n_j \log \paren{n_j}} 
    = \bigo*{d \cdot \poly*{\frac{d}{\epsilon}}}
\end{equation}
since we had set $\delta_j = \bigo{\epsilon/d}$ and $n_j=9d/\epsilon$.
It is easy to see that our universal processor (\cref{fig:universalProgramming}), when supplied just copies of $\ket{\psi_j}^{\otimes n_j}$ instead of program states encoding the symmetric subspace,
has a gate complexity of $\bigo{d^2/\epsilon}$
because we then do not need the computationally expensive Schur transform.

\textcite{Yang2020} give an algorithm for universal programming with
\begin{equation}\label{eq:yangUB}
    \log_2 (d_P) \le \paren*{\frac{d^2-1}{2}} \log_2 \paren*{\frac{\bigtheta{d^2}}{\epsilon}}
\end{equation}
that requires entangled program states prepared with query access to the unitary $U^{\otimes n}$.
We observe that \cref{thm:universalProgram} has a worse scaling in $d_P$ for the limit of small $\epsilon$
but improved scaling in the limit of large $d$.

Finally, we obtain a lower bound on programming reflections by our reduction to universal programming
due to a lower bound on universal programming~\cite{Yang2020}.
The lower bound is weaker than \cref{thm:lbReflection}, but requires much less effort.
Therefore, showing a reduction to universal programming may be a straightforward way to derive lower bounds on other gate families.
By \cref{thm:universalProgram}, the unitary requiring the largest program size in number of qubits to program to error $\epsilon > 0$
can be decomposed into some $d-1$ unitary rotations $\mathcal{R}_{\psi_{j}}(\alpha_{j})$, each with error $\epsilon/(d-1)$.
Let $c_{\alpha_j}(\epsilon/(d-1))$ be the number of qubits to program $\mathcal R_{\psi_j}(\alpha_j)$ to error $\epsilon/(d-1)$,
then we use a lower bound on universal programming~\cite{Yang2020} as
\begin{equation}
    \sum_{j=1}^{d-1} c_{\alpha_j}(\epsilon/(d-1)) \ge \frac{d^2-1}{2}\log \paren*{\frac{\Theta(d^{-4})}{\epsilon}}.
\end{equation}
Setting $\epsilon \to (d-1)\epsilon$ we get 
\begin{equation}\label{eq:universalLBSum}
    \sum_{j=1}^{d-1} c_{\alpha_j}(\epsilon) \ge \frac{d^2-1}{2}\log \paren*{\frac{\Theta(d^{-5})}{\epsilon}}.
\end{equation}
There exists at least one rotation $k \in [d-1]$ that is greater than or equal to the average value
such that
\begin{equation}\label{eq:newbound}
    c_{\alpha_k}(\epsilon) \ge \sum_{j=1}^{d-1} \frac{c_{\alpha_j}(\epsilon)}{d-1} \ge \frac{d+1}{2}\log \paren*{\frac{\Theta(d^{-5})}{\epsilon}}.
\end{equation}
The family of programmable rotations requires at least $c_{\alpha_k}(\epsilon)$ qubits to program.

Although it is conceivable that the worst case rotation angle depends on the details of the programmable quantum processor,
we suspect that it might be always $\pi$
since this is the case for the algorithms described in this paper.
Assuming this, we compare \cref{eq:newbound} to the lower bound on the program cost for rotations with angle $\pi$ given in \cref{thm:lbReflection}
\begin{align}\label{eq:ourlowerbound}
    (d-1) \ln\left(\frac{\Theta(d^{-4})}{\epsilon}\right)
\end{align}
and see that the latter bound is better by a factor of 2 overall and a factor of $d$ inside the logarithm.

\section{Conclusion}\label{sec:conclusion}
We have seen how program states of the form $\ket{\psi}^{\otimes n}$ can be used to efficiently program reflections and rotations.
We have also obtained a lower bound on the program dimension without assuming a particular form for the program states. 
Our approximate programmable processors utilizing program states of the form $\ket{\psi}^{\otimes n}$ almost saturate the lower bound and thus are optimal up to logarithmic factors in $d$.
Finally, we showed that approximate programming of rotations, $\mathcal{R}_{\psi_{j}}(\alpha_{j})$,
also give a construction for a universal processor in \cref{thm:universalProgram}
with a better scaling of $d$ inside the logarithm.

\textcite{Yang2020} conjecture that the program dimension obeys
\begin{equation}\label{eq:yangDimension}
	\log d_P \sim \frac{\nu}{2} \log\paren*{\frac{C_{\nu,d}}{\epsilon}},
\end{equation}
for $\nu$ the number of real parameters of the gate family being programmed,
and $C_{\nu,d}$ is a parameter independent of $\epsilon$.
In the case of reflections about a pure state,
we confirm the conjecture \cref{eq:yangDimension} in \cref{cor:reflectionProgramDimension}
with $C_{\nu, d} = C_{2d-2,d} = \bigo{d^{-1}}$ 
since the number of real parameters is the real dimension of the manifold of pure states.
Compared to universal programming,
programming reflections requires a leading factor $(d+1)/2$ fewer samples that is both necessary and sufficient.

Our probe state \cref{eqn:rys} in the calculation of the Holevo information of a uniform ensemble of reflections has the same form as the optimal state for estimation of a general $U\in U(d)$ using an invariant cost function~\cite{PhysRevA.72.042338}.
When restricting to reflection operations,
we cannot count out the possibility that different entangled states in the $\lambda$ sectors can be used which could make the optimal sector amplitudes $\sqrt{q_{\lambda}}$ easier to identify.
Regardless, as one can see from taking the supremum of (\ref{eqn:lbchain}) over probe states $\Phi_{n}$,
the lower bound we obtain in \cref{thm:lbReflection} is a lower bound on the maximal amount of information about a reflection that can be stored in entangled states of the full space $(\mathbb{C}^{d})^{\otimes n}\otimes (\mathbb{C}^{d})^{\otimes n}$,
and we have obtained a matching upper bound to within a logarithmic factor.

There are some open questions from this work.
It could be of interest to see if some of our assumptions on the structure of an "approximate reflection channel" can be removed.
In particular, it could be interesting to allow for an isometry $V$ instead of unitary $V$.
This would be a step towards showing a bound on an arbitrary quantum channel as a processor for reflections.

Measure-and-operate algorithms, while optimal for universal programming of $U(d)$,
may not be optimal for programming proper submanifolds of $U(d)$. 
In \cref{sec:mr} we show that $n$ for a particular measure-and-operate algorithm acting on $\ket{\psi}^{\otimes n}$ has a $d$ dependence.
It is an open question whether a measure-and-operate algorithm operating on this or a different program state can achieve the performance of the algorithms described in this work.

There are also various generalizations of rotations that may be worth investigating.
We could consider programming a rotation around a subspace specified by a projector $P \subseteq \mathbb C^d$
such that the rotation approximates $e^{i \alpha P}$.
If we are provided an orthonormal basis $\set{\ket{\psi_i}}_i$ of $P$ as program states,
then this is just a special case of our universal programming algorithm (\cref{sec:universal}).
Yet, when the program states are not orthogonal it is unclear how to implement such a subspace reflection.

\subsection*{Acknowledgements}
We thank Touheed Anwar Atif, Yuxiang Yang, Giulio Chiribella, and Akihito Soeda for discussions and Rolando Somma for early stages work.
We thank Alireza Seif for pointing out \cite{Lloyd2014} and Maris Ozols for discussions on \cite{Kimmel2017}.
ES, YS, and TJV acknowledge initial support from
the Laboratory Directed Research
and Development program of Los Alamos National Laboratory (LANL) under project number 20210639ECR. 
This material is based upon work supported by the U.S. Department of Energy, Office of Science, National Quantum Information Science Research Centers, Quantum Science Center (QSC).
ES and YS were funded by the QSC to perform the analytical calculations and to write the manuscript along with the other authors.
ES finished proving bounds on rotations, added universal programming, and finished writing the manuscript at IBM.
DG was supported by the U.S.\ DOE through a quantum computing program sponsored by the LANL Information Science \& Technology Institute.
DG and YS acknowledge support from Laboratory Directed Research and Development (LDRD) program of Los Alamos National Laboratory (LANL) under project number 20230049DR.
DG was also supported by an NWO Vidi grant (Project No.\ VI.Vidi.192.109) and NWO grant NGF.1623.23.025 (“Qudits in theory and experiment”).

\printbibliography%

\appendix
\section{Dimension-dependent copy complexity: measure and reflect}\label{sec:mr}
A major result of \cite{Yang2020} is the fact that an optimal approximate universal processor for gates in $U(d)$ can be obtained by using a measure-and-operate method in which $U(d)$-valued measurement is carried out on a program state of the form $U^{\otimes n}\ket{\psi_{\text{probe}}}$.
Although program states do not necessarily have to take this restricted form,
optimality of the measure-and-operate algorithm establishes a connection between optimal approximate universal unitary programming and optimal quantum metrology.
Specifically, the connection was made by choosing $\ket{\psi_{\text{probe}}}$ according to the optimal probe states for $U(d)$ estimation (viz., the ``class states'') \cite{PhysRevA.72.042338,PhysRevLett.93.180503,PhysRevA.75.022326} in such a way that the entanglement fidelity of the measure-and-operate algorithm  has a similar mathematical structure to the average case error for covariant estimation of $U(d)$ or elapsed time estimation by a quantum clock \cite{PhysRevLett.82.2207}.
One is then left with the (non-trivial) task of porting the optimal probe state coefficients from the estimation setting to the entanglement fidelity expression.
Note that the representation-theoretic framework for finding the optimal probe state for estimation of $U(d)$ was developed by extending the covariant measurement theory for optimal estimation of covariant state families \cite{Chiribella_2011,HolevoBook}.

Motivated by the measurement-based approximate universal-NOT gate for qubits \cite{Buzek1999}, we formulate a measurement-based approximate programmable processor for $R_{\psi}$ (programmable simply by the state $\psi^{\otimes n}$, in lieu of $R_{\psi}^{\otimes n}\ket{\psi_{\text{probe}}}$) and show that it exhibits a dimension-dependent copy complexity. This result, combined with the analyses of covariant approximate reflection algorithms above, shows that measure-and-operate algorithms, while optimal for universal programming of $U(d)$, are not necessarily optimal for programming proper submanifolds of $U(d)$. 

Specifically, we define the measure-and-reflect algorithm $\mathcal{E}_{\text{MR},\psi}$
\begin{equation}
    \mathcal{E}_{\text{MR},\psi}(\phi):= \left( \text{tr}P_{n} \right) \int d\xi \, \vert \langle \xi \vert \psi \rangle\vert^{2n} \, \mathcal{R}_{\xi}(\phi),
\end{equation}
where $P_{n}$ is the projection to the symmetric subspace of $(\mathbb{C}^{d})^{\otimes n}$ and the integral over $\ket{\xi}$ is with respect to the uniform measure. $\mathcal{E}_{\text{MR}}$ can be interpreted as the expected outcome of a reflection-valued measurement. One can straightforwardly see that $\mathcal{E}_{\text{MR},\psi}$ is covariant with respect to $\langle R_{\psi},\idm\rangle'$ either by explicitly verifying the covariance condition $[U\otimes \bar{U},J_{\mathcal{E}_{\text{MR},\psi}}]=0$ for any $U\in \langle R_{\psi},\idm\rangle'$ using the Choi state $J_{\mathcal{E}_{\text{MR},\psi}}$ of $\mathcal{E}_{\text{MR},\psi}$, or directly from the definition by noting
\begin{align}
    \left( \int d\xi \, \vert \langle\xi\vert\psi\rangle\vert^{2n} \mathcal{R}_{\xi} \circ \mathcal{U} \right)(\phi) &= \int d\xi \, \vert \langle\xi\vert\psi\rangle\vert^{2n} R_{\xi}U\phi U^{\dagger}R_{\xi} \nonumber \\
    &= \int dv \, \vert \langle v \vert U^{\dagger}\vert \psi\rangle\vert^{2n} UR_{v}\phi R_{v}U^{\dagger} \nonumber\\
    &= \left( \mathcal{U}\circ \int dv \, \vert \langle v \vert \psi\rangle\vert^{2n} \mathcal{R}_{v} \right)(\phi)
\end{align}
where $\ket{v}=U^{\dagger}\ket{\xi}$.

\begin{theorem} For $n$ the number of copies of $\ket{\psi}$ used in the measure and reflect algorithm and for any dimension $d\ge 2$,
\begin{equation}
    {8(n+1)(d-1)\over (n+d+1)(n+d)} \le \opnorm{\mathcal{E}_{\mathrm{MR},\psi} - \mathcal{R}_{\psi}}
    \end{equation}
    with equality obtained for $d=2$.
For all $d\ge 2$, the following asymptotic holds as $n\rightarrow \infty$:
    \begin{equation}
     \opnorm{\mathcal{E}_{\mathrm{MR},\psi} - \mathcal{R}_{\psi}} \sim 
{4(d+\sqrt{d(d-2)+1}-1)\over n} 
\label{eqn:mrasymp}
\end{equation}
\end{theorem}

\begin{proof}
Use orthonormal basis $\lbrace \ket{\psi},\ket{\psi_{1}},\ldots,\ket{\psi_{d-1}}\rbrace$ for $\mathbb{C}^{d}$ and consider the pure state $\ket{\phi_{p}}_{RS}$ from Lemma \ref{lem:twirledStates}
which, for some $p$, attains the maximum diamond distance due to the covariance under $\langle R_{\psi},\idm\rangle'$ of the measure and reflect channel  (\cref{lem:twirledStates}).
The action of $\mathbf{1}_{R}\otimes \mathcal{R}_{\psi}$ on $\phi_{p}$ is fully defined by
\begin{align}
    \mathcal{R}_{\psi}(\ket{\psi}\bra{\psi}) &= \ket{\psi}\bra{\psi} \nonumber \\
    \mathcal{R}_{\psi}(\ket{\psi_{i}}\bra{\psi}) &= -\ket{\psi_{i}}\bra{\psi}\nonumber \\
    \mathcal{R}_{\psi}(\ket{\psi_{i}}\bra{\psi_{j}}) &= \ket{\psi_{i}}\bra{\psi_{j}}
    \label{eqn:rhrh}
\end{align}
for $i,j = 1,\ldots, d-1$.
To calculate $\mathbf{1}_{R}\otimes \mathcal{E}_{\text{MR},\psi}(\phi_{p})$, we use the expression $P_{n} = {n+d-1\choose d-1}\int d\xi \, \ket{\xi}\bra{\xi}^{\otimes n}$ for the projection on the symmetric subspace of $(\mathbb{C}^{d})^{\otimes n}$,
and obtain for $d\ge 3$
{\allowdisplaybreaks
\begin{align}
    \mathcal{E}_{\text{MR},\psi}(\ket{\psi}\bra{\psi}) &= \left( 1-{4\text{tr}P_{n}\over \text{tr}P_{n+1}}+{4\text{tr}P_{n}\over \text{tr}P_{n+2}} \right)\ket{\psi}\bra{\psi} + {4\text{tr}P_{n}\over (n+2)\text{tr}P_{n+2}}\sum_{j=1}^{d-1}\ket{\psi_{j}}\bra{\psi_{j}}  \\
    \mathcal{E}_{\text{MR},\psi}(\ket{\psi}\bra{\psi_{j}})&= \left(1 - {2\text{tr}P_{n}\over \text{tr}P_{n+1}} - {2\text{tr}P_{n}\over (n+1)\text{tr}P_{n+1}} +{4\text{tr}P_{n}\over (n+2)\text{tr}P_{n+2}}\right) \ket{\psi}\bra{\psi_{j}} \nonumber \\
    &= \left(1  - {2(n+2)\text{tr}P_{n}\over (n+1)\text{tr}P_{n+1}} +{4\text{tr}P_{n}\over (n+2)\text{tr}P_{n+2}}\right) \ket{\psi}\bra{\psi_{j}}  \\ 
    \mathcal{E}_{\text{MR},\psi}(\ket{\psi_{i}}\bra{\psi_{i}})&= \left( 1-{4\text{tr}P_{n}\over (n+1)\text{tr}P_{n+1}} +{8\text{tr}P_{n}\over (n+2)(n+1)\text{tr}P_{n+2}} \right)\ket{\psi_{i}}\bra{\psi_{i}} \nonumber \\
    &{} +{4\text{tr}P_{n}\over (n+2)\text{tr}P_{n+2}}\ket{\psi}\bra{\psi} \\
    &{} + \sum_{i'\neq i}{4\text{tr}P_{n}\over (n+1)(n+2)\text{tr}P_{n+2}}\ket{\psi_{i'}}\bra{\psi_{i'}}\label{eqn:hsdhsd} \\
    &= \left( 1-{4\text{tr}P_{n}\over (n+1)\text{tr}P_{n+1}} +{4\text{tr}P_{n}\over (n+2)(n+1)\text{tr}P_{n+2}} \right)\ket{\psi_{i}}\bra{\psi_{i}}\\
    &{} +{4\text{tr}P_{n}\over (n+2)\text{tr}P_{n+2}}\ket{\psi}\bra{\psi} + \sum_{i'=1}^{d-1}{4\text{tr}P_{n}\over (n+1)(n+2)\text{tr}P_{n+2}}\ket{\psi_{i'}}\bra{\psi_{i'}}  \\
    \mathcal{E}_{\text{MR},\psi}(\ket{\psi_{i}}\bra{\psi_{j}})&=\left( 1-{4\text{tr}P_{n}\over (n+1)\text{tr}P_{n+1}} +{ 4\text{tr}P_{n}\over (n+1)(n+2)\text{tr}P_{n+2}} \right)\ket{\psi_{i}}\bra{\psi_{j}},\nonumber \\
    &{}\text{ for } i\neq j .\label{eqn:gsdgsd}
\end{align}
}
Using (\ref{eqn:rhrh}) and taking $p=1$ in $\Vert (\mathbf{1}_{R}\otimes (\mathcal{E}_{\text{MR},\psi}-\mathcal{R}_{\psi})(\phi_{p}) \Vert_{1}$ gives the stated lower bound by a straightforward computation. For general $p$, the trace norm decomposes into the sum of the trace norm of a diagonal matrix and a $2\times 2$ matrix. The full result, up to $O(n^{-2})$ terms coming from the trace norm of the $2\times 2$ matrix is
\begin{align}
\Vert (\text{Id}_{R}\otimes (\mathcal{E}_{\text{MR},\psi}-\mathcal{R}_{\psi})(\phi_{p}) \Vert_{1} &= {4(d-1)p(n+1)\over(n+d)(n+d+1)} + {4(1-p)\over n+d+1}\nonumber \\
&{}+ {4\sqrt{1+dp(d-2)}\over n} + O(n^{-2}) \nonumber \\
&= {4(d-1)p\over n} + {4(1-p)\over n} + {4\sqrt{1+dp(d-2)}\over n}\nonumber\\
&{} + O(n^{-2})\label{eqn:werwer}
\end{align}
where the second line considers the asymptotic in $n$ of the exact expressions in the first line. For all $d$, the supremum over $p$ is obtained for $p=1$, so
\begin{align}
\lim_{n\rightarrow \infty}n\sup_{p}\Vert (\text{Id}_{R}\otimes (\mathcal{E}_{\text{MR},\psi}-\mathcal{R}_{\psi})(\phi_{p}) \Vert_{1}&= \lim_{n\rightarrow \infty}\sup_{p}n\Vert (\mathbf{1}_{R}\otimes (\mathcal{E}_{\text{MR},\psi}-\mathcal{R}_{\psi})(\phi_{p}) \Vert_{1}\nonumber \\
&= \sup_{p}\lim_{n\rightarrow \infty}n\Vert (\mathbf{1}_{R}\otimes (\mathcal{E}_{\text{MR},\psi}-\mathcal{R}_{\psi})(\phi_{p}) \Vert_{1}\nonumber \\
&= 4(d+\sqrt{1+d(d-2)}-1)
\end{align}
where the second equality holds due to the uniform convergence shown in (\ref{eqn:werwer}).

To show that $\opnorm{\mathcal{E}_{\text{MR},\psi} - \mathcal{R}_{\psi}} \sim 
{8\over n}$ for $d=2$, and therefore satisfies the asymptotic in (\ref{eqn:mrasymp}), one notes that (\ref{eqn:hsdhsd}) and (\ref{eqn:gsdgsd}) do not occur for $d=2$. One finds that 
\begin{equation}
\Vert (\text{Id}_{R}\otimes (\mathcal{E}_{\text{MR},\psi}-\mathcal{R}_{\psi})(\phi_{p}) \Vert_{1} = {8(n+1)\over (n+2)(n+3)}
\end{equation}
regardless of $p$, so $\lim_{n\rightarrow \infty}n\sup_{p}\Vert (\mathbf{1}_{R}\otimes (\mathcal{E}_{\text{MR},\psi}-\mathcal{R}_{\psi})(\phi_{p}) \Vert_{1} = 8$. \end{proof}

Note that given a program state $\bigotimes_{j=1}^{K}\ket{\psi_{j}}^{\otimes n}$, one can define a measure-and-reflect version of subspace reflection.

\section{Proof of \cref{lem:compchan}}\label{sec:compproof}

The preparation channel $\Prep(\psi)_{S\rightarrow P}$ was defined in (\ref{eqn:prepchannelstine}).
With $G=\langle \idm,R_{\psi}\rangle'$ defining the $U(d-1)$ subgroup that commutes with $\psi$, one verifies the $G$-covariance of $\widehat{\mathcal{E}}_{\psi,R_{\theta}} - \Prep(\psi)_{S\rightarrow P}$ and uses the twirling argument from the proof of \cref{lem:twirledStates} to conclude that the optimization over states of $SR$ that appears in diamond distance $\opnorm{\widehat{\mathcal{E}}_{\psi, \apxRef_{\theta}} - \Prep(\psi)_{S\rightarrow P}}$ can be restricted, without loss of generality, to the one-parameter set of states $\phi_{p}$, $p\in [0,1]$ in \cref{lem:twirledStates}
\begin{equation}
 \Vert \widehat{\mathcal{E}}_{\psi, \apxRef_{\theta}} - \Prep(\psi)_{S\rightarrow P}\Vert_{\diamond}=\max_{p\in [0,1]}\Vert \text{Id}_{R}\otimes \left( \widehat{\mathcal{E}}_{\psi, \apxRef_{\theta}} - \Prep(\psi)_{S\rightarrow P}\right) (\phi_{p})_{RS}\Vert_{1}.
 \end{equation}
One finds that
\begin{equation}
    \text{Id}_{R}\otimes \Prep(\psi)_{S\rightarrow P}(\phi_{p}) = \left( p\ketbra{0}{0}_{R} + {1-p\over d-1}(\idm_{R} - \ketbra{0}{0}_{R})\right)\otimes \psi^{\otimes n}_{P}.
\end{equation}
Expanding the expression $\text{Id}_{R}\otimes \widehat{\mathcal{E}}_{\psi,R_{\theta}}(\phi_{p})$, one obtains
\begin{align}
&{}\text{tr}_{S}\left[ \left( \idm_{R}\otimes \left( \idm_{SP} + {e^{i\theta}-1\over n+1}\sum_{\ell=0}^{n}C^{\ell} \right) \right)  ((\phi_{p})_{RS} \otimes \psi^{\otimes n}_{P}) \left( \idm_{R}\otimes \left( \idm_{SP} + {e^{-i\theta}-1\over n+1}\sum_{\ell=0}^{n}C^{-\ell} \right) \right) \right]\nonumber \\
 &= \text{Id}_{R}\otimes \Prep(\psi)_{S\rightarrow P}(\phi_{p}) \nonumber \\
 &+ \left( {e^{i\theta}-1\over n+1} \ptrace*{S}{ \left((n+1)\sqrt{p}\ket{0}_{R}\ket{\psi}_{S}\ket{\psi}^{\otimes n}_{P} + \sqrt{1-p\over d-1}\sum_{i=1}^{d-1}\sum_{\ell=0}^{n}\ket{i}_{R}C^{\ell}\ket{\psi_{i}}_{S}\ket{\psi}^{\otimes n}_{P}\right) \bra{\phi_{p}}_{RS}\bra{\psi}^{\otimes n}_{P}} \right. \nonumber \\
 &{} \qquad \left. + \text{h.c.} \vphantom{{e^{i\theta}-1\over n+1}}\right) \nonumber \\
 &+ {2-2\cos\theta \over (n+1)^{2}} \ptrace*{S}{\left( (n+1)\sqrt{p}\ket{0}_{R}\ket{\psi}_{S}\ket{\psi}^{\otimes n}_{P} + \sqrt{1-p\over d-1}\sum_{i=1}^{d-1}\sum_{\ell=0}^{n}\ket{i}_{R}C^{\ell}\ket{\psi_{i}}_{S}\ket{\psi}^{\otimes  n}_{P} \right) \left( \vphantom{\sum_{\ell=0}^{n}} \cdots \vphantom{\sum_{\ell=0}^{n}}\right)^{\dagger} }
\end{align}
where $(\cdots)^{\dagger}$ refers to the conjugate vector of the immediately preceding vector.

In applying the cyclic permutations to the $n+1$ qudit register $SP$, we decompose the register $P=P_{1}\cdots P_{n}$ and use the notation $P\setminus P_{\ell} = P_{1}\cdots P_{\ell-1}P_{\ell+1}\cdots P_{n}$.
We now explicitly write the partial trace $\text{tr}_{S}$ of $\text{Id}_{R}\otimes( \widehat{\mathcal{E}}_{\psi,R_{\theta}} 
 - \Prep_{S\rightarrow P} ) (\phi_{p})$ explicitly so that the support of the state on $RP$ becomes clear. 
\begin{align}
    &{} \left( {e^{i\theta}-1\over n+1} \left( \vphantom{\sqrt{p(1-p)\over d-1}} (n+1)p\ket{0}\bra{0}_{R}\otimes \psi^{\otimes n}_{P}  + {1-p\over d-1}(\idm_{R}-\ket{0}\bra{0}_{R})\otimes \psi^{\otimes n}_{P} \right. \right. \nonumber \\
    & \left.\left. +\sqrt{p(1-p)\over d-1}\sum_{i=1}^{d-1}\ket{i}\bra{0}_{R} \otimes \sum_{\ell=1}^{n}\ket{\psi_{i}}_{P_{\ell}}\ket{\psi}^{\otimes n-1}_{P\setminus P_{\ell}}\bra{\psi}^{\otimes n}_{P}\right)  +h.c. \right)  \nonumber \\
    &+ {2-2\cos\theta \over (n+1)^{2}}\left[ \vphantom{\sqrt{p(1-p)\over d-1}} (n+1)^{2}p\ket{0}\bra{0}_{R}\otimes \psi^{\otimes n}  \right. \nonumber \\
    & \left. +\left( (n+1)\sqrt{p(1-p)\over d-1}\sum_{i=1}^{d-1}\ket{i}\bra{0}_{R}\otimes \sum_{\ell=1}^{n-1}\ket{\psi_{i}}_{P_{\ell}}\ket{\psi}^{\otimes n-1}_{P\setminus P_{\ell}}\bra{\psi}^{\otimes n}_{P} + h.c. \right) \right. \nonumber \\
    &{} \left. +{1-p\over d-1} \sum_{i,j=1}^{d-1}\ket{i}\bra{j}_{R}\otimes \left(\delta_{i,j}\psi^{\otimes n}_{P} + \sum_{\ell=1}^{n}\sum_{\ell'=1}^{n}\ket{\psi_{i}}_{P_{\ell}}\ket{\psi}^{\otimes n-1}_{P\setminus P_{\ell}}\bra{\psi_{j}}_{P_{\ell'}}\bra{\psi}^{\otimes n-1}_{P\setminus P_{\ell'}}\right)\right].
\end{align}
One notices that the terms with the projection $\ket{0}\bra{0}$ in the $R$ register  (which appear only in the first term of the first line, the Hermitian conjugate of the the first term of the first line, and in the third line) sum to zero.  Collecting also the contribution from ${1-p\over d-1}(\idm_{R} - \ket{0}\bra{0}_{R})\otimes \psi^{n}_{P}$, we obtain
\begin{align}
    &{} {1-p\over d-1}\left( {2\cos\theta -2 \over n+1} + {2-2\cos\theta \over (n+1)^{2}}\right) (\idm_{R} - \ket{0}\bra{0}_{R})\otimes \psi^{n}_{P}  \nonumber \\
    & + \paren*{{e^{i\theta}-1\over n+1}\sqrt{np(1-p)}\ket{\xi}\bra{\sigma}_{RP} + h.c.} \nonumber \\
    &+ {2-2\cos\theta \over (n+1)^{2}} \left( (n+1)\sqrt{np(1-p)}\ket{\xi}\bra{\sigma}_{RP} +h.c. \right) \nonumber \\
    &+ {(2-2\cos\theta)n(1-p) \over (n+1)^{2}}\ket{\xi}\bra{\xi}_{RP}
    \label{eqn:ttyy}
\end{align}
where we introduced the orthonormal states
\begin{align}
    \ket{\sigma}_{RP}&:=\ket{0}_{R}\ket{\psi}^{n}_{P} \nonumber \\
    \ket{\xi}_{RP}&:= {1\over \sqrt{n(d-1)}}\sum_{i=1}^{d-1}\ket{i}_{R}\otimes \sum_{\ell=1}^{n}\ket{\psi_{i}}_{P_{\ell}}\ket{\psi}^{\otimes n-1}_{P\setminus P_{\ell}}.
\end{align}

The first line of (\ref{eqn:ttyy}) is proportional to a projection that is orthogonal to the last three lines of the expression. Therefore, the trace norm of (\ref{eqn:ttyy}) becomes
\begin{align}\Vert & \text{Id}_{R}\otimes (\widehat{\mathcal{E}}_{\psi, \apxRef_{\theta}} - \Prep(\psi)_{S\rightarrow P})(\phi_{p})_{RS}\Vert_{1} = {n(1-p)\vert w(\theta)\vert^{2}\over (n+1)^{2}} \nonumber \\
    &+ \norm*{{\sqrt{np(1-p)}\over n+1}\left( w(\theta)\ket{\xi}\bra{\sigma} + h.c. \right) + {n(1-p)\vert w(\theta)\vert^{2}\over (n+1)^{2}}\ket{\xi}\bra{\xi}}_{1}\nonumber \\
    &= {n(1-p)\vert w(\theta)\vert^{2}\over (n+1)^{2}} \nonumber \\
&+ \sqrt{ {\vert w(\theta)\vert^{4}n^{2}(1-p)^{2}\over (n+1)^{4}}+{4np(1-p)(2-2\cos\theta)\over (n+1)^{2}}}.
    \label{eqn:interminterm}
\end{align}
where $w(\theta):=1-2\cos\theta + e^{i\theta}$, $\vert w(\theta)\vert^{2}=4\sin^{2}\left( {\theta\over 2} \right)$.
For large $n$, (\ref{eqn:interminterm}) is asymptotically equal to 
\begin{equation}
    2\vert w(\theta) \vert \sqrt{{p(1-p)\over n}}
\end{equation}
which is maximized by $p=1/2$, yielding the asymptotic in (\ref{eqn:compldiam}).

\section{Derivation of diamond distance bounds using Kraus operators}

In this appendix we first provide an alternative derivation of the Kraus operators of the approximate reflection channel and its complementary channel. Then we use the Kraus operators to upper bound the diamond distance between the approximate reflections channel and the exact reflection channel and lower bound the diamond distance between the complementary channel and the state preparation channel for the program state.

\subsection{Alternative Derivation of Kraus Operators}

\subsubsection{Approximate Reflection Channel\label{sec:directkraus}}

 Let us consider the following basis for the entire program register $P\cong (\mathbb{C}^{d})^{\otimes n}$
\begin{align}
    \ket{\Psi_{\vec{j}}}_{P} = \bigotimes_{m=1}^{n} \ket{\psi_{j_m}}_{P_{m}}
\end{align}
where $\vec{j}=(j_1,\dots,j_n)$ with $j_m \in \{0,\dots,d-1\}$ and $\{\ket{\psi_j}\}_{j=0}^{d-1}$ is an orthonormal basis for the Hilbert space of the system such that $\ket{\psi_0} = \ket{\psi}$.
Given the Stinespring dilation of the programmable reflection channel \begin{align}
\mathcal{E}(\rho) =\text{tr}_{P} \left( V_{SP} \left(\rho_S \otimes \psi_{P}^{\otimes n}\right) V^\dagger_{SP} \right),
\end{align} the Kraus operators can be expressed as
\begin{align}
    K_{\vec{j}} = \bra{\Psi_{\vec{j}}}_P V_{SP} \ket{\psi}^{\otimes  n}_P.
\end{align}
For $V=\sum_{\ell=0}^n c_\ell C^\ell$ it can be seen by inspection that $K_{\vec{j}}=0$ whenever more than one of the indices $j_i$ are nonzero. The only nonzero Kraus operators are then 
\begin{align}
    K_0 &:= K_{(0,\dots,0)} = c_0 \idm + \sum_{\ell=1}^n c_\ell \ketbra{\psi}{\psi} = c_0 \idm -(c_0-\tilde{c}_0 )\ketbra{\psi}{\psi} \\
    K_{\ell,j} &:= K_{(0,\dots,j,\dots,0)} = c_\ell \ketbra{\psi}{\psi_j} \qquad j > 0 \, \text{in the }\ell\text{-th position}.
\end{align}
This makes a total of $1+n(d-1)$ Kraus operators. It turns out these can be further grouped together by noticing:
\begin{align}
    \sum_{\ell=1}^n K_{\ell,j} (\cdot) K_{\ell,j}^\dagger 
    &=\sum_{\ell=1}^n |c_\ell|^2 \bra{\psi_j}(\cdot)\ket{\psi_j} \ketbra{\psi}{\psi} \\
    &= \sum_j K_j (\cdot) K_j^\dagger,
    \intertext{where we defined, for $j>$0,}
    K_j &\coloneqq \sqrt{\sum_{\ell=1}^n |c_\ell|^2} \, \ketbra{\psi}{\psi_j}.
\end{align}
This way we end up with $d$ Kraus operators $\{ K_j\}_{j=0}^{d-1}$.
These agree with \cref{eqn:krausdirect} from the Choi matrix.

\subsubsection{Complementary Channel\label{sec:cckraus}}

To obtain the Kraus operators of the complementary channel we use the orthonormal basis for the system register $\{\ket{\psi_j}\}_{j=0}^{d-1} $ with $\ket{\psi_0} = \ket{\psi}$  and write
\begin{align}
    K_j = \bra{\psi_j}_S V \ket{\psi}_P^{\otimes  n}. 
\end{align}
Since $V=\sum_{\ell=0}^n c_\ell C^l$
\begin{align}
    K_j = \sum_{\ell=0}^{n} c_\ell \bra{\psi_j}_S C^\ell \ket{\psi}_P^{\otimes n} 
\end{align}
By inspections we see that
\begin{align}
    K_0 &= c_0 \ket{\psi}^{\otimes n}_{P} \bra{\psi}_{S} + \sum_{\ell=1}^{n} c_\ell \ket{\psi}^{\otimes  n-1}_{P\setminus P_{\ell}} \otimes I_{S\rightarrow P_{\ell}}\\
    K_{j} &= c_0 \ket{\psi}^{\otimes  n}_{P} \bra{\psi_j}_{S}, & \text{for } j &> 0.
\end{align}
In the above, $I_{S\rightarrow P_{\ell}} = \sum_{j=0}^{d-1} \ket{\psi_j}_{P_l}\bra{\psi_j}_S$ is the identity isometry from the system register to the $l$-th program register on the left. One has $I^{\dagger}_{S\rightarrow P} = I_{P\rightarrow S}$.

We now verify the normalization $\sum_{j=0}^{d-1} K_j^\dagger K_j = \idm_{S}$:
\begin{align}
    \sum_{j=0}^{d-1} K_j^\dagger K_j &= \sum_{j=1}^{d-1} K_j^\dagger K_j + K_0^\dagger K_0 \\
    \sum_{j=1}^{d-1} K_j^\dagger K_j &= |c_0|^2 \sum_{j=1}^{d-1} \ketbra{\psi_j}{\psi_j}_{S} \label{eqn:iuiuiu}\\
    K_0^\dagger K_0 &= \left( \bar{c}_0 \ket{\psi}_S\bra{\psi}^{\otimes n}_{P} + \sum_{\ell'=1}^n \bar{c}_{\ell'} \bra{\psi}^{\otimes  n-1}_{P\setminus P_{\ell'}}\otimes I_{P_{\ell'}\rightarrow S}  \right) \left( c_0 \ket{\psi}^{\otimes  n}_{P}\bra{\psi}_{S} + \sum_{\ell=1}^n c_{\ell} \ket{\psi}^{\otimes  n-1}_{P\setminus P_{\ell}}\otimes I_{S\rightarrow P_{\ell}}\right) \\
    &= |c_0|^2 \ketbra{\psi}{\psi}_{S} + \left(\bar{c}_0\sum_{\ell=1}^n c_\ell + c_0\sum_{\ell=1}^n \bar{c}_\ell  \right) \ketbra{\psi}{\psi}_{S} \\
    &\qquad \qquad+ \sum_{\ell,\ell'=1}^n \bar{c}_{\ell'} c_l \left( \bra{\psi}^{\otimes  n-1}_{P\setminus P_{\ell'}}\otimes I_{P_{\ell'}\rightarrow S}  \right) \left( \ket{\psi}^{\otimes  n-1}_{P\setminus P_{\ell}}\otimes I_{S\rightarrow P_{\ell}} \right)\\
    &= |c_0|^2 \ketbra{\psi}{\psi}_{S} + \left(\bar{c}_0\sum_{\ell=1}^n c_\ell + c_0\sum_{\ell=1}^n \bar{c}_\ell + \left| \sum_{\ell=1}^n c_\ell\right|^2 \right) \ketbra{\psi}{\psi}_{S} + \sum_{\ell=1}^n |c_\ell|^2 \sum_{j=1}^{d-1} \ketbra{\psi_j}{\psi_j}_{S} \nonumber \\
    &= \vert \tilde{c}_{0}\vert^{2} \ket{\psi}\bra{\psi}_{S} + \sum_{\ell=1}^n |c_\ell|^2 \sum_{j=1}^{d-1} \ketbra{\psi_j}{\psi_j}_{S}
    \label{eqn:iuiu}
\end{align}
Adding together (\ref{eqn:iuiuiu}) and (\ref{eqn:iuiu}) gives
\begin{align}
    \sum_{j=0}^{d-1} K_j^\dagger K_j &= 
    \ketbra{\psi}{\psi} + \sum_{\ell=0}^n |c_\ell|^2 \sum_{j=1}^{d-1} \ketbra{\psi_j}{\psi_j}_{S}\\
    &= \idm_{S}.
\end{align}

\subsection{Bounds on Diamond Distance\label{sec:bdd}}

\subsubsection{Approximate Reflection Channel\label{sec:chanupper}}

We define $\mathcal{K}_j(\cdot) = K_j (\cdot) K_j^\dagger$.
\begin{align}
    \opnorm{\rotChannel - \apxRefCh} &= \opnorm{\rotChannel - \sum_{j=0}^{d-1} \mathcal{K}_j} \\
    &\le \opnorm{\rotChannel-\mathcal{K}_0} + \opnorm{\sum_{j=1}^{d-1}\mathcal{K}_j}
\end{align}
For the first term
\begin{align}
    \opnorm{\rotChannel-\mathcal{K}_0} &= \max_\rho \norm{(R_\psi\otimes \idm)\rho(R_\psi^\dagger\otimes \idm) - (K_0\otimes \idm)\rho(K_0^\dagger\otimes \idm)}_1 \\
    &= \max_\rho \norm{(R_\psi\otimes \idm)\rho(R_\psi^\dagger\otimes \idm) - (K_0\otimes \idm)\rho (R_\psi^\dagger\otimes \idm) + (K_0\otimes \idm)\rho (R_\psi^\dagger\otimes \idm) - (K_0\otimes \idm)\rho(K_0^\dagger\otimes \idm)}_1 \\
    &\le \max_{\rho} \norm{((R_\psi - K_0)\otimes \idm)\rho (R_\psi^\dagger\otimes \idm)}_1 + \max_{\rho'} \norm{(K_0\otimes \idm)\rho' ((R_\psi^\dagger- K_0^\dagger)\otimes \idm)}_1, \\
    \intertext{where we used triangle inequality and the fact that maximization of each term individually can only increase the sum of the norms.
    The maximum is achieved by a pure state~\cite[Theorem~3.39]{WatrousBook}, in}
    & = \max_{\ket{\psi}} \norm{((R_\psi - K_0)\otimes \idm)\ket{\psi}}\norm{(R_\psi \otimes \idm) \ket\psi} + \max_{\ket{\psi'}} \norm{(K_0\otimes \idm)\ket{\psi'}} \norm{((R_\psi - K_0)\otimes \idm) \ket{\psi'}}. \\
    \intertext{
    By maximizing $\ket\psi$ and $\ket{\psi'}$ separately between the product of norms,
we can upper bound the individual norms can be maximized by the individual spectral norms (Schatten-$\infty$ norms) of the operators, so}
    &\le \norm{R_\psi - K_0} \norm{R^\dagger} + \norm{K_0} \norm{R^\dagger-K_0^\dagger} \\
    &\le 2 \norm{R_\psi-K_0}. \label{eq:1}
\end{align}
Finally, for the last inequality we used the fact the $\norm{K_0} \le 1$.
\begin{align}
    \norm{R_\psi-K_0} &= \norm{\idm-2\ketbra{\psi}{\psi} -(c_0 \idm -(c_0-\tilde{c}_0)\ketbra{\psi}{\psi})} \\
    &= \norm{(1-c_0)(\idm-\ketbra{\psi}{\psi}) - (1+\tilde{c}_0)\ketbra{\psi}{\psi}} \\
    &= \max(|1-c_0|,|1+\tilde{c}_0|)
\end{align}
Plugging this back into Eq.~\eqref{eq:1} we obtain 
\begin{align}
    \opnorm{\rotChannel - \apxRefCh} \le 4 \max(|1-c_0|,|1+\tilde{c}_0|) \, .
\end{align}
Recall that 
\begin{align}
    c_0 &= 1 + \frac{e^{i\theta}-1}{n+1} \\
    \tilde{c}_0 &= \frac{n e^{i\theta} +1}{n+1}
\end{align}
Then
\begin{align}
    |1-c_0| &= \left| \frac{1-e^{i\theta}}{n+1} \right| = \frac{\sqrt{(1+\cos\theta)^2+\sin^2\theta}}{n+1} \\
    |1+\tilde{c}_0| &= \frac{\sqrt{(2+n(1+\cos\theta))^2+n^2 \sin^2\theta}}{n+1}
\end{align}
With straightforward algebra it can be verified that $|1+\tilde{c}_0|>|1-c_0|$ for all $n>0$ and $\theta \in [0,2\pi)$.
Thus
\begin{align}
    \opnorm{\rotChannel - \apxRefCh} \le 4 \frac{\sqrt{(2+n(1+\cos\theta))^2+n^2 \sin^2\theta}}{n+1} \, .
\end{align}
For all $n$ this is minimized by $\theta=\pi$ which corresponds to HL algorithm (can be seen setting the derivative to 0). The upper bound in this case is
\begin{align}
    \opnorm{\rotChannel - \apxRefCh} \le \frac{8}{n+1} \, .
\end{align}
compared to the exact expression we obtained $\opnorm{\rotChannel - \apxRefCh}= 8n/(n+1)^2$. These agree to leading order and only differ at $O(1/n^2)$.

If we plug in the $\theta^*$ corresponding to the optimal algorithm we get a looser bound
\begin{align}
    \opnorm{\rotChannel - \apxRefCh} \le 4\sqrt{\frac{n^2 \sin^2\theta +(n+n\cos\theta+2)^2}{(n+1)^2}} = \frac{16}{n}-\frac{32}{n^2} + O(\frac{1}{n^3}) \, .
\end{align}
compared to the exact expression we obtained $\opnorm{\rotChannel - \apxRefCh}= 8(n+2)/(8+4n+n^2)$. These don't agree at the leading order. In fact this upper bound is worse than the upper bound we got for HL algorithm which was $8/n$ to leading order. 

\subsubsection{Complementary Channel\label{sec:compchanlower}}
Here we want to lower bound (not upper bound) the diamond distance between the program state preparation channel and the complementary channel of the approximate reflection channel. 
The Kraus operators of the state preparation channel are given by:
\begin{align}
    M_j &= \ket{\psi}^{\otimes n} \bra{\psi_j}
    \label{eqn:uyuyuy}
\end{align}
Let $\mathcal{M}_j(\cdot) = M_j(\cdot)M_j^\dagger$ for any linear operator $M_{j}$. 
We want to upper bound
\begin{align}
    \opnorm{\Prep_{S\rightarrow P}(\psi) -  \widehat{\mathcal{E}}_{\psi, \apxRef_{\theta}}} &= \opnorm{\sum_{j=0}^{d-1} (\mathcal{M}_j-\mathcal{K}_j)} \\
    &\ge  \opnorm{ \mathcal{M}_0-\mathcal{K}_0}- \opnorm{\sum_{j=1}^{d-1} (\mathcal{M}_j-\mathcal{K}_j)}  \label{eq:2}
\end{align}
First we show that the second term in Eq.~\eqref{eq:2} is $O(1/n)$. Using the Kraus operators in (\ref{eqn:compkraus}) and (\ref{eqn:uyuyuy})
\begin{align}
    \opnorm{\sum_{j=1}^{d-1} (\mathcal{M}_j-\mathcal{K}_j)} &= |1-c_0| \opnorm{\sum_{j=1}^{d-1}\mathcal{M}_j} \\
    &\le |1-c_0| \opnorm{\Prep(\psi)_{S\rightarrow P}}\\
    &= |1-c_0|\\
    &= \frac{|1-e^{i\theta}|}{n+1} \\
    &= O\left(\frac{1}{n}\right)
\end{align}
where the inequality uses the fact that $$\sum_{j=1}^{d-1}\mathcal{M}_{j}(\rho_{S}) = \text{prep}(\psi)_{S\rightarrow P}\left( (\idm_{S}-\psi_{S})\rho_{S}(\idm_{S}-\psi_{S})\right).$$

Next we want to argue the first term in Eq.~\eqref{eq:2} is $\Omega(1/\sqrt{n})$. It will be convenient to use the Kraus operator $K_0$ from (\ref{eqn:compkraus}) 
\begin{align}
    \opnorm{ \mathcal{M}_0-\mathcal{K}_0} &= \max_{\rho_{RS}} \norm{(\idm_{R}\otimes M_0)\rho_{RS}(\idm_{R}\otimes M_0^\dagger)-(\idm_{R}\otimes K_0)\rho_{RS}(\idm_{R}\otimes K_0^\dagger) }_1 
\end{align}
where $I$ is the identity on an auxiliary register $R$.
Using the fact that $\vert \tilde{c}_{0}\vert^{2}=1$, and restricting the maximization to states on $S$, we have 
\begin{align}
\label{eq:3}
    \opnorm{ \mathcal{M}_0-\mathcal{K}_0} &\ge \max_\rho \| \sum_{j>0}( \tilde{c}_0\ket{\psi}^{ \otimes n}_{P}\bra{\psi}_{S})\rho( \sum_{\ell=1}^n \bar{c}_l \ket{\psi_{j}}_{S}\bra{\psi}^{\otimes  n-1}_{P\setminus P_{\ell}} \bra{\psi_{j}}_{P_{\ell}}  + \text{h.c.} ) \\ \nonumber
    &+ \sum_{j,j'>0}(\sum_{\ell=1}^n c_l  \ket{\psi}^{\otimes  n-1}_{P\setminus P_{\ell}}\ket{\psi_{j}}_{P_{\ell}}\bra{\psi_{j}}_{S})\rho(\sum_{\ell'=1}^n \bar{c}_{\ell'} \ket{\psi_{j'}}_{S}\bra{\psi}^{\otimes  n-1}_{P\setminus P_{\ell'}}\bra{\psi_{j'}}_{P_{\ell'}})\|_1
\end{align}
Let us choose \begin{align}
    \rho=\frac{\ket{\psi}+\ket{\psi_1}}{\sqrt{2}}\frac{\bra{\psi}+\bra{\psi_1}}{\sqrt{2}}. 
\end{align}
Plugging into Eq.~\eqref{eq:3} and using the normalized vector $\ket{v_{1}}$ from (\ref{eqn:vvvec}), one obtains
\begin{align}
    \opnorm{\mathcal{M}_0-\mathcal{K}_0} &\ge \norm{\frac{\sqrt{1-|c_0|^2}}{2}(\tilde{c}_0\ket{\psi}^{ \otimes n}\!\bra{v_1}+\bar{\tilde{c}}_0\ket{v_1}\!\bra{\psi}^{\otimes  n}) + \frac{(1-|c_0|^2)}{2}\ketbra{v_1}{v_1}}_1 \\
    &= \frac{2 \sqrt{n^3+3n^2+n}}{(1+n)^2} \\
    &= \frac{2}{\sqrt{n}} - O\left(\frac{1}{n^{3/2}}\right)
\end{align}
where in the second line we have specialized to the programmable processor defined by $\theta = \pi$.
Thus, we have proven that 
\begin{align}
    \opnorm{\Prep(\psi)_{S\rightarrow P} - \widehat{\mathcal{E}}_{\psi, \apxRef_{\pi}}} \geq \frac{2}{\sqrt{n}} - O\left(\frac{1}{n} \right)
\end{align}
at $\theta = \pi$ and this agrees with \cref{lem:compchan}.

\end{document}